\documentclass[a4paper, 12pt]{article}

\usepackage{times}
\usepackage{amsmath, amsthm}
\usepackage{accents}
\usepackage{amsfonts}
\usepackage{amssymb}
\usepackage{graphicx}
\usepackage{subcaption}
\usepackage[round]{natbib}

\newtheorem{theorem}{Theorem}
\newcommand{\utilde}[1]{\underaccent{\tilde}{#1}}
\newcommand{\btheta}{\boldsymbol{\theta}}
\newcommand{\htheta}{\widehat{\boldsymbol{\theta}}_n}

\usepackage[margin=1.2in]{geometry}

\newenvironment{sciabstract}{%
\begin{quote} \bf}
{\end{quote}}

\graphicspath{ {./figures/} }

\title{Constructing Prediction Intervals Using the Likelihood Ratio Statistic} 

\author
{Qinglong Tian, Daniel J. Nordman, William Q. Meeker\\
\\
\normalsize{Department of Statistics, Iowa State University}\\
\\
}

\date{\today}

\begin{document} 


\baselineskip24pt


\maketitle 


\begin{sciabstract}

Statistical prediction plays an important role in many decision processes such as university budgeting (depending on the number of students who will enroll), capital budgeting (depending on the remaining lifetime of a fleet of systems), the needed amount of cash reserves for warranty expenses (depending on the number of warranty returns), and whether a product recall is needed (depending on the number of potentially life-threatening product failures). In statistical inference, likelihood ratios have a long history of use for decision making relating to model parameters (e.g., in evidence-based medicine and forensics). We propose a general prediction method, based on a likelihood ratio (LR) involving both the data and a future random variable. This general approach provides a way to identify prediction interval methods that have excellent statistical properties.  For example, if a prediction method can be based on a pivotal quantity, our LR-based method will often identify it. For applications where a pivotal quantity does not exist, the LR-based method provides a procedure with good coverage properties for both continuous or discrete-data prediction applications.

\end{sciabstract}


\section{Introduction}\label{sec:introduction}

\subsection{Background}\label{subsec:background}

Prediction is a fundamental part of statistical inference.
Prediction intervals are important for assessing the uncertainty of future random variables and have applications in business, engineering, science, and other fields.
For example, manufacturers require prediction intervals for the number of warranty claims to assure that there are sufficient cash reserves and spare parts to make repairs; engineers use historical data to compute prediction intervals for the remaining lifetime of systems.

Suppose that the available data are denoted by $\boldsymbol{X}_n$, and that we want to predict a random variable denoted by $Y$ (also known as the predictand).
We use a parametric distribution to model the data and the predictand.
Specifically, we consider the case where $\boldsymbol{X}_n=\{X_1,\dots,X_n\}$ corresponds to a sample of $n$ independent and identically distributed random variables with common density/mass function $f(\cdot;\btheta)$.
The density $f(\cdot;\btheta)$ depends on a vector $\btheta$ of unknown parameters.
The predictand $Y$ is a scalar random variable with conditional density $g(\cdot|\boldsymbol{x}_n;\btheta)$, where $\boldsymbol{X}_n=\boldsymbol{x}_n$ is the observed sample.
If $Y$ is independent of $\boldsymbol{X}_n$, then $g(\cdot|\boldsymbol{x}_n;\btheta)=g(\cdot;\btheta)$; further if $Y$ has the same distribution as the data, then $g(\cdot;\btheta)=f(\cdot;\btheta)$.
The goal is to obtain information about the unknown parameters $\btheta$ from the data $\boldsymbol{X}_n$ to construct a prediction interval for the predictand $Y$.

We use $\text{PI}_{1-\alpha}(\boldsymbol{X}_n)$ to denote a prediction interval for $Y$ with a nominal confidence level of $1-\alpha$.
Letting $\Pr_{\btheta}(\cdot|\boldsymbol{X}_n)$ be conditional probability given $\boldsymbol{X}_n$, the conditional coverage probability of $\text{PI}_{1-\alpha}(\boldsymbol{X}_n)$ is
\[
\text{CP}[\text{PI}_{1-\alpha}(\boldsymbol{X}_n)|\boldsymbol{X}_n]=\Pr{}_{\!\btheta}\left[Y\in\text{PI}_{1-\alpha}(\boldsymbol{X}_n)|\boldsymbol{X}_n\right].
\]
We can obtain the unconditional coverage probability by taking the expectation of the conditional coverage probability $\text{CP}[\text{PI}_{1-\alpha}(\boldsymbol{X}_n)|\boldsymbol{X}_n]$ with respect to $\boldsymbol{X}_n$,
\[
\text{CP}[\text{PI}_{1-\alpha}(\boldsymbol{X}_n)]=\Pr{}_{\!\btheta}[Y\in\text{PI}_{1-\alpha}(\boldsymbol{X}_n)]=\text{E}_{\btheta}\left\{\text{CP}[\text{PI}_{1-\alpha}(\boldsymbol{X}_n)|\boldsymbol{X}_n]\right\}.
\]
Unlike the conditional coverage probability, which is a random variable, the unconditional coverage probability is a fixed property of a prediction interval procedure.
Hence, the unconditional coverage probability is used to evaluate a prediction interval method and the term ``coverage probability'' is used to denote the unconditional coverage probability unless stated otherwise.
If $\text{CP}[\text{PI}_{1-\alpha}(\boldsymbol{X}_n)]=1-\alpha$, we say the prediction method is exact; if $\text{CP}[\text{PI}_{1-\alpha}(\boldsymbol{X}_n)]\to1-\alpha$ as $n\to\infty$, we say the prediction method is asymptotically correct.

\subsection{Related Literature}\label{subsec:literature}

Some prediction interval methods are based on a pivotal or an approximate pivotal quantity.
The main idea is to find a function of $\boldsymbol{X}_n$ and $Y$, say $q(\boldsymbol{X}_n, Y)$, that has a distribution that is free of parameters $\btheta$ (or approximately so for large samples).
Then, the distribution of $q(\boldsymbol{X}_n, Y)$ can be used to construct a $1-\alpha$ prediction region for $Y$ as
\begin{equation*}
	\mathcal{P}_{1-\alpha}(\boldsymbol{x}_n)=\left\{y:q(\boldsymbol{x}_n, y)\leq q_{n, 1-\alpha}\right\},
\end{equation*}
where $\boldsymbol{X}_n=\boldsymbol{x}_n$ denotes the observed value of sample and $q_{n,1-\alpha}$ is the $1-\alpha$ quantile of $q(\boldsymbol{X}_n, Y)$ (i.e., $\Pr{}_{\!\btheta}\left[q(\boldsymbol{X}_n, Y)\leq q_{n,1-\alpha}\right]=1-\alpha$).
If $q(\boldsymbol{x}_n, y)$ is a monotone function of $y$, then
$\mathcal{P}_{1-\alpha}(\boldsymbol{x}_n)$ provides a one-sided prediction
bound; if $-q(\boldsymbol{x}_n, y)$ is a unimodal function of $y$,
then $\mathcal{P}_{1-\alpha}(\boldsymbol{x}_n)$ becomes a (two-sided) prediction interval.
Relevant references of this pivotal method include \citet{cox1975}, \citet{atwood1984},
\cite{beran1990}, \citet{barncox1996}, \citet{nelson2000},
\citet{lawless2005}, and \citet{fonseca2012}.

One implementation of the pivotal method is through a hypothesis test.
\citet{cox1975} and \citet[Page 243]{cox1979theoretical} suggested to construct prediction intervals by inverting a hypothesis test, and gave examples with distributions having simple test statistics.
Suppose the data $\boldsymbol{X}_n$ and the predictand $Y$ have densities $f(\boldsymbol{x}_n;\btheta)$ and $g(y;\btheta^\dagger)$ governed by real-valued $\btheta$ and $\btheta^\dagger$, respectively, and a hypothesis
test can be found for the null hypothesis $\btheta=\btheta^\dagger$.
Let $w_\alpha$ be a critical region for the test $H_0:\btheta=\btheta^\dagger$ with size $\alpha$.
For critical region $\omega_\alpha$ we have the probability statement
\[
\Pr\left[(\boldsymbol{X}_n,Y)\in w_\alpha\right]=\alpha.
\]
Then for $\boldsymbol{X}_n=\boldsymbol{x}_n$, a $1-\alpha$ prediction
region for $Y$ can be defined as
\begin{equation}\label{eq:hypothesis-test-method}
	\mathcal{P}_{1-\alpha}(\boldsymbol{x}_n)=\{y:(\boldsymbol{x}_n, y)\not\in w_\alpha\}.
\end{equation}
Thus, for the critical region defined in (\ref{eq:hypothesis-test-method}), we have that for all $\btheta$
\begin{equation*}
	\Pr{}_{\!\btheta}\left[Y\in\mathcal{P}_{1-\alpha}(\boldsymbol{X}_n)\right]=1-\alpha,
\end{equation*}
so that (\ref{eq:hypothesis-test-method}) defines an exact prediction procedure for $Y$.
In (\ref{eq:hypothesis-test-method}), one could also potentially use a critical region $w_{\alpha}\equiv w_{\alpha,n}$ having size $\alpha$ asymptotically (i.e., $\lim_{n\to \infty }\Pr_{\theta}\left[(\boldsymbol{X}_,Y) \in w_{\alpha}\right]=\alpha$);  then (\ref{eq:hypothesis-test-method}) becomes an asymptotically correct $1-\alpha$ prediction region for $Y$. 

\citet[pp. 245]{cox1979theoretical} illustrated this test-based prediction region (\ref{eq:hypothesis-test-method}) with the normal distribution.
Suppose $\boldsymbol{X}_n$ is an independent random sample from
$\text{Norm}(\mu,\sigma)$ and $Y$ is a further independent random
variable with the same distribution.
By assuming that $\boldsymbol{X}_n\sim\text{N}(\mu_1,\sigma)$
and $Y\sim\text{N}(\mu_2,\sigma)$, a test statistic for the null
hypothesis $H_0:\mu_1=\mu_2$ is
\begin{equation}\label{eq:normal-invert-t-test}
	t=\frac{\bar{X}_n-Y}{s\sqrt{(n+1)/n}}\sim t_{n-1},
\end{equation}
where $s^2=\sum_{i=1}^{n}(X_i-\bar{X}_n)^2/(n-1)$ and $t_{n-1}$ denotes a $t$-random variable with $n-1$ degrees of freedom.
This corresponds to the form of a two-sample $t$-test that is often used for comparison of means.
Then a $1-\alpha$ equal-tailed (i.e., equal probability of being outside either endpoint) prediction interval based on inverting the $t$-test is
\begin{equation*}
	\text{PI}_{1-\alpha}(\boldsymbol{X}_n)=\left[\bar{X}_n-t_{n-1,\alpha/2}s\sqrt{(n+1)/n},\,\bar{X}_n+t_{n-1,\alpha/2}s\sqrt{(n+1)/n}\right],
\end{equation*}
where $t_{n-1,\alpha}$ denotes the $\alpha$ quantile of a $t_{n-1}$ distribution.

As a contrast to the pivotal prediction method, \citet{bjo1990} reviewed an alternative prediction method called the predictive likelihood method.
The main idea of the predictive likelihood method is to obtain an approximate density for $Y$ by eliminating the unknown parameters in the joint likelihood (or density) of the data and the predictand $(\boldsymbol{X}_n,Y)$.
The resulting predictive likelihood then provides a type of distribution for computing a prediction interval for $Y$ given $\boldsymbol{X}_n=\boldsymbol{x}_n$.
For example, a Bayesian predictive distribution for $Y$ involves steps of integrating out the unknown parameters of a posterior distribution based on the joint likelihood of the data and the predictand.

\subsection{Motivations}\label{subsec:motivations}

As reviewed in Section~\ref{subsec:literature}, prediction intervals can be constructed by inverting hypothesis tests for parameters.
However, the construction of such tests often needs to be tailored to each problem, where the determination of an appropriate hypothesis test is an essential step in the construction of such prediction intervals.
For example, in the normal distribution example, we obtain the prediction interval by inverting a $t$-test.
However, in many cases, there is no well-known or clear hypothesis test, making it difficult to implement a test-based method for obtaining prediction intervals.
As a remedy, the purpose of this paper is to propose a general prediction method based on inverting a type of LR test.
The advantage of the LR approach is that this principle applies broadly to different settings where prediction intervals are needed---and particularly for cases where an appropriate test statistic or pivotal quantity is not available or obvious for the need.
In addition, we will demonstrate that this general method has desirable statistical properties.

\subsection{Overview}\label{subsec:overview}

This paper is organized as follows.
Sections~\ref{sec:method}--\ref{sec-choose-full} focus on predictions with continuous data.
Section~\ref{sec:method} describes how to construct a prediction interval by formulating a certain LR statistic.
Section~\ref{sec:pivotal-quantity} discusses situations where the proposed prediction method provides exact coverage, while Section~\ref{asymptotic-results} shows that, more broadly, that the method is generally (under weak  conditions) guaranteed to provide asymptotically correct coverage (i.e., improving coverage properties with increasing sample sizes).
Section~\ref{sec-choose-full} discusses some further details about constructing a suitable LR test.
Section~\ref{sec:discrete-distribution} focuses on applying the proposed method to prediction problems involving discrete data.
Section~\ref{sec-relationship-pred-dist} describes how the proposed LR prediction method compares and differs from predictions based on predictive likelihood methods (mentioned in Section~\ref{subsec:literature}).
Section~\ref{sec:conclusion} concludes by describing potential areas for future research.

\section{A General Method}\label{sec:method}

In Section~\ref{subsec:lrt}, we show how to construct general (not necessarily equal-tailed) two-sided prediction intervals with an LR statistic.
Section~\ref{subsec:one-sided} describes how to construct one-sided prediction intervals by using an LR and how this method can also be applied to calibrate equal-tailed two-sided prediction intervals.
In this section, we assume that both the data $\boldsymbol{X}_n$ and the predictand $Y$ are continuous.
For clarity in the exposition and ease of presentation, we further assume that $Y\sim g(\cdot;\btheta)$ is independent of $\boldsymbol{X}_n\sim f(\cdot;\btheta)$ and has the same distribution/density (i.e., $g(\cdot;\btheta)=f(\cdot;\btheta)$).

\subsection{Prediction Intervals Based on an LR Test}\label{subsec:lrt}

\citet{nelson2000} proposed a prediction interval method for predicting the number of failures
in a future inspection of a sample of units, based on a likelihood
ratio test in combination with the Wilks' theorem.
Although \citet{nelson2000} only considered a specific prediction problem, we extend the principle of LR-based prediction interval statistics to a more general setting.
The approach may also be viewed as a generalization of test-based prediction intervals explained in \citet{cox1979theoretical}, using an LR in the role of the test statistic.

\subsubsection{Reduced and Full Models}

Recalling its traditional use for parametric inference, the LR test provides a general approach for comparing two nested models for data (or parameter configurations) based on the observed data $\boldsymbol{X}_n=\boldsymbol{x}_n$.
The null hypothesis about the parameters corresponds to a reduced model, which is nested within a larger full model (i.e., a parameter subset of the full model).
Let $\mathcal{L}_n(\btheta;\boldsymbol{x}_n)$ be the likelihood function for the full model having a parameter space $\Theta$ and suppose that the reduced model (corresponding to the null hypothesis) has a constrained parameter space $\Theta_0\subset\Theta$.
The LR for testing the null hypothesis $H_0:\btheta\in\Theta_0$ is then
\begin{equation*}
	\Lambda_n = \frac{\sup_{\btheta\in\Theta_0}\mathcal{L}_n(\btheta;\boldsymbol{x}_n)}{\sup_{\btheta\in\Theta}\mathcal{L}_n(\btheta;\boldsymbol{x}_n)},
\end{equation*}
and the log-LR statistic is $-2\log\Lambda_n$.
Generally, the distribution of $\Lambda_n$ or $-2\log\Lambda_{n}$ needs to be determined, either analytically, through approximate large-sample distributional results, or through Monte Carlo simulation.
Then, such a distribution can be used
to determine the critical region for the LR test of the null hypothesis or relatedly a confidence interval/region for parameters.
For example, if the reduced model
is true, and if Wilks' theorem (cf. \citet{wilks1938}) applies (as it does under particular regularity conditions), the asymptotic distribution of the log-LR statistic is given by
$-2\log\Lambda_n\xrightarrow{d}\chi^2_d$ as the sample size $n\to\infty$,
where $\chi_d^2$ denotes a chi-square random variable with $d$ degrees of freedom and where $d$ is the difference in the lengths of $\Theta$ and $\Theta_0$.
The latter large sample chi-square distribution approximation is often used to calibrate the critical region of an LR test.

As we describe next, a log-LR statistic for model parameters can be modified to provide a log-LR statistic for a future random variable $Y$ in a general manner, which in turn can be used to construct prediction intervals for $Y$.
To outline the approach, suppose the available $\boldsymbol{X}_n$ represent an iid sample with common density $f(\cdot;\boldsymbol{\theta})$ and $Y$ denotes a future random variable with the same density $f(\cdot;\boldsymbol{\theta})$ ($Y$ is again independent of $\boldsymbol X_n$ here).
A log-LR statistic for $Y$ can then be broadly framed as a type of parameter $\boldsymbol{\theta}$ comparison involving full vs reduced models for the joint distribution $(\boldsymbol X_n,Y)$.
While $f(\cdot;\boldsymbol{\theta})$ denotes the true marginal density for both the data $\boldsymbol X_n$ and the predictand $Y$ (with parameter space $\boldsymbol{\theta} \in \Theta$), the main idea is to define a hypothesis test regarding an enlarged (and fictional) parameter space $\Theta_E  \equiv \{(\boldsymbol{\theta},\boldsymbol{\theta_y})\} $, where the data $X_n$ have a common density $f(\cdot;\boldsymbol{\theta})$, where the predictand $Y$ has a density $f(\cdot;\boldsymbol{\theta}_y)$, say, and where $\boldsymbol{\theta}$ and $\boldsymbol{\theta}_y$ differ in exactly one pre-selected component when $(\boldsymbol{\theta},\boldsymbol{\theta}_y)\in\Theta_E$.  For example, supposing $\boldsymbol{\theta}=(\theta_1,\ldots,\theta_k) \in \Theta$ consists of $k \geq 1$ components,  then we choose exactly one parameter component, say $\theta_\ell$, from among $\{\theta_1,\ldots,\theta_k\}$ to vary and subsequently define $\boldsymbol{\theta}_y \in \Theta$ to match $\boldsymbol{\theta}\in \Theta$, except for the component $\ell$, which is $\theta_\ell$ for $\boldsymbol{\theta}$ but $\theta_{\ell,y}$ say in $\boldsymbol{\theta}_{y}$.
This framework sets up a comparison of a contrived full model ($\boldsymbol{X}_n\sim f(\cdot;\btheta)$ marginally and $Y\sim f(\cdot;\btheta_y)$ for ($\btheta, \btheta_y)\in\Theta_E$) versus a reduced model ($\btheta_y=\btheta\in\Theta$), where the parameter space of the reduced model is nested within $\Theta_E$ with the constraint $\btheta=\btheta_{y}$.

The purpose of this contrived LR test is not to conduct hypothesis tests of parameters---as we already know that the reduced model is a true model---but to construct a predictive root (i.e., a test statistic containing $\boldsymbol{X}_n$ and $Y$), which will be used to predict $Y$.
In particular, the extra degree of freedom between the parameter spaces of the full model and the reduced model is used to identify the predictand $Y$ when formulating a log-LR statistic for $\btheta=\btheta_y$.
For example, in the case of data from a normal distribution $\boldsymbol{X}_n,Y\sim\text{Norm}(\mu,\sigma)$, we may define a full model as $\boldsymbol{X}_n\sim\text{Norm}(\mu,\sigma)$ and $Y\sim \text{Norm}(\mu_y,\sigma)$ for parameters $\btheta=(\mu,\sigma)$ and $\btheta_y=(\mu_y,\sigma)\in\mathbb{R}\times(0,\infty)$, where the reduced model $\btheta_y=\btheta$ corresponds to the true underlying model $\boldsymbol{X}_n,Y\sim\text{Norm}(\mu,\sigma)$ in prediction.

Let the joint likelihood function for $(\boldsymbol{X}_n,Y)$ be
\[
\mathcal{L}(\btheta,\btheta_y;\boldsymbol{x}_n,y)=f(y;\btheta_y)\prod_{i=1}^{n}f(x_i;\btheta)
\]
under the full model and suppose the maximum likelihood (ML) estimators of $(\btheta,\btheta_y)$ are estimable under both the reduced ($\btheta=\btheta_y$) and the full ($(\btheta,\btheta_y)\in\Theta_E$) models.
Then the joint LR statistic based on $(\boldsymbol{X}_n,Y)$ is
\begin{equation}\label{eq:likelihood-ratio}
	\Lambda_n(\boldsymbol{X}_n,Y)=\frac{\sup_{\btheta=\btheta_y\in\Theta}\mathcal{L}(\btheta,\btheta_y;\boldsymbol{X}_n,Y)}{\sup_{(\btheta,\btheta_y)\in\Theta_E}\mathcal{L}(\btheta,\btheta_y;\boldsymbol{X}_n, Y)}
\end{equation}
for the test of $\btheta=\btheta_y$.
The LR statistic in (\ref{eq:likelihood-ratio}) and its distribution can then be applied to obtain prediction intervals for the future predictand $Y$ based on observed data values $\boldsymbol{X}_n=\boldsymbol{x}_n$.
Note that the construction (\ref{eq:likelihood-ratio}) depends on which parameter from $\btheta$ is selected to vary in defining $\btheta_y$.
Typically, this selected parameter will be a mean-type parameter for purposes of identifying $Y$ in the LR statistic (\ref{eq:likelihood-ratio}); more details about this selection are given in Section~\ref{sec-choose-full}.

\subsubsection{Determining the Distribution of the LR}
\label{sec-determine}
The next step is to determine a critical region as in (\ref{eq:hypothesis-test-method}) so that we can compute the prediction region for $Y$, based on the LR statistic $\Lambda_n(\boldsymbol{X}_n,Y)$ from (\ref{eq:likelihood-ratio}).
This, however, requires the distribution of $\Lambda_n(\boldsymbol{X}_n,Y)$ (or $-2\log\Lambda_{n}(\boldsymbol{X}_n,Y)$).
There are three potential approaches for determining or approximating the distribution of $-2\log\Lambda_{n}(\boldsymbol{X}_n, Y)$.

The first approach is to obtain the distribution of $-2\log\Lambda_{n}(\boldsymbol{X}_n, Y)$ analytically.
For illustration, consider the situation with an iid sample $\boldsymbol{X}_n\sim\text{Norm}(\theta,\sigma)$ where $\sigma$ is known and the future random variable $Y$ is from the same distribution. Here there is one parameter $\theta\equiv\mu$ where the full model is
$\boldsymbol{X}_n\sim\text{Norm}(\mu,\sigma)$ and $Y\sim\text{Norm}(\mu_y,\sigma)$ for $(\mu,\mu_y)\in\mathbb{R}^2$ in the LR construction of (\ref{eq:likelihood-ratio}); the corresponding log-LR statistic for $Y$ based on $\boldsymbol{X}_n$ is then
\begin{equation*}
		-2\log\Lambda_n(\boldsymbol{X}_n, Y)=\frac{n}{n+1}\left(\frac{Y-\bar{X}_n}{\sigma}\right)^2\sim\chi^2_1.
\end{equation*}
Then, a $1-\alpha$ prediction region for $Y$ given $\boldsymbol{X}_n=\boldsymbol{x}_n$ is
\[
\begin{split}
\mathcal{P}_{1-\alpha}(\boldsymbol{x}_n)&=\left\{y:-2\log\Lambda_{n}(\boldsymbol{x}_n, y)\leq\chi^2_{1,1-\alpha}\right\}\\
&=\left\{y:\bar{x}_n-\sigma\sqrt{\frac{n+1}{n}\chi^2_{1,1-\alpha}}\leq y\leq\bar{x}_n+\sigma\sqrt{\frac{n+1}{n}\chi^2_{1,1-\alpha}}\right\}\\
&=\left\{y:\bar{x}_n-z_{1-\alpha/2}\sigma\sqrt{\frac{n+1}{n}}\leq y\leq\bar{x}_n+z_{1-\alpha/2}\sigma\sqrt{\frac{n+1}{n}}\right\}
\end{split}
\]
where $\chi_{1,1-\alpha}^2$ is the $1-\alpha$ quantile of $\chi^2_1$ and $z_{1-\alpha/2} = \sqrt{\chi^2_{1,1-\alpha}}$ is the $1-\alpha/2$ quantile of a standard normal variable.
In this example, because $-2\log\Lambda_n(\boldsymbol{X}_n, Y)$ is a unimodal function
of $Y$, the prediction region $\mathcal{P}_{1-\alpha}(\boldsymbol{x}_n)$
leads to a prediction interval and the LR prediction method has exact coverage probability because the log-LR statistic is a pivotal quantity (i.e., $\chi_1^2$-distributed).

The second approach for approximating the distribution of $-2\log\Lambda_n(\boldsymbol{X}_n,Y)$, when applicable, is to use Wilks' theorem. Under the
conditions given in \citet{wilks1938}, the LR statistic $-2\log\Lambda_{n}(\boldsymbol{X}_n, Y)\xrightarrow{d}\chi_d^2$,
where $d$ is the difference in the dimensions of the full and reduced models ($d=1$ in our prediction interval construction).
Similarly, the $1-\alpha$ prediction region based on Wilks' theorem is
\[\mathcal{P}_{1-\alpha}(\boldsymbol{x}_n)=\left\{y:-2\log\Lambda_{n}(\boldsymbol{x}_n, y)\leq\chi^2_{d,1-\alpha}\right\}.\]
Wilks' theorem, however, does {\it not} apply in all prediction problems.
There exist important cases, particularly with discrete data, where the Wilks' result is valid for the log-LR statistic $-2\log\Lambda_n(\boldsymbol{X}_n,Y)$ in prediction and the chi-square-calibrated prediction region above is then appropriate; this is described in Section~\ref{sec:discrete-distribution}.
When Wilks' theorem does not apply, an alternative limiting distribution may still exist as illustrated in Section~\ref{asymptotic-results}.

The third approach, which is the most general one, is to use parametric bootstrap.
If $\lambda_{1-\alpha}$ is the $1-\alpha$ quantile of $\Lambda_{n}(\boldsymbol{X}_n,Y)$, then we have the following prediction region 
\begin{equation*}
	\mathcal{P}_{1-\alpha}(\boldsymbol{x}_n)=\left\{y:-2\log\Lambda_n(\boldsymbol{x}_n, y)\leq\lambda_{1-\alpha}\right\}.
\end{equation*}
The idea of this approach is to use a parametric bootstrap re-creation of the data $(\boldsymbol{X}_n^*,Y^*)$, which leads to a distribution for a bootstrap version $-2 \log \Lambda_n(\boldsymbol{X}_n^*,Y^*)$ of the log-LR statistic.  Then, the $1-\alpha$ quantile of the bootstrap distribution, say $\lambda_{1-\alpha}^*$, is used to approximate the unknown quantile $\lambda_{1-\alpha}$ of the true sampling distribution of  $-2 \log\Lambda_n (\boldsymbol{X}_n,Y)$.
Then the resulting parametric bootstrap prediction region is
\begin{equation}\label{eq:bootstrap-find-lik-ratio}
	\mathcal{P}_{1-\alpha}(\boldsymbol{x}_n)=\left\{y:-2\log\Lambda_n(\boldsymbol{x}_n, y)\leq\lambda_{1-\alpha}^\ast\right\}.
\end{equation}
An algorithm for implementing a Monte Carlo (i.e., simulation-based) approximation of the parametric bootstrap is as follows.
\begin{enumerate}
	\item Compute an estimate corresponding to a consistent estimator of $\btheta$ (usually the ML estimate) using observed data $\boldsymbol{X}_n=\boldsymbol{x}_n$, denoted by $\htheta$ (recall the data model is that the $\boldsymbol{X}_n$ are iid $f(\cdot;\btheta)$.)
	\item Generate a bootstrap sample $\boldsymbol{x}_n^\ast$ and $y^\ast$ as iid observations drawn from $f(\cdot;\widehat{\btheta})$.
	\item Evaluate the LR in (\ref{eq:likelihood-ratio}) using bootstrap pair $(\boldsymbol{x}_n^\ast, y^\ast)$ to get $\lambda^\ast\equiv-2\log\Lambda_{n}(\boldsymbol{x}_n^\ast, y^\ast)$.
	\item Repeat steps 2--3 $B$ times to obtain $B$ realizations of $\lambda^\ast$ as $\{\lambda^\ast_b\}_{b=1}^{B}$.
	\item Use the $1-\alpha$ sample quantile of $\{\lambda^\ast_b\}_{b=1}^{B}$ as $\lambda_{1-\alpha}^\ast$ in (\ref{eq:bootstrap-find-lik-ratio}) and compute the prediction region.
\end{enumerate}
The prediction region $\mathcal{P}_{1-\alpha}(\boldsymbol{x}_n)$ in (\ref{eq:bootstrap-find-lik-ratio}) is a prediction interval when $\Lambda_{n}(\boldsymbol{x}_n,y)$ is a unimodal function of $y$ for a given data set $\boldsymbol{X}_n=\boldsymbol{x}_n$.

Such prediction intervals generally do not have equal-tail
probabilities. In many applications, however, the cost of the
predictand being greater than the upper bound is different than
having it being less than the lower bound.
In such cases, it is better to have a prediction interval with equal-tail
probabilities. This can be achieved by calibrating separately the lower
and upper one-sided $1-\alpha/2$ prediction bounds and putting them
together to provide a two-sided $1-\alpha$ equal-tail-probability
prediction interval. The next section shows how to construct a one-sided prediction bound using the LR in (\ref{eq:likelihood-ratio}).

\subsection{Constructing One-Sided Prediction Bounds}\label{subsec:one-sided}

Suppose that the LR $\Lambda_{n}(\boldsymbol{x}_n,y)$ is a unimodal function of $y$ based on observed data $\boldsymbol{X}_n=\boldsymbol{x}_n$.
This is a common property (holding with probability 1) in most prediction problems.  Note that a two-sided prediction interval (4) for $Y$ based on $\boldsymbol{X}_n=\boldsymbol{x}_n$ is defined by a horizontal line drawn through the curve of $-2\log \Lambda_n(\boldsymbol{x}_n,y) $ (as a function of $y$) at an appropriate threshold $\lambda_{1-\alpha}$, as shown in Figure~\ref{fig:horizontal}.
\begin{figure}[t!]
	\centering
	\includegraphics[width=\textwidth]{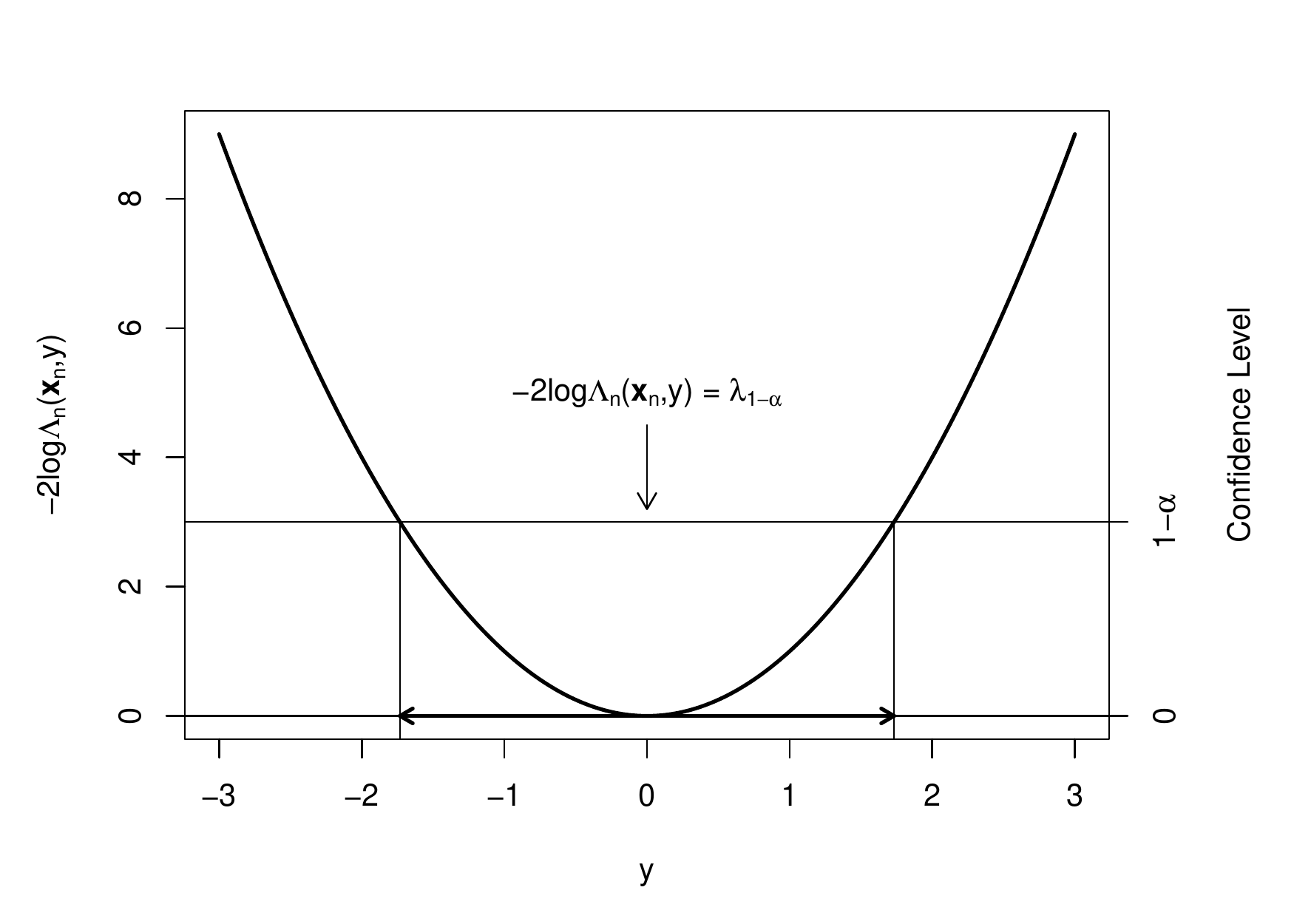}
	\caption{Example of log-LR statistic (as a function of $y$) for given data $\boldsymbol{x}_n$, which is an illustration of the prediction interval procedure in (\ref{eq:bootstrap-find-lik-ratio}).}
	\label{fig:horizontal}
\end{figure}

Here we describe a method for calibrating one-sided bounds directly, without resorting to (4) by adjusting the log-LR curve so that it becomes a monotone function.
For a given data set $\boldsymbol{X}_n=\boldsymbol{x}_n$, let $y_0\equiv y_0(\boldsymbol{x}_n)$ denote the value of $y$ that maximizes $\Lambda_n(\boldsymbol{x}_n,y)$, where $\Lambda_n(\boldsymbol{x}_n,y_0)=1$ at $y_0$.  Define a signed log-LR statistic $\zeta_n(\boldsymbol{x}_n,y)$ based on (3) as
\begin{equation}
\label{eq-extended-lik-ratio}
\zeta_n(\boldsymbol{x}_n,y)\equiv(-1)^{\text{I}(y \leq y_0)} \left[-2  \log \Lambda_n(\boldsymbol{x}_n,y)\right]  = \left\{ 
\begin{array}{ll}
	2  \log \Lambda_n(\boldsymbol{x}_n,y) \in (-\infty,0]& \mbox{$y \leq y_0$}\\
	-2  \log \Lambda_n(\boldsymbol{x}_n,y)    \in [0,\infty) & \mbox{$y \geq y_0$},\\
\end{array}
\right.
\end{equation}   
where $\text{I}(\cdot)$ denotes the indicator function.  That is,  $(-1)^{\text I(y \leq y_0)} \left[-2  \log \Lambda_n(\boldsymbol{x}_n,y)\right]$ is a signed version of the log-LR statistic $-2 \log \Lambda_n(\boldsymbol{x}_n,y)$ which, unlike the latter statistic, is an increasing function of $y$
and is negative when $y<y_0$ (but positive when $y>y_0$).  
Hence, to set a one-sided bound for $Y$, we calibrate the signed log-LR statistic $\zeta_n(\boldsymbol{x}_n,y)$ which has a one-to-one correspondence to $y$-values when $\Lambda_n(\boldsymbol{x}_n,y)$ is unimodal (unlike $\Lambda_n(\boldsymbol{x}_n,y)$ itself).
\begin{figure}[t!]
	\centering
	\includegraphics[width=\textwidth]{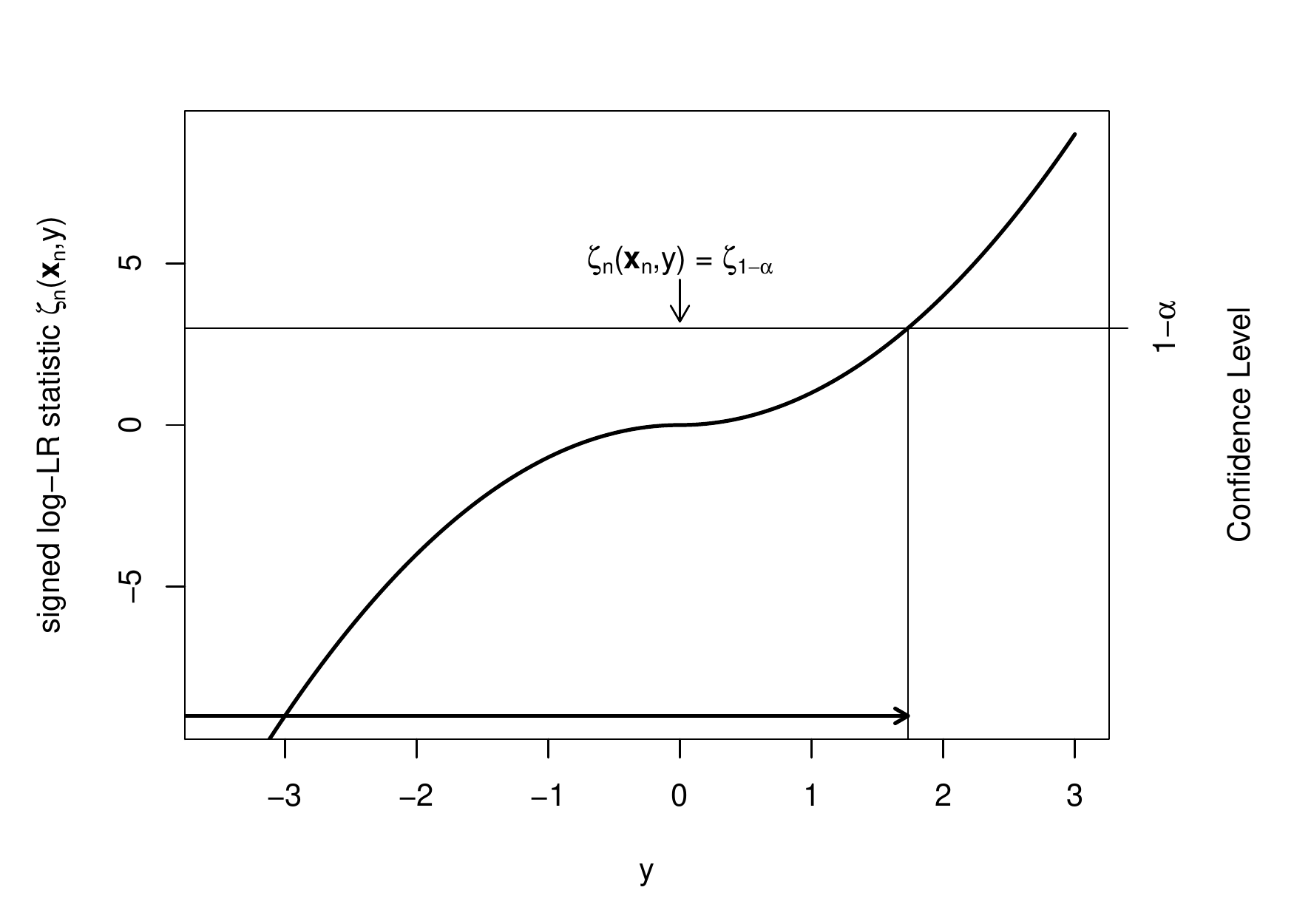}
	\caption{An illustration of the one-sided prediction bound procedure in (\ref{eq-vertical}).}
	\label{fig:vertical}
\end{figure}
Note that if the $1-\alpha$ quantile of the distribution of $\zeta_n(\boldsymbol{X}_n,Y)$, denoted by $\zeta_{1-\alpha}$, were known,  we could set a $1-\alpha$ upper prediction bound for $Y$ given $\boldsymbol{X}_n=\boldsymbol{x}_n$ as
\begin{equation}\label{eq-vertical}
\tilde{y}_{1-\alpha}(\boldsymbol{x}_n) \equiv \sup\{y\in \mathbb{R}: \zeta_n(\boldsymbol{x}_n,y) \leq \zeta_{1-\alpha}\}.
\end{equation}
Figure~\ref{fig:vertical} provides a graphical illustration of (\ref{eq-vertical}), illustrating the resulting prediction region.
Similar to the third approach in Section~\ref{sec-determine}, we can approximate the quantile $\zeta_{1-\alpha}^\ast$ using the $1-\alpha$ quantile of $\zeta_{n}(\boldsymbol{X}_n^\ast,Y^\ast)$, which is the bootstrap version of the signed log-LR statistic.
Then, a bootstrap prediction bound is obtained by replacing $\zeta_{1-\alpha}$ with $\zeta_{1-\alpha}^\ast$ in (\ref{eq-vertical}), and the $1-\alpha$ upper prediction bound $\tilde{y}_{1-\alpha}(\boldsymbol{x}_n)$ is defined as
\begin{equation}\label{eq:upper}
	\tilde{y}_{1-\alpha}(\boldsymbol{x}_n)=\sup_{y\in\mathbb{R}}\left\{y:\zeta_{n}(\boldsymbol{x}_n,y)\leq \zeta_{1-\alpha}^\ast\right\}.
\end{equation}
Constructing the $1-\alpha$ lower prediction bound $\utilde{y}_{1-\alpha}(\boldsymbol{x}_n)$ is similar, and the $1-\alpha$ lower prediction bound is
\begin{equation}\label{eq;lower}
	\utilde y_{1-\alpha}(\boldsymbol{x}_n)=\sup_{y\in\mathbb{R}}\left\{y:\zeta_{n}(\boldsymbol{x}_n,y)\leq \zeta_\alpha^\ast\right\}.
\end{equation}

The following algorithm describes how to compute the $1-\alpha$ upper (and lower) prediction bound $\tilde{y}_{1-\alpha}(\boldsymbol{x}_n)$ (and $\utilde{y}_{\alpha}(\boldsymbol{x}_n)$) using a Monte Carlo approximation of the bootstrap distribution $\zeta_{n}(\boldsymbol{X}_n^\ast,Y^\ast)$ and the bootstrap quantile $\zeta_{1-\alpha}^\ast$ (and $\zeta_{\alpha}^\ast$).
\begin{enumerate}
	\item Compute $\htheta$ using the observed data $\boldsymbol{X}_n=\boldsymbol{x}_n$.
	\item Simulate a sample $\boldsymbol{x}_n^\ast$ using a parametric bootstrap with $\htheta$ and compute $y_0(\boldsymbol{x}_n^\ast)$.
	\item Simulate $y^\ast$ from distribution $f(\cdot;\htheta)$ and compute $$\zeta^\ast\equiv\zeta_n(\boldsymbol{x}_n^\ast,y^\ast)=(-1)^{\text{I}[y^\ast\leq y_0(\boldsymbol{x}_n^\ast)]}\left[-2\log\Lambda_n(\boldsymbol{x}_n^\ast,y^\ast)\right].$$
	\item Repeat steps 2--3 $B$ times to obtain $B$ realizations of $\zeta^\ast$ as $\{\zeta_{i}^\ast\}_{i=1}^{B}$.
	\item Use the $1-\alpha$ (or $\alpha$) sample quantile from $\{\zeta_{i}^\ast\}_{i=1}^{B}$ as $\zeta_{1-\alpha}^\ast$ (or $\zeta_{\alpha}^\ast$) in (\ref{eq:upper}) (or (\ref{eq;lower})) to compute the $1-\alpha$ upper (or lower) prediction bound.
\end{enumerate}

Note that, in the algorithm for one-side bounds, one can simultaneously keep track of bootstrap statistics $\lambda^\ast=\left|\xi^\ast\right|$ for computing the two-sided bounds in (\ref{eq:bootstrap-find-lik-ratio}) (i.e., the same resamples can be used).

\section{Exact Results}\label{sec:pivotal-quantity}

The LR-based prediction method can often uncover and exploit pivotal quantities involving the data $\boldsymbol{X}_n$ and the predictand $Y$ when these exist.
In these cases, the LR statistic is pivotal, often emerging as a function of another pivotal quantity from $(\boldsymbol{X}_n,Y)$.
Consequently, in these cases, prediction intervals or bounds for $Y$ based on the LR-statistic (\ref{eq:likelihood-ratio}) will have exact coverage, when either based on the direct distribution of LR statistic (when available analytically) or more broadly when based on a bootstrap.
In this section, we provide more explanation about when the LR prediction method is exact, beginning with some illustrative examples.

\noindent\textbf{Exponential Distribution}:
Suppose the data $X_1,\dots,X_n$ and future predictand $Y$ are iid $\text{Exp}(\theta)$ with mean $\theta>0$. Letting $\widehat\theta_{\boldsymbol{x}_n,y}\equiv\left(\sum_{i=1}^{n}x_i+y\right)/(n+1)$, and $\widehat{\theta}_{\boldsymbol{x}_n}\equiv\sum_{i=1}^{n}x_i/n$ based on data $\boldsymbol{X}_n=\boldsymbol{x}_n$ and a given value $y>0$ of $Y$, then the LR statistic (\ref{eq:likelihood-ratio}) is
\[
\Lambda_n(\boldsymbol{x}_n,y)=\frac{\widehat\theta^{-n-1}_{\boldsymbol{x}_n,y}\exp\left[-\frac{\sum_{i=1}^{n}x_i+y}{\widehat{\theta}_{\boldsymbol{x}_n,y}}\right]}{\widehat{\theta}_{\boldsymbol{x}_n}^{-n}\exp\left[-\frac{\sum_{i=1}^{n}x_i}{\widehat{\theta}_{\boldsymbol{x}_n}}\right]y^{-1}\exp\left(-\frac{y}{y}\right)}=\frac{y\widehat{\theta}_{\boldsymbol{x}_n}^n}{\widehat{\theta}_{\boldsymbol{x}_n,y}^{n+1}}=\frac{\left(\frac{\bar{x}_n}{y}\right)^n}{\left[\frac{n}{n+1}\frac{\bar{x}_n}{y}+\frac{1}{n+1}\right]^{n+1}},
\]
which is a function of the pivotal quantity $\bar{X}_n/Y$ and a unimodal function of $y$.
Thus, the LR prediction method is exact (when based on the $F$-distribution of $\bar{X}_n/Y\sim F_{n,1}$ or the bootstrap as in (\ref{eq:bootstrap-find-lik-ratio})), and the prediction region becomes a prediction interval.

\noindent\textbf{Normal Distribution}:
Let $X_1,\dots,X_n,Y\sim\text{Norm}(\mu,\sigma)$, where both $\mu\in\mathbb{R}$ and $\sigma>0$ are unknown.
We construct the full model by allowing the predictand $Y$ to have a different location parameter $\mu$: $X_1\dots,X_n\sim\text{Norm}(\mu,\sigma)$ and $Y\sim\text{Norm}(\mu_y,\sigma)$ (i.e., $\btheta=(\mu,\sigma)$ and $\btheta_y=(\mu_y,\sigma)$).
Then for the full model, the ML estimators are
\[
\widehat{\mu}=\bar{X}_n,\quad\widehat{\mu}_y=Y,\quad\widehat{\sigma}=\sqrt{\frac{\sum_{i=1}^{n}(X_i-\bar{X}_n)^2}{n+1}},
\]
while for the reduced model $\theta=\theta_y$, the ML estimators are
\[
\widehat{\mu}=\frac{\sum_{i=1}^{n}X_i+Y}{n+1},\quad\widehat{\sigma}=\sqrt{\frac{\sum_{i=1}^{n}(X_i-\widehat{\mu})^2+(Y-\widehat{\mu})^2}{n+1}}.
\]
Then the resulting LR statistic (\ref{eq:likelihood-ratio}) is
\begin{equation}\label{eq:normal-two-para}
		\Lambda_n(\boldsymbol{X}_n,Y)=\left(1 + \frac{n^2+1}{n^2-1}\frac{t^2}{n} \right)^{-(n+1)/2}
\end{equation}
where
\[
t\equiv\frac{\bar{X}_n-Y}{s}\sqrt\frac{n}{n+1},
\]
and $s^2 \equiv  \sum_{i=1}^n (X_i-\bar{X}_n)^2/(n-1)$.
Here, $t \sim t_{n-1}$ has the same $t$-test statistic form as in (\ref{eq:normal-invert-t-test}).
Hence, the LR is pivotal and also $\Lambda_n(\boldsymbol{x}_n,y)$ is a unimodal function of $y$.
Thus the resulting prediction interval procedure has exact coverage probability when based on the bootstrap as in (\ref{eq:bootstrap-find-lik-ratio}) (or using the $t_{n-1}$ distribution here).

In fact, the results for the normal distribution can be generalized to the (log-)location-scale family, which contains many other important distributions.
Theorem~\ref{theorem-location-scale-family-11} says that, by allowing the location parameter of the predictand to be different from that of the data to create a full versus reduced model comparison, the resulting LR statistic (\ref{eq:likelihood-ratio}) is a pivotal quantity so that the prediction method is exact.
\begin{theorem}\label{theorem-location-scale-family-11}
	(i) Suppose the LR-statistic (\ref{eq:likelihood-ratio}) is a pivotal quantity.
	Then, the corresponding $1-\alpha$ prediction region (\ref{eq:bootstrap-find-lik-ratio}) for $Y$ based on the parametric bootstrap will have exact coverage.
	That is,
	\[
	\Pr\left[Y\in\mathcal{P}_{1-\alpha}(\boldsymbol{X}_n)\right]=1-\alpha.
	\]
	(ii) Suppose also that both the data $X_1,\dots,X_n$ and $Y$ are from a location-scale distribution with density $f(\cdot;\mu,\sigma)=\phi\left[(x-\mu)/\sigma\right]$ with parameters $\btheta=(\mu,\sigma)\in\mathbb{R}\times(0,\infty)$.
	In the LR construction (\ref{eq:likelihood-ratio}), suppose the full model involves parameters $\btheta=(\mu,\sigma)$ and $\btheta_{y}=(\mu_y,\sigma)$ (i.e., $X_1,\dots,X_n\sim f(\cdot;\mu,\sigma)$ and $Y\sim f(\cdot;\mu_y,\sigma)$).
	Then the LR statistic $\Lambda_{n}(\boldsymbol{X}_n,Y)$ (or $-2\log\Lambda_{n}(\boldsymbol{X}_n,Y)$) is a pivotal quantity and the result of Theorem~\ref{theorem-location-scale-family-11}(i) holds.
\end{theorem}
\noindent The proof is given in Section~A of supplementary material.

\noindent\textbf{Remark 1.} If the LR statistic $\Lambda_{n}(\boldsymbol{x}_n,y)$ is a unimodal function of $y\in\mathbb{R}$ with probability 1 (as determined by $\boldsymbol{X}_n$) and if the signed LR-statistic $\zeta_{n}(\boldsymbol{X}_n,Y)$ is a pivotal quantity, then the Theorem~\ref{theorem-location-scale-family-11}(i) result (i.e., exact coverage) also applies for one-sided prediction bounds based on parametric bootstrap.
For (log-)location-scale distributions as in Theorem~\ref{theorem-location-scale-family-11}(ii), the signed LR-statistic $\zeta_{n}(\boldsymbol{X}_n,Y)$ is a pivot.

We next provide some illustrative examples.

\noindent\textbf{Simple Regression}: We consider the simple linear regression model $Y\sim\text{Norm}(\beta_0+\beta_1x,\sigma)$ with given $x$ and data $Y_1,\dots,Y_n$ that satify $Y_i\sim\text{Norm}(\beta_0+\beta_1x_i,\sigma)$ where $x_i, i=1,\dots,n$.
Similar to the normal distribution example, it is natural to choose $\beta_0+\beta_1x$ to construct the ``full'' model.
In fact, choosing $\beta_0$ or $\beta_1$ gives the same log-LR statistic as $\beta_0+\beta_1x$, which is given by
\[
(n+1)\log\left(1+\frac{1}{n-2}T^2\right),
\]
where $T$ matches the standard statistic for normal theory predictions (i.e., a studentized version of $Y-\widehat{\beta}_0-\widehat{\beta}_1x$) having a $t$-distribution with $n-2$ degrees of freedom.

\noindent\textbf{Two-Parameter Exponential Distribution}:
Suppose $X_1,\dots,X_n,Y$ are independent observations from a two-parameter exponential distribution $\text{Exp}(\mu, \beta)$.
That is, $(X_i-\mu)/\beta\sim\text{Exp}(1)$ with location and scale parameters as $\btheta=(\mu,\beta)$.
Hence, under Theorem~\ref{theorem-location-scale-family-11}, the LR $\Lambda_{n}(\boldsymbol{X}_n,Y)$ is a pivotal quantity and bootstrap-calibrated prediction regions for $Y$ are exact.
In fact, an exact form of the LR-statistic may be determined as
\begin{equation*}
	\Lambda_n(\boldsymbol{x}_n, y)=
	\left[\frac{\sum_{i=1}^{n}x_i+y-(n+1)\min\{x_{(1)}, y\}}{\sum_{i=1}^{n}x_i-nx_{(1)}}\right]^{n+1}
	\label{eq:two-parameter-exponential}
\end{equation*}
based on given positive data $\boldsymbol{x}_n=(x_1,\ldots,x_n)$ where $x_{(1)}$ denotes the first order statistic.
Note that $\Lambda_n(\boldsymbol{x}_n,y)$ is a unimodal function of $y$ given $\boldsymbol{X}_n=\boldsymbol{x}_n$ (with probability 1); hence, one-sided prediction bounds for $Y$ based on a parametric bootstrap will also have exact coverage by Remark~1.
Replacing $\boldsymbol{x}_n$ and $y$ with corresponding random variables $\boldsymbol{X}_n$ and $Y$
in (\ref{eq:two-parameter-exponential}) gives
\[ \Lambda_n(\boldsymbol{X}_n, Y)\stackrel{d}{=}\left[\frac{\sum_{i=1}^{n}E_i+T-(n+1)\min\{E_{(1)}, T\}}{\sum_{i=1}^{n}E_i-nE_{(1)}}\right]^{n+1}, \]
where $E_1, \dots, E_n, T$ denote iid $\text{Exp}(1)$ random variables and $E_{(1)}$ is the first order statistic of
$E_1, \dots, E_n$; this verifies that $\Lambda_n(\boldsymbol{X}_n, Y)$ is
indeed a pivotal quantity for the exponential data case, as claimed in Theorem~\ref{theorem-location-scale-family-11}.
Thus this prediction interval procedure is exact when parametric bootstrap is used to obtain the distribution of $\Lambda_n(\boldsymbol{X}_n, Y)$.

\noindent\textbf{Uniform Distribution}:
Suppose $X_1,\dots,X_n,Y$ are iid $\text{Unif}(0,\theta)$, which is a one-parameter scale family.
The LR statistic (\ref{eq:likelihood-ratio}) has a form
\begin{equation}\label{eq-uniform-lr}
\Lambda_{n}(\boldsymbol{x}_n,y)=\frac{(x_{(n)}/y)^n}{[\max(x_{(n)}/y,1)]^{n+1}},
\end{equation}
where $x_{(n)}$ denotes the maximum of $\{x_1,\dots,x_n\}$.
Hence, $\lambda_n(\boldsymbol{x}_n,y)$ is a unimodal function $y$ given $\boldsymbol{X}_n=\boldsymbol{x}_n$ (with probability 1) and $\Lambda_{n}(\boldsymbol{X}_n,Y)$ can also be seen to be a pivotal quantity as
\[
\frac{X_{(n)}}{Y}=\frac{X_{(n)}/\theta}{Y/\theta}\stackrel{d}{=}\frac{\max\{U_1,\dots,U_n\}}{U_0},
\]
where $U_0,U_1,\dots,U_n$ denote iid $\text{Unif}(0,1)$ variables.
Hence, by Theorem~\ref{theorem-location-scale-family-11}(i) and Remark~1, both the two-sided prediction interval procedure (\ref{eq:bootstrap-find-lik-ratio}) as well as the one-sided bound procedures (\ref{eq:upper})-(\ref{eq;lower}) based on bootstrap have exact coverage.
That is, bootstrap simulation provides an effective and unified means for estimating the distribution of $\Lambda_{n}(\boldsymbol{x}_n,y)$ and constructing prediction intervals.

\section{General Results}\label{asymptotic-results}

Section~\ref{sec:pivotal-quantity} discusses cases where the LR prediction method is exact, particularly when the construction (\ref{eq:likelihood-ratio}) results in a pivotal quantity.
For some prediction problems, however, the LR statistic may not be a pivotal quantity, as the next example illustrates.

\noindent\textbf{Gamma Distribution}:
Let $X_1,\dots,X_n,Y$ denote iid random variables from a gamma density $f(x;\alpha,\beta)=\beta^{-\alpha} x^{\alpha-1} \exp(-x/\beta)/{\Gamma(\alpha)}$, $x>0$, with scale $\beta>0$ and shape $\alpha>0$ parameters.
In the LR construction (\ref{eq:likelihood-ratio}) with parameters $\btheta=(\beta,\alpha)$, suppose the full model involves parameters $\btheta$ and $\btheta_y=(\beta_y,\alpha)$ or $X_1,\dots,X_n\sim\text{Gamma}(\alpha,\beta)$ and $Y\sim\text{Gamma}(\alpha,\beta_y)$.
The LR statistic is then given by
\[
\Lambda_{n}(\boldsymbol{x}_n,y)=\frac{\sup_{\alpha} [\Gamma(\alpha)]^{-n} [(n\bar{x}_n+y)/(n+1)]^{-\alpha(n+1)} (y \prod_{i=1}^n x_i)^{\alpha-1}}{\sup_{\alpha} [\Gamma(\alpha)]^{-n} [n\bar{x}_n]^{-\alpha n} (y/\alpha)^{-\alpha} (y \prod_{i=1}^n x_i)^{\alpha-1}}
.
\]
Unlike the previous examples, the LR statistic is no longer a pivotal quantity.
However, we can use the bootstrap method to approximate the distribution for $\Lambda_{n}(\boldsymbol{X}_n,Y)$.
A small simulation study was conducted to investigate the coverage probability of the LR prediction method.
Figure~\ref{fig:gamma} shows the coverage probability of 90\% and 95\% one-sided prediction bounds for a future gamma variate based on the LR prediction method (i.e., (\ref{eq:upper}) and (\ref{eq;lower})), and compares the LR prediction method with other methods.
Sample size values $n=4,5,6,7,8,9,10,30,50,70,90,100$ were used.
Without loss of generality, the scale parameter was set to $\beta=1$, and the shape parameter values $\alpha=1,2,3$ were used.
$N=2000$ Monte Carlo samples were used to compute the coverage probability and $B=2000$ bootstrap samples were used to approximate the distribution of the signed log-LR statistic.
\begin{figure}[t!]
	\centering
	\includegraphics[width=\textwidth]{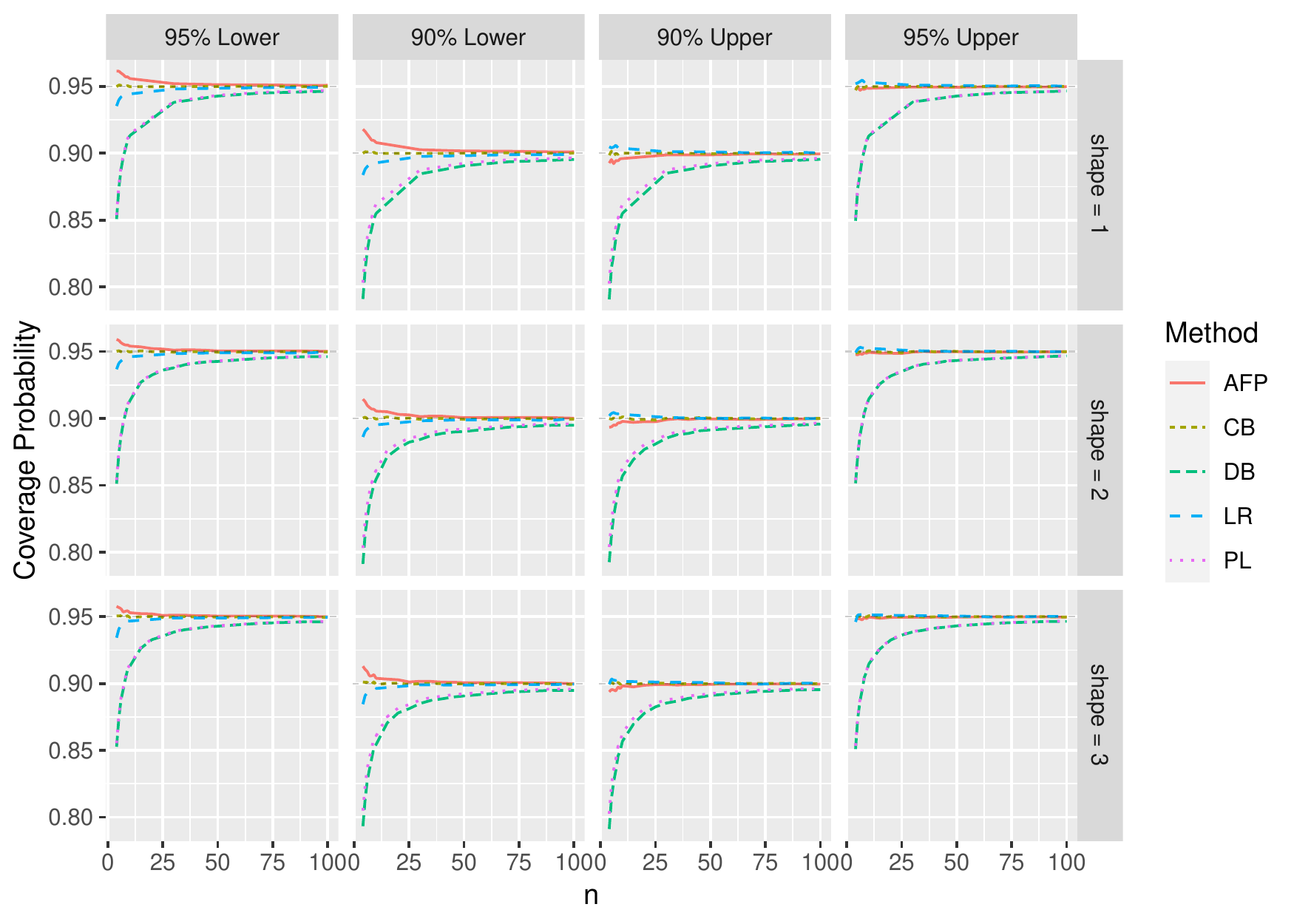}
	\caption{Coverage probabilities for predicting a gamma random variable versus the sample size $n$ for the $90\%$ and $95\%$ one-sided prediction bounds: Approximate Fiducial Prediction (AFP) (\citealt{chen2017approximate}), Calibration Bootstrap (CB) (\citealt{beran1990}), Direct-Bootstrap (DB) (\citealt{harris1989}), Likelihood Ratio Prediction (LR) (cf.~(\ref{eq:upper}), and (\ref{eq;lower})); Plug-in (PL) Methods.}
	\label{fig:gamma}
\end{figure}
The simulation results show that the calibration-bootstrap method has the best coverage while the LR prediction and the approximate fiducial (or GPQ) methods also work well.
When $n=4$, the difference between the true coverage of the LR prediction method and the nominal level (combined with Monte Carlo error) is less than $2\%$.
When the sample size $n$ increases, the discrepancy quickly shrinks.
This illustrtates that even when the LR statistic has a complicated and non-pivotal distribution, using parameteric bootstrap can be effective and useful.

Theorem~\ref{theorem-multiple-parameters}, given next, shows that the LR prediction method, combined with bootstrap calibration is asymptotically correct for continuous prediction problems under general conditions.
The theorem consists of two parts: the first part establishes that the log-LR statistic $-2\log\Lambda_{n}(\boldsymbol{X}_n,Y)$ has a limit distribution as $n\to\infty$.
However, this limit distribution will sometimes {\it not} be chi-square as in Wilks' theorem and may even depend on one or more of the parameters.
Nevertheless, the second part of Theorem~\ref{theorem-multiple-parameters} establishes that the bootstrap version of the log-LR statistic $-2\log\Lambda_{n}(\boldsymbol{X}_n^\ast,Y^\ast)$ can capture the distribution of $-2\log\Lambda_{n}(\boldsymbol{X}_n,Y)$.
Consequently, $1-\alpha$ bootstrap-based prediction regions (\ref{eq:bootstrap-find-lik-ratio}) for $Y$ will have coverage probabilities that converge to the correct coverage level $1-\alpha$ as the sample size $n$ increases.
\begin{theorem}\label{theorem-multiple-parameters}
Suppose a random sample $\boldsymbol{X}_n$ of size $n$ and a predictand $Y$ (independent of $\boldsymbol{X}_n$) have a common density $f(\cdot;\btheta)$, and that the LR construction (3) is used with $\boldsymbol{\theta}= (\theta,\boldsymbol{\theta}^\prime)$ and $\boldsymbol{\theta}_y = (\theta_y,\boldsymbol{\theta}^\prime)$ having common parameters $\boldsymbol{\theta}^\prime$ (and real-valued parameters $\theta,\theta_y$ that may differ).
Then, under mild regularity conditions (detailed in the supplement),
	\begin{enumerate}
		\item As $n\to\infty$,
		\[
		-2\log\Lambda_{n}(\boldsymbol{X}_n,Y)\xrightarrow{d}-2\log\left[\frac{f(Y;\btheta_0)}{\sup_{\theta_{y}}f(Y;\theta_{y},\btheta_{0}^\prime)}\right],
		\]
		where $\btheta_0=(\theta_0,\btheta_{0}^\prime)$ denotes the true value of the parameter vector $\btheta$.
		\item The bootstrap provides an asymptotically consistent estimator of the distribution of the log-LR statistic; that is,
		\[
		\sup_{\lambda\in\mathbb{R}}\left|\Pr{}_{\!\ast}(-2\log\Lambda_{n}^\ast\leq\lambda)-\Pr\left(-2\log\Lambda_{n}\leq\lambda\right)\right|\xrightarrow{p}0,
		\]
		where $\Pr{}_{\!\ast}$ is the bootstrap induced probability and $\Lambda^\ast_n\equiv\Lambda_n(\boldsymbol{X}_n^\ast,Y^\ast)$.
	\end{enumerate}
\end{theorem}
\noindent\textbf{Remark 2.}
Similar to Remark~1, if $\Lambda_{n}(\boldsymbol{X}_n,y)$ is a unimodal function of $y$ (with probability 1 or with probability approaching 1 as $n\to\infty$), then the signed log-LR statistic converges as well
\[
(-1)^{\text{I}[Y\leq y_0(X_n)]}\left[-2\log\Lambda_n(\boldsymbol{X}_n,Y)\right] \stackrel{d}{\rightarrow}
\begin{cases}
	2\log\left[\frac{f(Y;\boldsymbol{\theta}_0)}{\sup_{\theta_y} f(Y;\theta_y,\boldsymbol{\theta}^\prime_0)}\right], & Y\leq m_0.\\
	-2\log\left[\frac{f(Y;\boldsymbol{\theta}_0)}{\sup_{\theta_y} f(Y;\theta_y,\boldsymbol{\theta}^\prime_0)}\right], & Y> m_0.
\end{cases}
\]
where $m_0$ is maximizer of ${f(y;\boldsymbol{\theta}_0)}/{\sup_{\theta_y} f(y; \theta_y,\boldsymbol{\theta}^\prime_0)}$ over $y$.
The bootstrap approximation for the signed log-LR statistic is also valid asymptotically.
The proof is described in the supplementary material along with a proof of Theorem~\ref{theorem-multiple-parameters}.

We use two examples to illustrate Theorem~\ref{theorem-multiple-parameters}.
In the uniform example of Section~\ref{sec:pivotal-quantity}, if $\theta_0>0$ denotes the true parameter value (i.e., $Y\sim\mbox{Unif}(0,\theta_0)$), then the limit distribution in Theorem~\ref{theorem-multiple-parameters}(i) for the log-LR statistic is
\begin{equation}\label{eq-new-key-1}
-2\log\Lambda_{n}(\boldsymbol{X}_n,Y)\xrightarrow{d}-2\log\left(\frac{Y}{\theta_0}\right),
\end{equation}
which has a $\chi^2_{2}$ distribution.
This result can be alternatively verified by using the LR in (\ref{eq-uniform-lr}) to determine that
\[
\Lambda_{n}(\boldsymbol{X}_n,Y)=\frac{(X_{(n)}/Y)^n}{[\max(X_{(n)}/Y,1)]^{n+1}}=\frac{Y}{X_{(n)}}\frac{(X_{(n)}/Y)^{n+1}}{\left[\max(X_{(n)}/Y,1)\right]^{n+1}}\xrightarrow{d}\text{Unif}(0,1)
\]
from which $-2\log\Lambda_{n}(\boldsymbol{X}_n,Y)\xrightarrow{d}\chi_2^2$ follows.
While $\Lambda_{n}(\boldsymbol{X}_n,Y)$ is a pivotal quantity for any $n\geq1$ (so that bootstrap calibration is exact by Theorem~\ref{theorem-location-scale-family-11}), Theorem~\ref{theorem-multiple-parameters} shows that the bootstrap also captures the limiting distribution of the log-LR statistic $\chi_2^2$ in (\ref{eq-new-key-1}).

To consider a distribution with more than one parameter, we re-visit the gamma distribution example in this section.
Using Theorem~\ref{theorem-multiple-parameters}, the limit distribution is
\begin{equation}\label{key}
\begin{split}
-2\log\Lambda_{n}(\boldsymbol{X}_n,Y)&\xrightarrow{d}-2\log\left[\frac{f(Y;\beta_0,\alpha_0)}{\sup_\beta f(Y;\beta,\alpha_0)}\right]\\
&=-2\log\left[\left(\frac{Y}{\beta_0}\right)^{\alpha_0}\exp\left(-\frac{Y}{\beta_0}+\alpha_0\right)\right]\\
&=2(Z-\alpha_0)-2\alpha_0\log(Z),
\end{split}
\end{equation}
where $Z\equiv Y/\beta_0\sim\text{Gamma}(\alpha_0,1)$.
Even though the log-LR statistic (\ref{eq:likelihood-ratio}) depends on the shape parameter $\alpha_0$ in this example, a bootstrap approximation for the distribution of the log-LR statistic is asymptotically correct by Theorem~\ref{theorem-multiple-parameters}.
This is demonstrated numerically through the coverage behavior of Figure~\ref{fig:gamma}.

In addition to the bootstrap (Theorem~\ref{theorem-multiple-parameters}~(ii)), the limit distribution of the log-LR statistic in Theorem~\ref{theorem-multiple-parameters}~(i) (as well as that of the signed log-LR statistic $\zeta_n$ from Remark~2) may also be used as an alternative approach to construct prediction intervals.
That is, we may use the $1-\alpha$ quantile of the limit distribution in Theorem~\ref{theorem-multiple-parameters}~(i) to replace the quantile $\lambda_{1-\alpha}$ in (4) (corresponding to the finite sampling distribution of the log-LR statistic).  For example, in the uniform prediction example above, the limit distribution is $\chi_2^2$ from (\ref{eq-new-key-1}) and an approximate $1-\alpha$ prediction region for $Y$ would be $\{y : -2 \log \Lambda_n(\boldsymbol{x}_n,y) \leq \chi^2_{2,1-\alpha}\}$, which has asymptotically correct coverage by Theorem~\ref{theorem-multiple-parameters}~(i).  As another example from the gamma prediction case, we can use the $1-\alpha$ quantile of the limit distribution in (\ref{key}) to replace $\lambda_{1-\alpha}$ in (\ref{eq:bootstrap-find-lik-ratio}).
This limit distribution, however, depends on the unknown shape parameter $\alpha_0$ in (\ref{key}), which differs from the uniform case where the log-LR statistic has a limit distribution in (\ref{eq-new-key-1}) that is free of unknown parameters.  However, in prediction cases such as the gamma distribution, where the limit distribution of log-LR statistic from Theorem~\ref{theorem-multiple-parameters}~(i) does depend on unknown parameters, we can still approximate and use the limit distribution by replacing any unknown parameters with consistent estimators.
To illustrate with gamma predictions, we may estimate the unknown shape parameter $\alpha_0$ in (\ref{key}) with the ML estimate $\widehat{\alpha}$ from the data $\boldsymbol{X}_n=\boldsymbol{x}_n$ and compute the $1-\alpha$ quantile of the ``plug-in'' version of the limit distribution $2(Z-\widehat\alpha)-2\widehat\alpha\log(Z)$.
Such use of the limit distribution of the log-ratio statistic (Theorem~\ref{theorem-multiple-parameters}~(i)), possibly with plug-in estimation, can be computationally simpler than parametric bootstrap and may have advantages for large sample sizes or when the numerical costs of repeated ML estimation (i.e., as in bootstrap) are prohibitive.

\section{How to Choose the Full Model}\label{sec-choose-full}

When $\btheta$ is a parameter vector, construction of the LR statistic depends on which parameter component in $\btheta$ is varied to create $\btheta_y$ in a full model, where $(\btheta,\btheta_{y})$ again differ in exactly one component.
Our recommendation is to choose a parameter that is most readily identifiable from a single observation.
In other words, we can envision maximizing $f(y|\btheta)$, the density of one observation, with respect to a single unknown parameter of our choice, with all remaining parameters fixed at arbitrary values; the parameter that represents the simplest single maximization step of $f(y|\btheta)$ corresponds to a good parameter choice in the LR construction and choosing such a parameter can simplify computation.
Under some one-to-one reparameterization, if necessary, such a parameter is often given by the mean or the median of the model density $f(y|\btheta)$ that can naturally be identified through a single observation $Y$.
This approach is also supported by Theorem~\ref{theorem-multiple-parameters} where the limiting distribution of the LR statistic is determined by a single-parameter maximization.
We provide some examples in the rest of this section.

\noindent\textbf{Normal Distribution}: For $\text{Norm}(\mu,\sigma)$ with unknown $\mu,\sigma$, consider maximizing the normal density $f(y;\mu,\sigma)$ of a single observation $Y$ with respect to one parameter while the other parameter is fixed.
If choosing $\mu$, then we can estimate $\mu$ simply as $\widehat{\mu}_y=y$.
However, if choosing $\sigma$, we have $\widehat{\sigma}_y^2=(y-\mu)^2$, which is less simple.
More technically, the LR construction for the normal model then eventually involves a complicated estimation of the remaining parameter $\mu$ (from a ``full'' model sample $x_1,\dots,x_n,y$) as the maximizer of $-2\log|y-\mu|-n\log[\sum_{i=1}^{n}(x_i-\mu)^2]$, which can exhibit numerical sensitivity in the value of $y$.
We have seen in Section~\ref{sec:pivotal-quantity} that choosing the mean parameter $\mu$ gives a LR statistic with a nice form and coverage properties, but choosing $\sigma$ results in a much less tractable LR statistic.

\noindent\textbf{Gamma Distribution}: For a gamma distribution with shape $\alpha$ and scale $\beta$, we select the parameter that most easily maximizes a single gamma density $f(y;\alpha,\beta)$ when the other parameter is fixed.
Choosing $\beta$ is simpler than choosing $\alpha$ because the maximizer of the gamma pdf with respect to $\beta$ is $\widehat{\beta}=\alpha/y$ while choosing $\alpha$ does not yield a closed-form maximizer.
Also, choosing $\alpha$ leads to a more complicated LR statistic and a less tractable limit distribution, from Theorem~\ref{theorem-multiple-parameters}.
Alternatively, to more closely align parameter choice in the gamma distribution with parameter identification from one observation, we use another parameterization $(\alpha\beta,\alpha)$ and choose the mean $\alpha\beta$ (i.e., estimated as $y$ analogously to the normal case).
This choice will produce the same LR statistic as choosing $\beta$ in the parameterization $(\beta,\alpha)$.
Hence, choosing a parameter with the simplest stand-alone maximization step in a parameterization and choosing a parameter based on identifiability considerations (e.g., means) in a second parameterization are related concepts.

\noindent\textbf{Generalized Gamma Distribution}:
The (extended) generalized gamma distribution, using the \citet{farewell1977study} parameterization (see also Section~4.13 of \citealt{meeker2021}) has (on the log scale) a location $\mu$, a scale $\sigma$, and a shape parameter $\lambda$ with a pdf given by
\[
f(y;\mu,\sigma,\lambda)=
\begin{cases}
	\dfrac{|\lambda|}{\sigma y}\phi_{lg}[\lambda\omega+\log(\lambda^{-2});\lambda^{-2}]\quad&\text{if}\quad\lambda\neq0\\
	\dfrac{1}{\sigma y}\phi_{\text{norm}}(\omega)&\text{if}\quad\lambda=0,
\end{cases}
\]
where $y>0$, $\omega=[\log(y)-\mu]/\sigma$, $-\infty<\mu<\infty$, $-\infty<\lambda<\infty$, $\sigma>0$, $\phi_{\text{norm}}(\cdot)$ is the pdf of $\text{Norm}(0,1)$, and
$
\phi_{lg}(z;\kappa)=\exp\left[\kappa z-\exp(z)\right]/{\Gamma(\kappa)}
$.
When considering the maximization of a single density $f(y;\mu,\sigma,\lambda)$ for one of the three parameters (with others fixed), the ML estimator of $\mu$ has the simplest form as $\widehat{\mu}=\log(y)$.
Hence, we choose $\mu$ to construct the ``full'' model.
A small simulation study was done to investigate the coverage probability of one-sided prediction bounds.
Fixing the location parameter at $\mu=0$ and the scale parameter at $\sigma=1$ without loss of generality, we consider four different levels for the shape parameter $\lambda=0.5, 1, 1.5, 2$.
The Monte Carlo sample size was set as $N=2000$; the bootstrap sample size was set as $B=2000$.
The data sample sizes are $n=30,35,40,45,50,55,60,70,80,90,100$.
Figure~\ref{fig:gengamma-cp} shows the coverage probabilities for the LR prediction method and for the plug-in method (where the unknown parameters are replaced with the ML estimates) versus the sample size.
We can see that the LR method has good coverage probability and consistently outperforms the plug-in method.
Also, the coverage probability of the LR method, if not close to the nominal confidence level, is conservative.
The plug-in method, however, is always anti-conservative in this simulation study.
\begin{figure}[t!]
	\centering
	\includegraphics[width=0.8\linewidth]{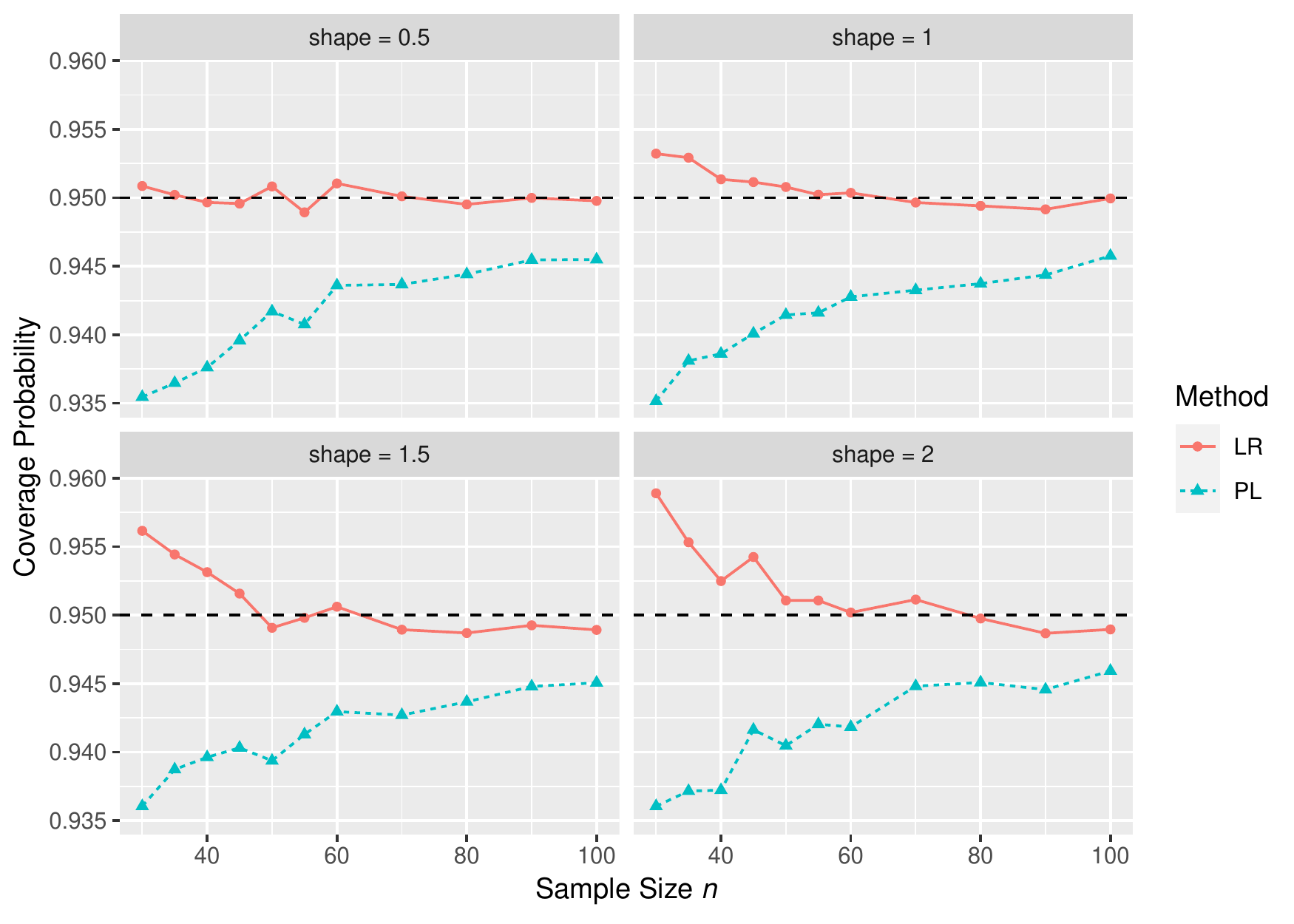}
	\caption{Coverage probabilities of 95\% upper bounds using LR prediction method (LR) and naive plug-in method (PL) versus the sample size $n$.}
	\label{fig:gengamma-cp}
\end{figure}

\section{Discrete Distributions}\label{sec:discrete-distribution}

Prediction methods for discrete distributions are less well developed when compared to those for continuous distributions.
Many methods (e.g., the calibration-bootstrap method
proposed by \citealt{beran1990}) that generally work in continuous
settings are not applicable for certain discrete data models.
This section presents three prediction applications based on discrete distributions and
shows that the LR prediction method not only works for discrete
distributions but also has performance comparable to existing methods that were especially tailored to these particular discrete prediction problems.  Because the LR prediction method is a generally applicable method for prediction, the good performance of the method against specialized alternatives in these discrete cases is also suggestive that the LR approach may apply well in other prediction problems.

\subsection{Binomial Distribution}\label{subsec:binomial-dist}

We consider the prediction problem where there are two independent binomial
samples with the same probability $p$. The initial sample $X$ has a distribution
$\text{Binom}(n, p)$ and the predictand $Y$ has a distribution $\text{Binom}(m, p)$.
Both $n$ and $m$ are known, and note here that the data and predictand have related, though not identical, distributions (unlike predictions in Sections~\ref{sec:pivotal-quantity}-\ref{asymptotic-results} with continuous $Y$).
The goal is to construct a prediction interval for
$Y$ given observed data $X=x$.

Using the fact that the conditional distribution of $X$ given the sum $X+Y$ does
not depend on the parameter $p$, \citet{thatcher1964relationships} proposed a prediction
method based on the cdf of the hypergeometric distribution.
\citet{faulkenberry1973method} proposed a similar method using the conditional
distribution of $Y$ given the sum $X+Y$, which is also free of the parameter $p$.
\citet{nelsonapplied} proposed a different approach using the asymptotic normality of an approximate pivotal
statistic. However, numerical studies in \citet{wang2008coverage} and
\citet{krishnamoorthy2011improved} showed that Nelson's method has poor coverage probability, and proposed alternative prediction methods using asymptotic normality (e.g., based on inverting a score-like statistic instead of a Wald-like statistic).

To construct prediction intervals using the LR prediction method,
the reduced model is that $X$ and $Y$ have the same parameter $p$ while the
full model allows $X$ and $Y$ to have a different $p$ in the construction (\ref{eq:likelihood-ratio}).
The LR statistic is then
\[ \Lambda_{n,m}(x, y)=\frac{\text{dbinom}(x, n, \widehat{p}_{xy})\times\text{dbinom}(y, m, \widehat{p}_{xy})}
{\text{dbinom}(x, n, \widehat{p}_{x})\times\text{dbinom}(y, m, \widehat{p}_{y})}=\frac{(\widehat{p}_{xy})^{x+y}(1-\widehat{p}_{xy})^{n+m-x-y}}{(\widehat{p}_x)^{x}(1-\widehat{p}_x)^{n-x}(\widehat{p}_y)^{y}(1-\widehat{p}_y)^{m-y}},
\]
where $\text{dbinom}$ is the binomial pmf, $\widehat{p}_{x}=x/n$, $\widehat{p}_y=y/m$,
and $\widehat{p}_{xy}=(x+y)/(n+m)$. The asymptotic distribution of
the log-LR statistic is $-2\log\Lambda_{n,m}(X,Y)\xrightarrow{d}\chi^2_1$ as
$n\to\infty$ and $m\to\infty$; this theoretical result is explained further in Section~\ref{sec-theories-discrete} for discrete data.
The prediction region is defined as
\begin{equation}\label{eq:binomial-bounds}
	\mathcal{P}_{1-\alpha}(x)=\{y:-2\log\Lambda_{n,m}(x, y)\leq\chi_{1,1-\alpha}^2\},
\end{equation}
which gives an approximate $1-\alpha$ prediction interval procedure that has, asymptotically, equal-tail probabilities.

Due to the discrete nature of data here, we can further refine the LR prediction method by making a continuity
correction at the extreme values $x=0$ or $x=n$ and $y=0$ or $y=m$.
We first define $x^\prime\equiv x+0.5\text{I}_{x=0}-0.5\text{I}_{x=n}$
and $y^\prime\equiv y+0.5\text{I}_{y=0}-0.5\text{I}_{y=m}$ and further define
$\widehat{p}^\prime_{x}\equiv x^\prime/n$, $\widehat{p}^\prime_y\equiv y^\prime/m$,
and $\widehat{p}^\prime_{xy}\equiv(x^\prime+y^\prime)/(n+m)$. The corrected LR statistic is then
\begin{equation*}
	\Lambda_{n,m}^{\prime}(x,y)=\frac{(\widehat{p}_{xy}^\prime)^{x^\prime+y^\prime}(1-\widehat{p}_{xy}^\prime)^{n+m-x^\prime-y^\prime}}{(\widehat{p}_x^\prime)^{x^\prime}(1-\widehat{p}_x^\prime)^{n-x^\prime}(\widehat{p}_y^\prime)^{y^\prime}(1-\widehat{p}_y^\prime)^{m-y^\prime}}.
\end{equation*}

\begin{figure}[!ht] 
	\begin{subfigure}{0.55\textwidth}
		\includegraphics[width=\linewidth]{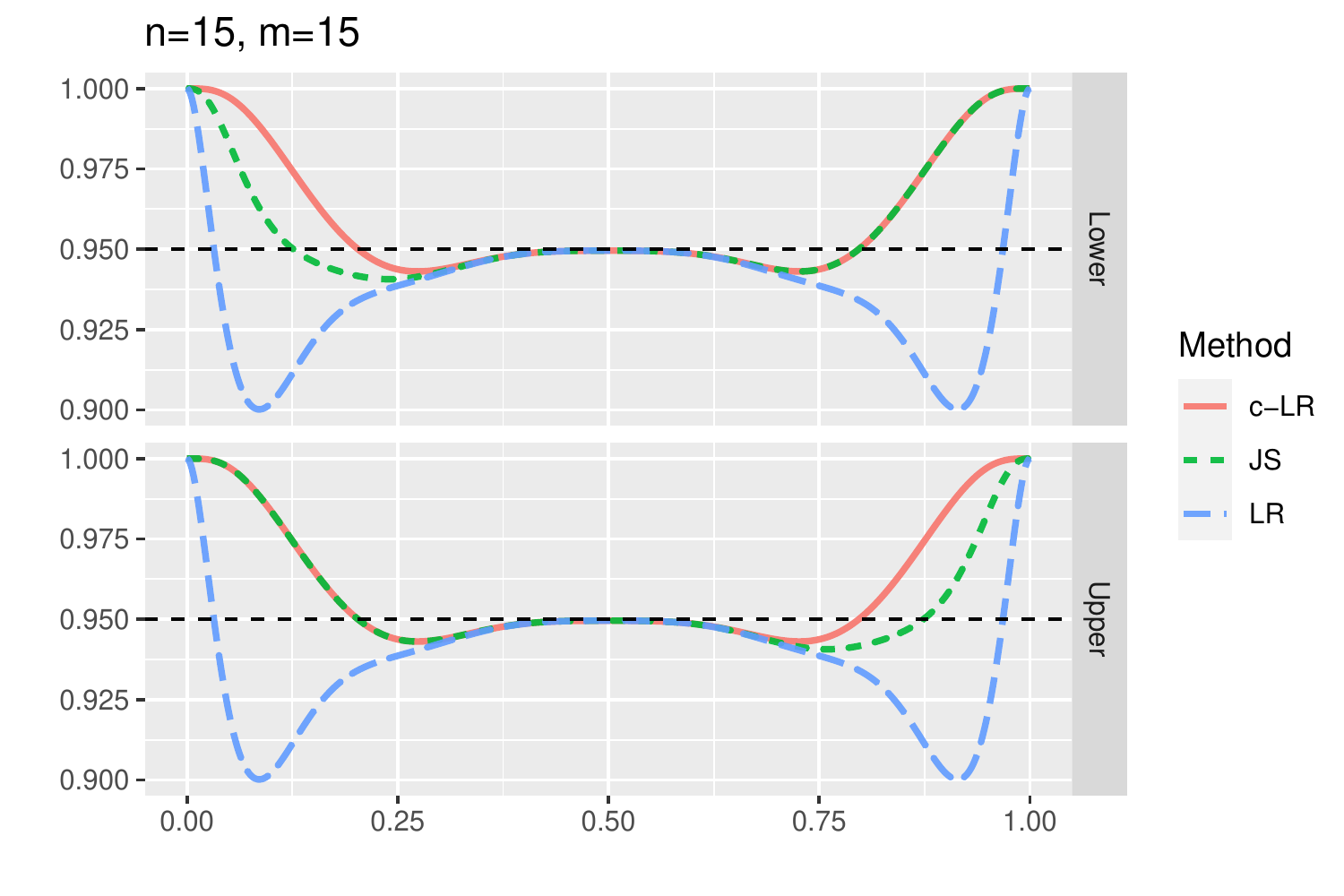}
	\end{subfigure}\hspace*{\fill}
	\begin{subfigure}{0.55\textwidth}
		\includegraphics[width=\linewidth]{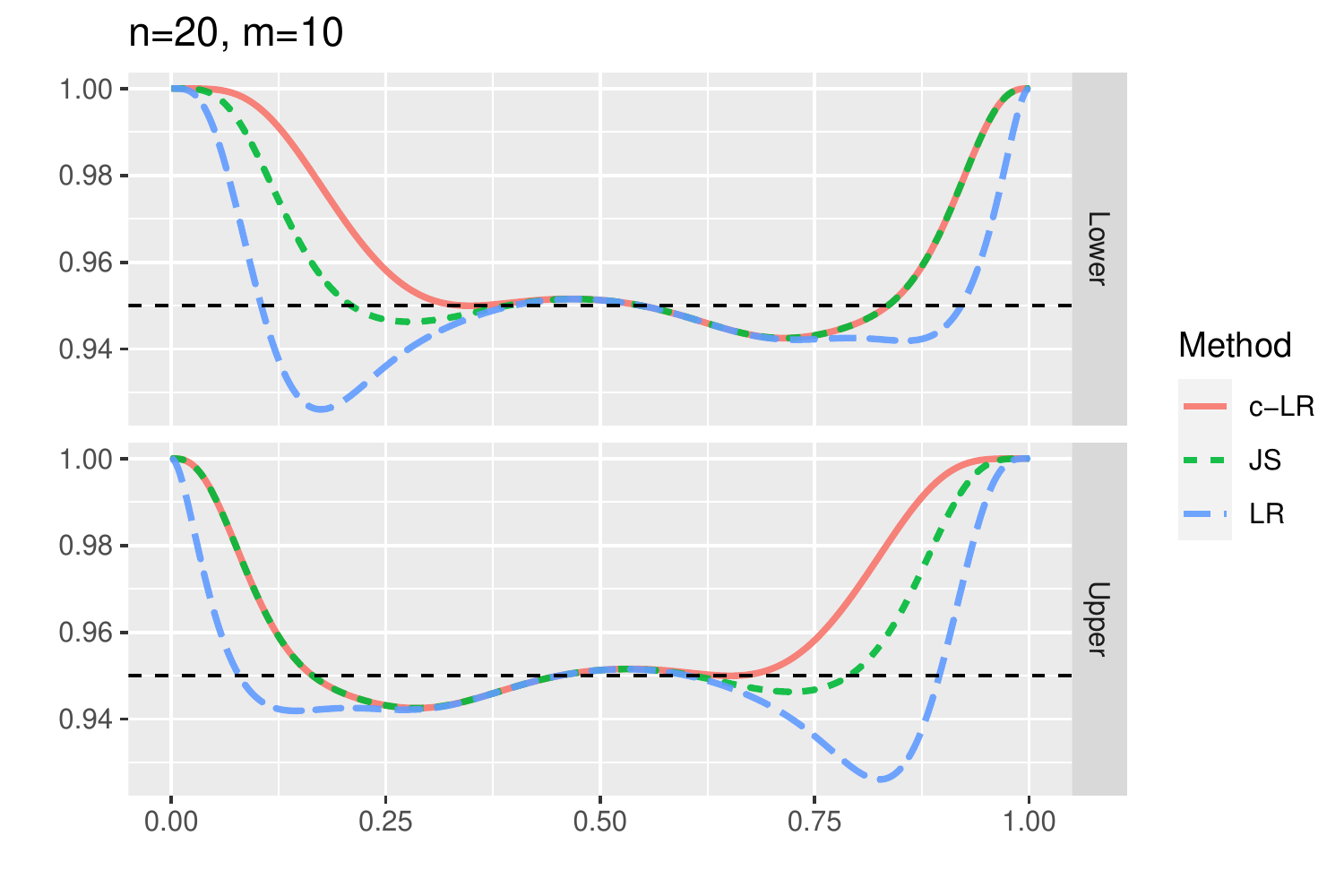}
	\end{subfigure}
	
	\medskip
	\begin{subfigure}{0.55\textwidth}
		\includegraphics[width=\linewidth]{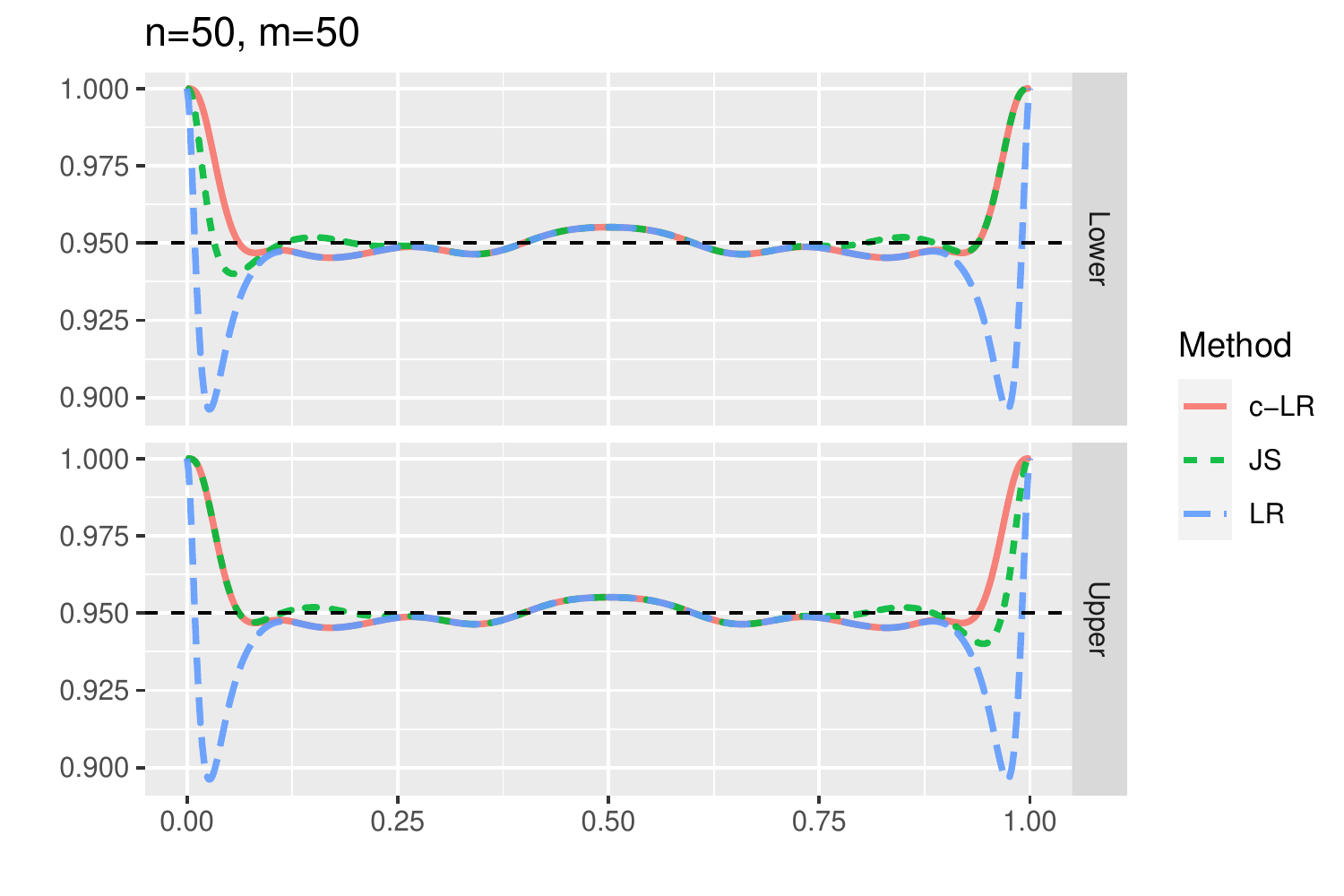}
	\end{subfigure}\hspace*{\fill}
	\begin{subfigure}{0.55\textwidth}
		\includegraphics[width=\linewidth]{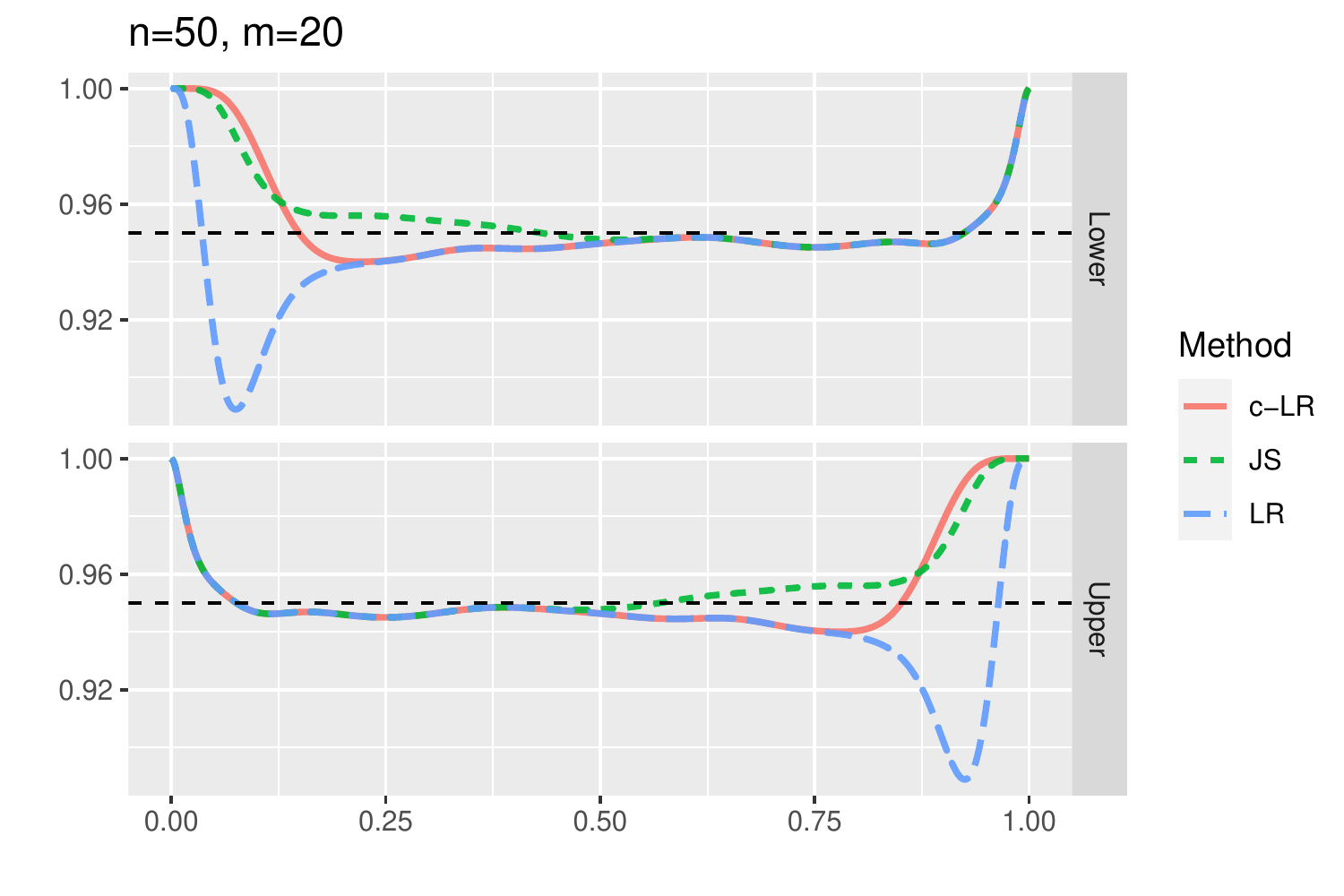}
	\end{subfigure}
	
	\medskip
	\begin{subfigure}{0.55\textwidth}
		\includegraphics[width=\linewidth]{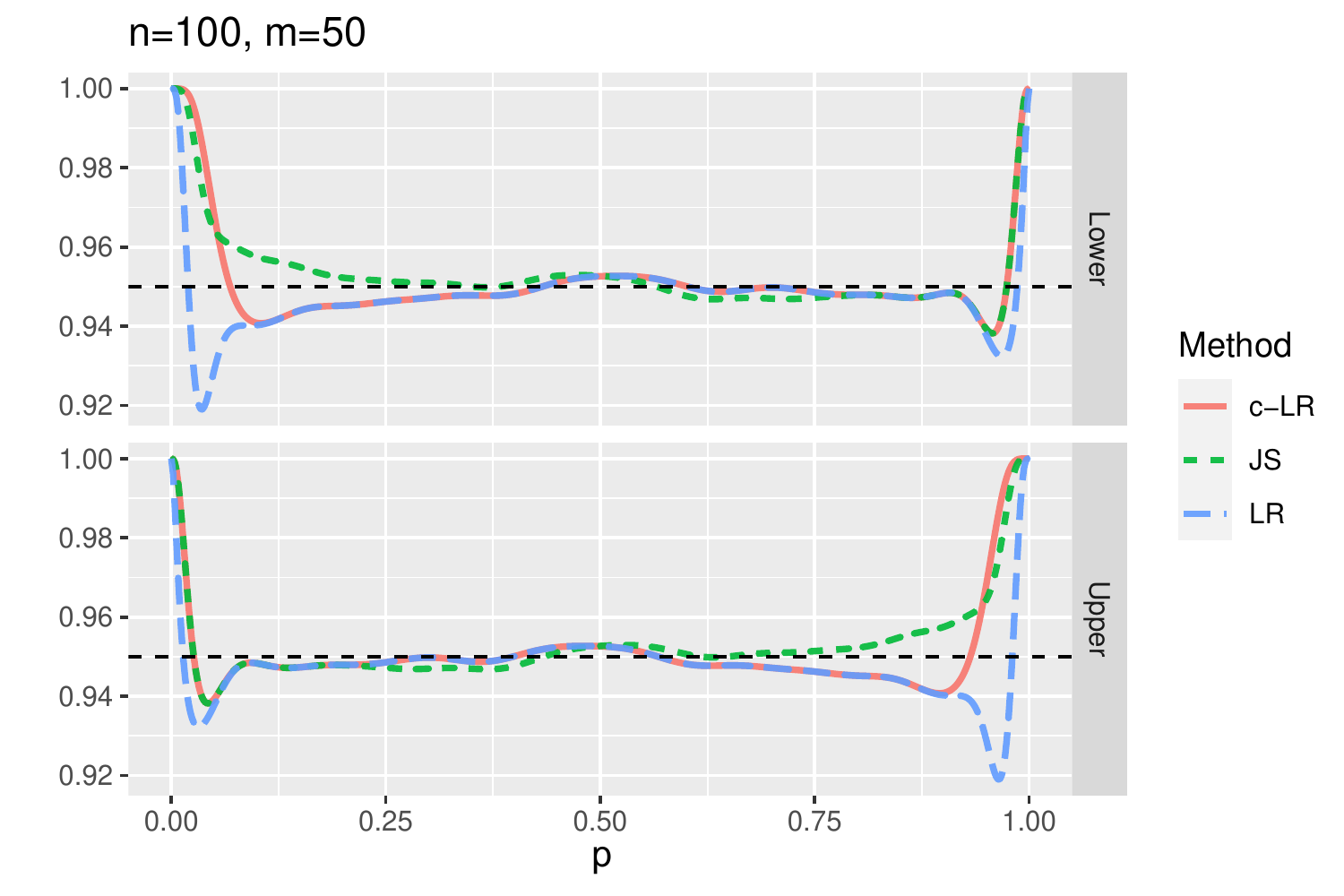}
	\end{subfigure}\hspace*{\fill}
	\begin{subfigure}{0.55\textwidth}
		\includegraphics[width=\linewidth]{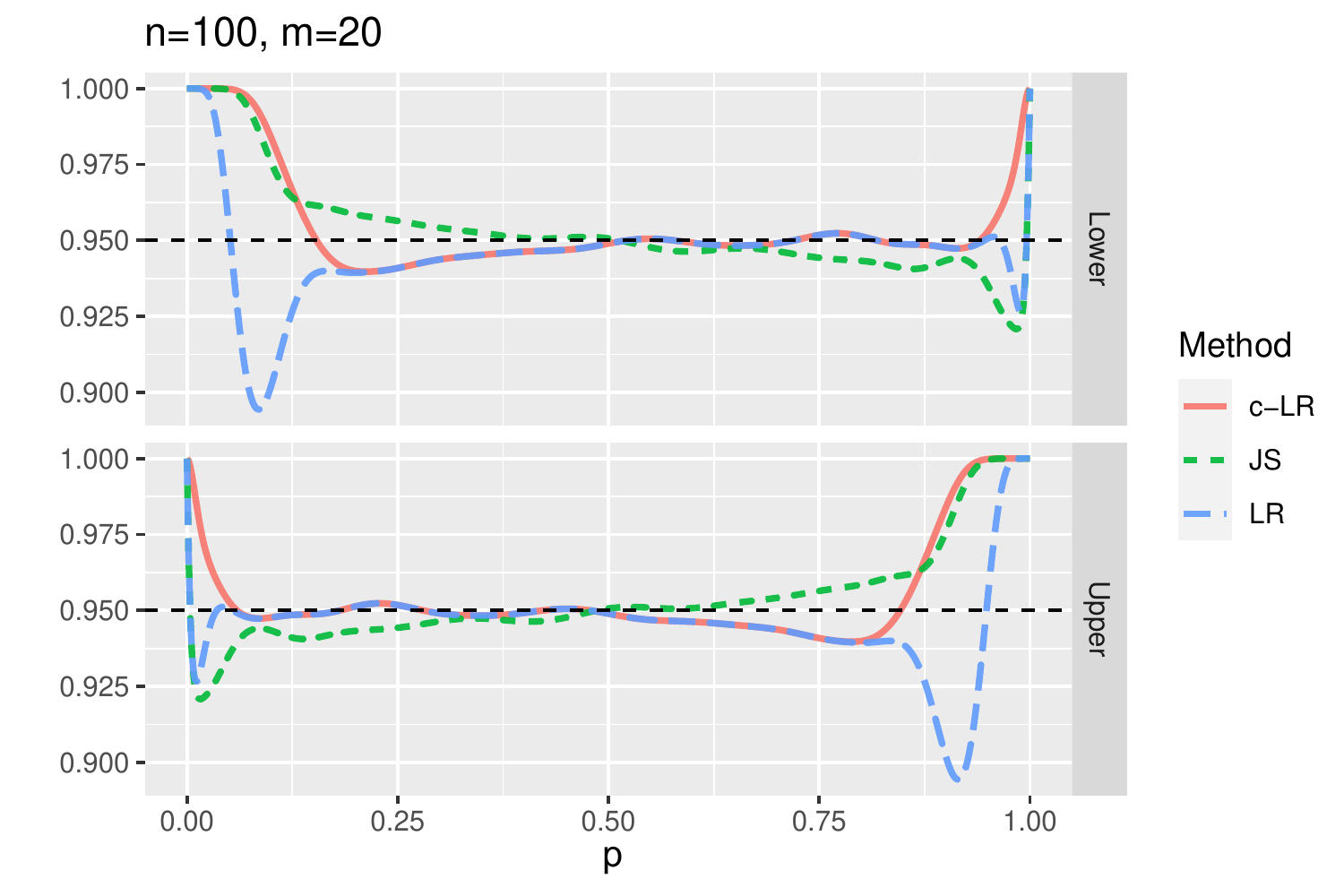}
	\end{subfigure}
	
	\caption{Coverage probabilities of 95\% lower and upper prediction bounds using corrected LR prediction method (c-LR),
		joint sample prediction method (JS), and LR prediction method (LR) as a function of $p$.}
	\label{fig:binomial}
\end{figure}

A numerical study was done to investigate the coverage probability of the LR prediction methods
and we also used the joint sampling prediction method as a benchmark for comparison because of its good
coverage probability (cf.~\citealt{krishnamoorthy2011improved}). The results in Figure~\ref{fig:binomial} show that
the original LR prediction method can have poor coverage for small sample sizes (e.g., $n=15$) when $p$ is near 0
or 1. However, with the continuity correction, the coverage probability of the corrected
LR prediction method is comparable to that of the joint sampling prediction method.
Unlike the joint sampling prediction method though, the LR prediction method is a general approach, which applies outside binomial prediction problems and has not been specifically designed for this purpose.
The numerical results here aim to provide evidence that the LR prediction method can be a generally effective procedure for prediction.

\subsection{Poisson Distribution}\label{subsec:poisson-dist}

Suppose $X\sim\text{Poi}(n\lambda)$ and $Y\sim\text{Poi}(m\lambda)$, where
$n$ and $m$ are known positive integers and $\lambda>0$ is unknown. The goal is to construct
prediction intervals for $Y$ based on data $X=x$. Similar to the binomial example, one can construct
prediction intervals using the fact that the conditional distribution of $X$ or $Y$
given $X+Y$ is binomial while \citet{nelsonapplied} and \citet{krishnamoorthy2011improved}
proposed alternative methods using a Wald-like approximate pivotal quantity.

To construct prediction intervals using the LR prediction method, the reduced model for the LR statistic (\ref{eq:likelihood-ratio}) is that
$X$ and $Y$ have the same $\lambda$ parameter while for the full model,
$X$ and $Y$ may not have the same $\lambda$ parameter. The LR statistic is given by
\[ \Lambda_{n,m}(x, y)=\frac{\text{dpois}(x, n\widehat{\lambda}_{xy})\times\text{dpois}(y, m\widehat{\lambda}_{xy})}
{\text{dpois}(x, n\widehat{\lambda}_x)\times\text{dpois}(y, m\widehat{\lambda}_y)}=\frac{\exp\left[-(n+m)\widehat{\lambda}_{xy}\right](n\widehat{\lambda}_{xy})^{x}(m\widehat{\lambda}_{xy})^{y}}{\exp\left(-n\widehat{\lambda}_x-m\widehat{\lambda}_y\right)(n\widehat{\lambda}_x)^{x}(m\widehat{\lambda}_y)^{y}}, \]
where $\widehat{\lambda}_{xy}=(x+y)/(n+m)$, $\widehat{\lambda}_x=x/n$, $\widehat{\lambda}_y=y/m$, and $\text{dpois}$ is the Poisson pmf.
The prediction interval can be obtained using
the same procedure in (\ref{eq:binomial-bounds}); see Section~\ref{sec-theories-discrete} for justification.
We can also refine the LR prediction method with a continuity
correction at the extremes $x=0$ or $y=0$ by letting $x^\prime\equiv x+0.5\text{I}_{x=0}$ and $y^\prime\equiv y+0.5\text{I}_{y=0}$.
Then define $\widehat{\lambda}_{xy}^\prime\equiv (x^\prime+y^\prime)/(n+m)$,
$\widehat{\lambda}^\prime_x\equiv x^\prime/n$, and $\widehat{\lambda}_y^\prime\equiv y^\prime/m$ so that the corrected LR statistic is
\begin{equation*}
	\Lambda_{n,m}^\prime(x,y)=\frac{\exp\left[-(n+m)\widehat{\lambda}_{xy}^\prime\right](n\widehat{\lambda}_{xy}^\prime)^{x^\prime}(m\widehat{\lambda}_{xy}^\prime)^{y^\prime}}{\exp\left(-n\widehat{\lambda}_x^\prime-m\widehat{\lambda}_y^\prime\right)(n\widehat{\lambda}_x^\prime)^{x^\prime}(m\widehat{\lambda}_y^\prime)^{y^\prime}}.
\end{equation*}

A numerical study was done to investigate the coverage probability of
the proposed methods. Similar to the binomial example, the joint sampling
method from \citet{krishnamoorthy2011improved} was used for comparison
because of its good coverage properties. Figure~\ref{fig:poisson} shows that
the continuity correction improves the poor coverage of the LR prediction method when $\lambda$ is small. The coverage
probability of the corrected LR prediction method is comparable to
that of the joint sampling method. In the bottom-right subplot of Figure~\ref{fig:poisson},
the corrected method has better performance than the joint sampling method.
Again, unlike the joint sampling prediction method, the LR prediction method is general and not specifically designed for Poisson predictions.

\begin{figure}[ht!] 
	\begin{subfigure}{0.55\textwidth}
		\includegraphics[width=\linewidth]{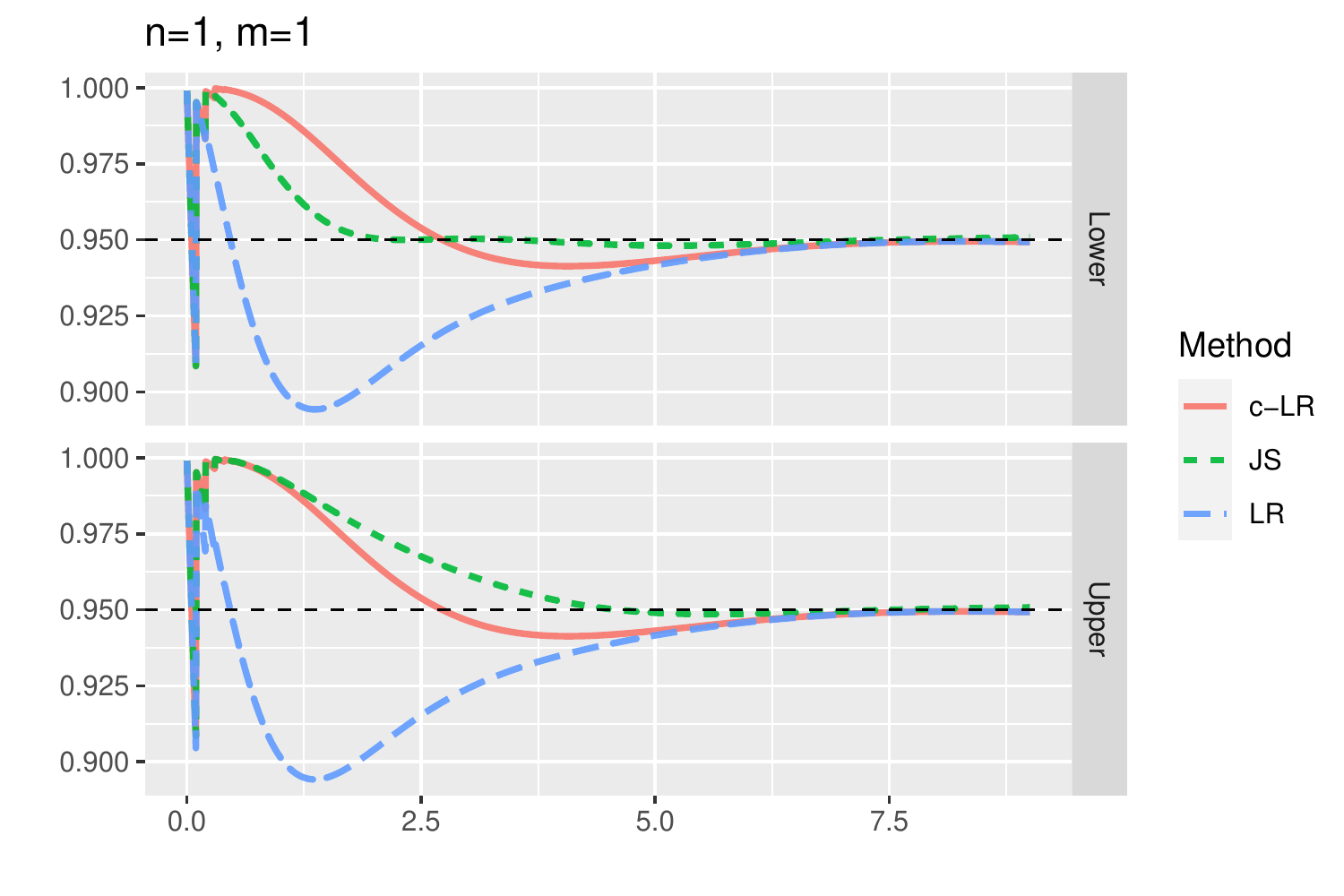}
	\end{subfigure}\hspace*{\fill}
	\begin{subfigure}{0.55\textwidth}
		\includegraphics[width=\linewidth]{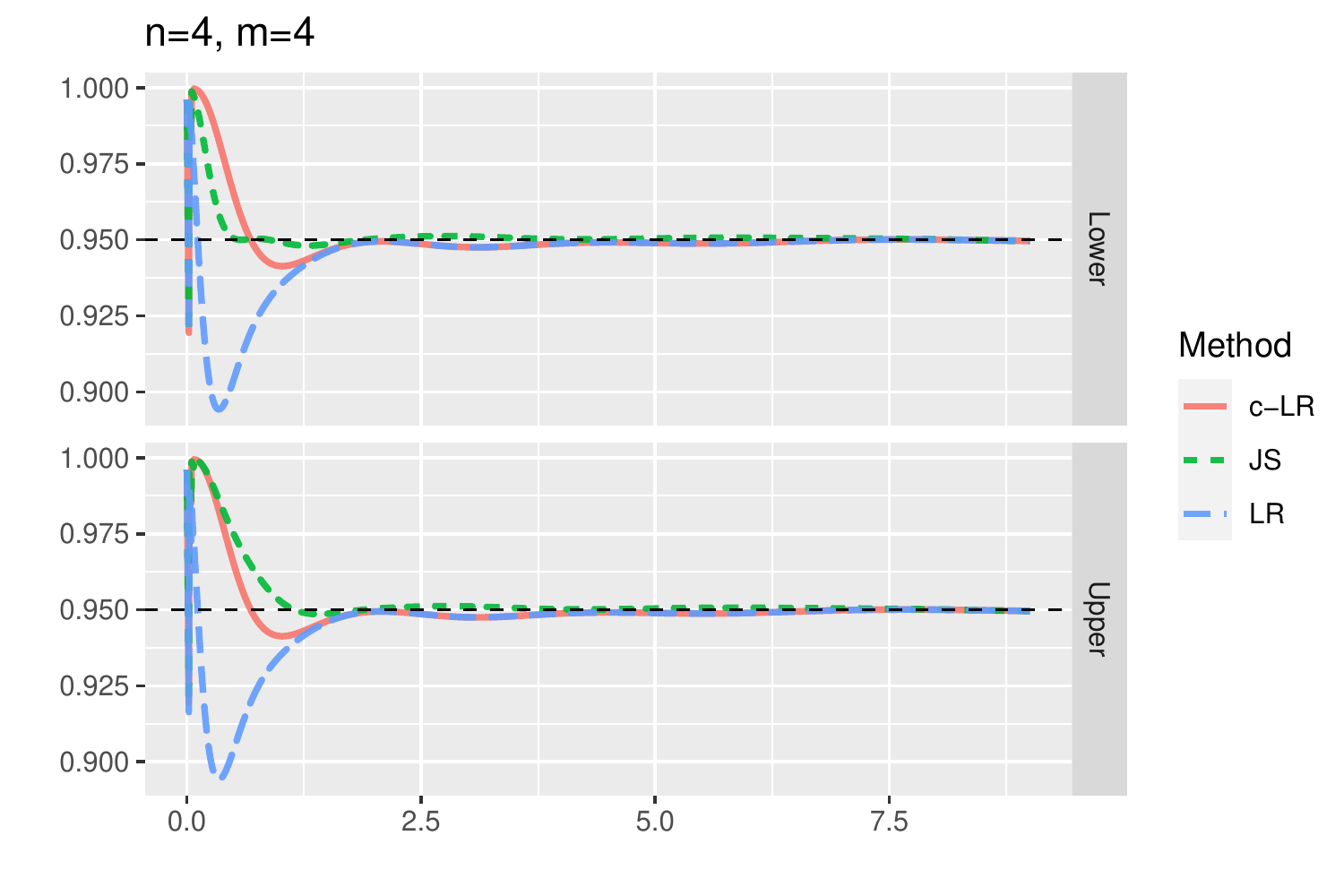}
	\end{subfigure}
	
	\medskip
	\begin{subfigure}{0.55\textwidth}
		\includegraphics[width=\linewidth]{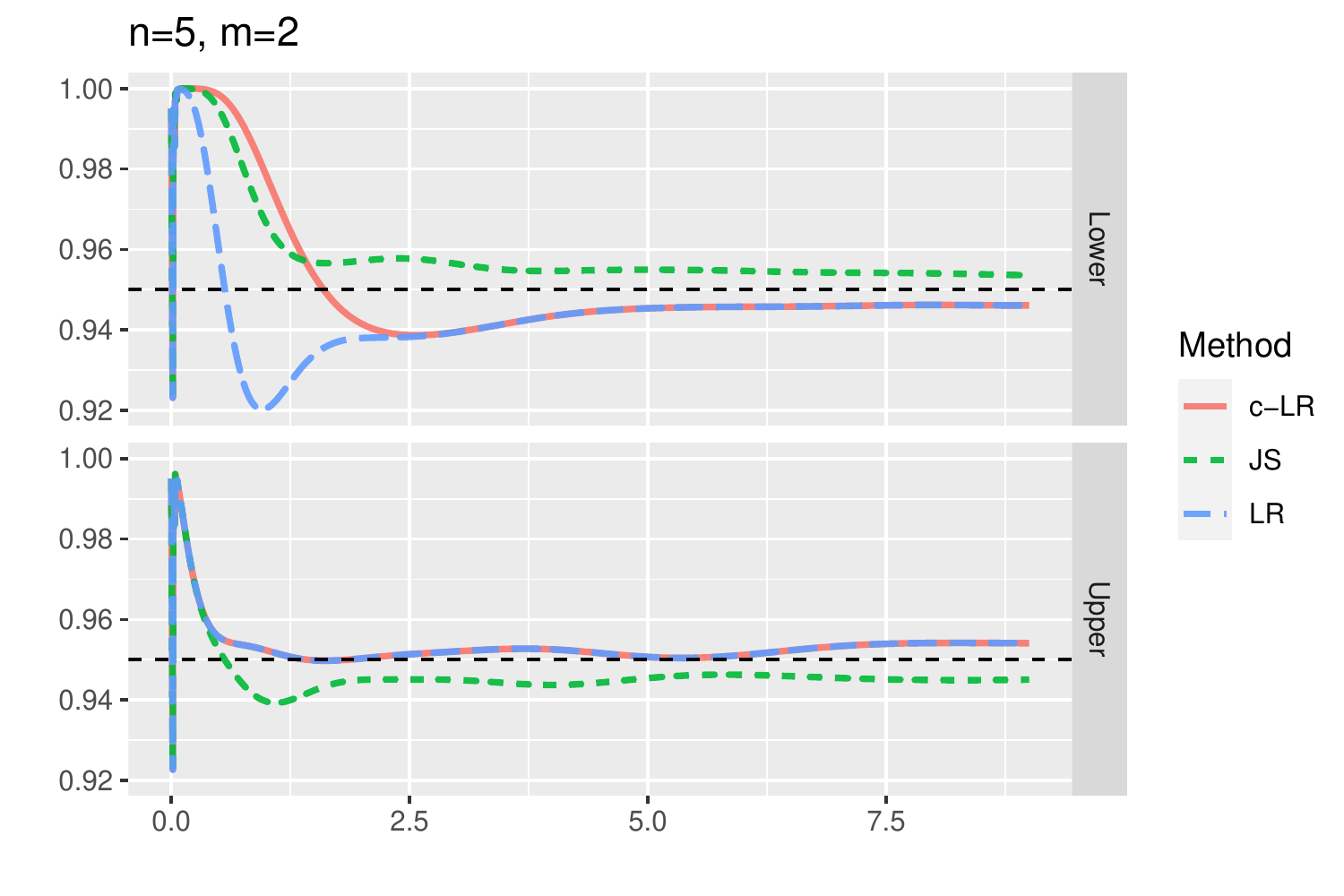}
	\end{subfigure}\hspace*{\fill}
	\begin{subfigure}{0.55\textwidth}
		\includegraphics[width=\linewidth]{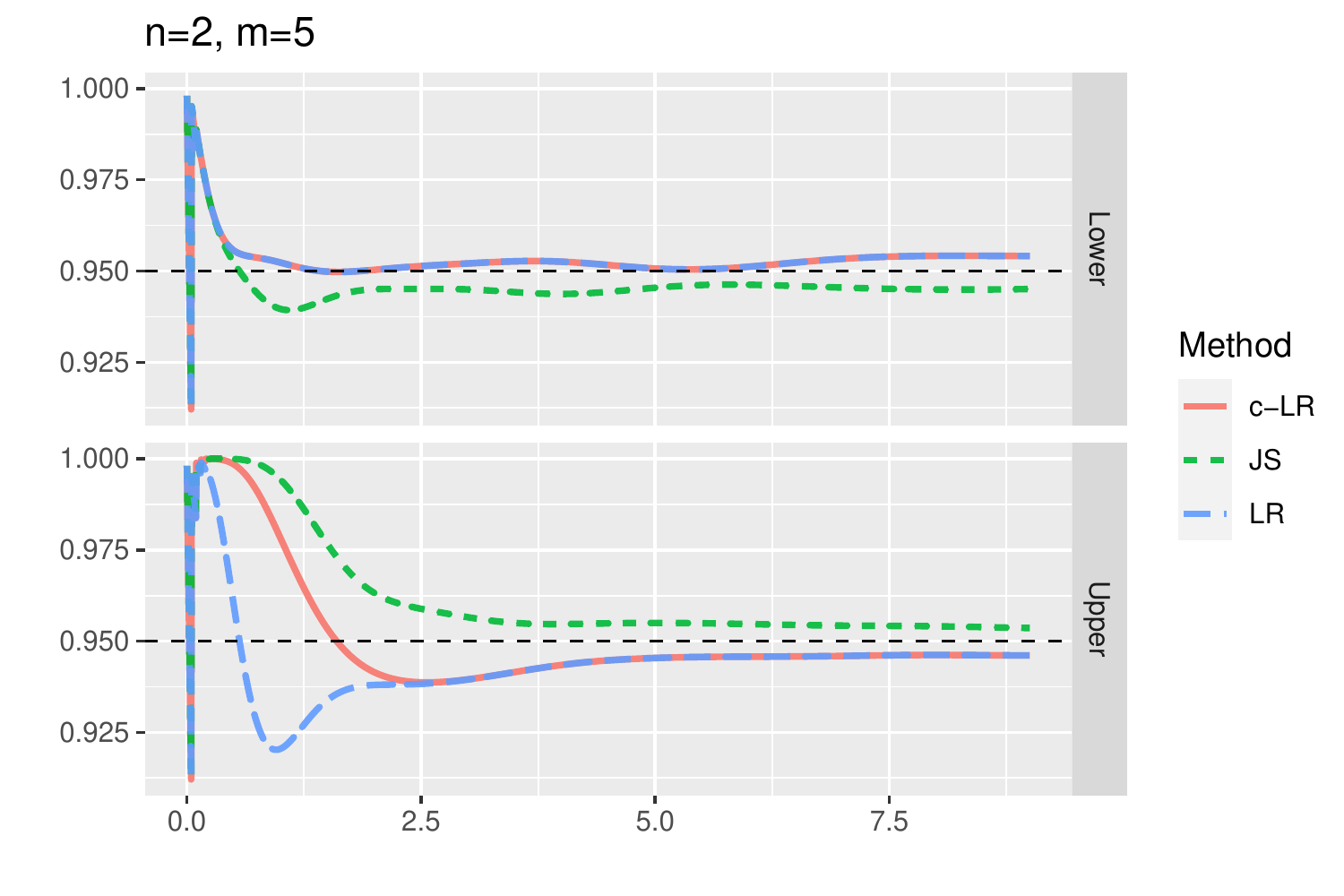}
	\end{subfigure}
	
	\medskip
	\begin{subfigure}{0.55\textwidth}
		\includegraphics[width=\linewidth]{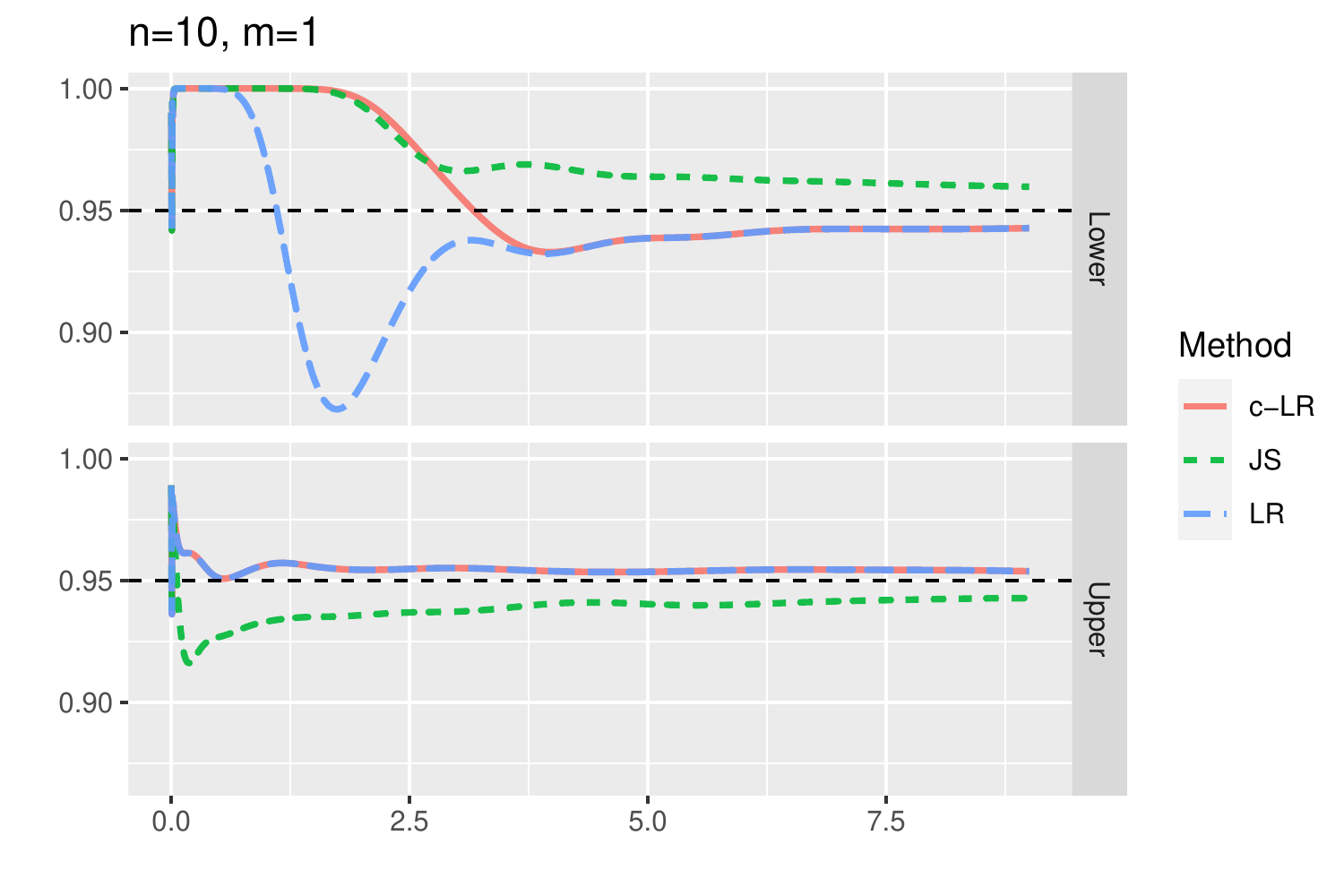}
	\end{subfigure}\hspace*{\fill}
	\begin{subfigure}{0.55\textwidth}
		\includegraphics[width=\linewidth]{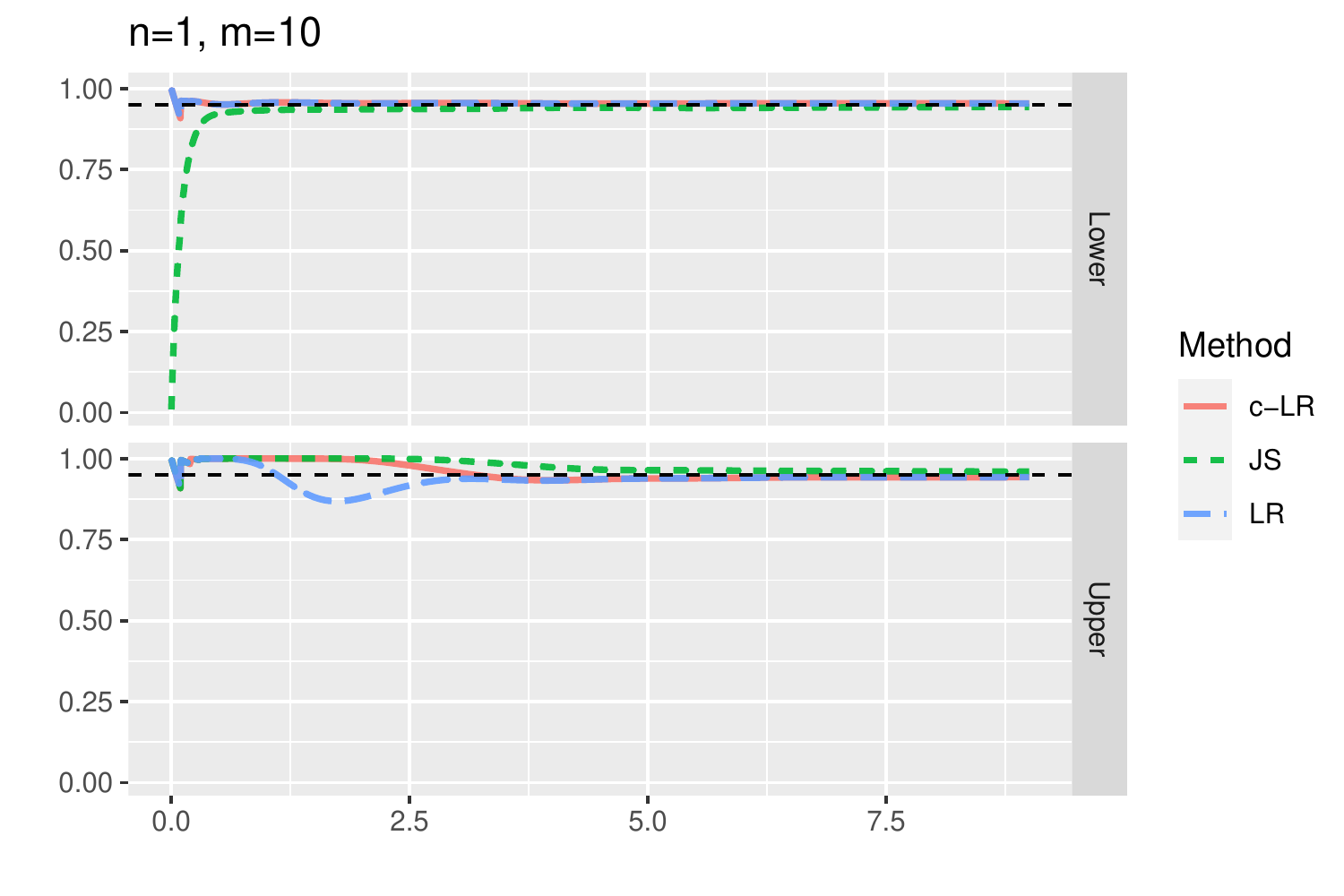}
	\end{subfigure}
	
	\caption{Coverage probabilities of 95\% Poisson lower and upper prediction bounds using the corrected LR prediction method (c-LR),
		joint sample prediction method (JS), and LR prediction method (LR) as a function of $\lambda$.}
	\label{fig:poisson}
\end{figure}

\subsection{Predicting the Number of Future Events}

\label{subsec:within-sample-prediction}

Suppose $n$ units start service at time $t=0$ and that the lifetime of each unit has a continuous parametric distribution with cdf $F(t;\btheta)$ and density $f(t;\btheta)$.
At a data freeze date, the unfailed units have accrued $t_c$ time units of service (e.g., hours or months in service) while $r_n$ failures have occurred and the failure
times (all less than $t_c$) are known. A prediction interval
for the number of failures that will occur in the interval $(t_c, t_w]$ $(t_w>t_c,)$ is required.
This problem is called within sample prediction because the predictand
and the observed Type-I censored data are from the same sample.
The within-sample prediction and related problems have been studied in \citet{elawqm1999} using a calibration method.
Similar problems have been studied in \citet{nelson2000} and \citet{nordman2002weibull} based on an LR statistic without calibration.
\citet{tian2020pred} showed that the simple plug-in method (where ML estimates replace the unknown parameters in the distribution of the predictand and the $\alpha/2$ and $1-\alpha/2$ quantiles of the resulting distribution define an approximate $1-\alpha$ prediction interval procedure) is not asymptotically correct and proposed three alternative methods, based on parametric bootstrap samples, that are asymptotically correct.
In this paper, we propose another solution based on an LR statistic,
that does not require bootstrap samples.

Suppose that a random sample $T_1, \dots, T_n\sim F(t;\btheta)$ is observed under
Type-I censoring with $r_n=\sum_{i=1}^{n}\text{I}(T_i\leq t_c)$ censored
units (failures).
The predictand is the number $Y=\sum_{i=1}^{n}\text{I}(t_c\leq
T_i\leq t_w)$ of events occurring in the future interval $(t_c,t_w]$.
For the $n-r_n$ units surviving at $t_c$,
the conditional probability of each unit to fail in $(t_c, t_w]$, given that the unit survived to $t_c$, is given by
\begin{equation}\label{eq:conditional-prob}
	p\equiv \Pr\left(t_c<T_1\leq t_w|T_1>t_c\right).
\end{equation}

\subsubsection{Implementing the LR Prediction Method}
To implement the LR prediction method, we specify a reduced model versus~full model comparison in order to construct an LR statistic analogous to (\ref{eq:likelihood-ratio}).
Such models will be formulated in terms of the value (\ref{eq:conditional-prob}) of the conditional probability $p$ for the interval $(t_c,t_w]$, recalling that the predictand $Y$ is the number of failures (out of $n-r_n$ possible) that will occur in this interval.
For the reduced model, we assume that the time-to-failure process is
governed by $F(t;\btheta)$ in the interval $(0, t_w]$ and that the conditional probability (\ref{eq:conditional-prob}) of a failure in $(t_c,t_w]$ is
\[
p=\frac{F(t_w;\btheta)-F(t_c;\btheta)}{1-F(t_c;\btheta)}.
\]
The likelihood function for the reduced model is
\begin{equation}\label{eq:reduced-lik}
	\mathcal{L}_1(\btheta;\boldsymbol{t}_n, y)=\binom{n-r_n}{y}\prod_{i=1}^{r}f(t_{(i)};\btheta)
	\left[F(t_w;\btheta)-F(t_c;\btheta)\right]^{y}\left[1-F(t_w;\btheta)\right]^{n-y-r_n}.
\end{equation}
For the full model, $F(t;\btheta)$ will still be the time-to-failure distribution in the interval $(0, t_c]$ but not $(t_c, t_w]$, so that the value (\ref{eq:conditional-prob}) of the conditional probability $p\in(0,1]$ becomes one additional parameter.
The likelihood function for the full model is
\begin{equation}\label{eq:full-lik}
	\mathcal{L}_2(\btheta, p;\boldsymbol{t}_n, y)=\binom{n-r_n}{y}\prod_{i=1}^{r}f(t_{(i)};\btheta)
	p^y(1-p)^{n-y-r_n}.
\end{equation}
By maximizing the likelihood functions in (\ref{eq:reduced-lik}) and
(\ref{eq:full-lik}), the LR statistic is
\[ \Lambda_n(\boldsymbol{t}_n, y)=\frac{\sup_{\btheta}L_1(\btheta;\boldsymbol{t}_n, y)}{\sup_{\btheta, p}L_2(\btheta,p;\boldsymbol{t}_n, y)}. \]
The asymptotic (as $n\to\infty$) distribution of $-2\log\Lambda_n(\boldsymbol{T}_n, Y)$ is $\chi_1^{2}$,
because the full model has one more parameter than
the reduced model and standard regularity conditions hold (see also Section~\ref{sec-theories-discrete}). An approximate $1-\alpha$ prediction region is defined as
\begin{equation}\label{eq:within-sample-prediction-set}
	\{y:-2\log\Lambda_n(\boldsymbol{t}_n, y)\leq\chi_{1,1-\alpha}^2\},
\end{equation}
where $\chi_{1,1-\alpha}^2$ is the $1-\alpha$ quantile of the $\chi_1^2$ distribution.
Because $\Lambda_n(\boldsymbol{t}_n, y)$ is a unimodal function of $y$,
the prediction region in (\ref{eq:within-sample-prediction-set}) provides the desired approximate prediction interval.

\subsubsection{A Simulation Study}
A simulation study was done to examine the coverage probability of the
LR prediction method for the within-sample prediction problem.
We simulated Type-I censored data with censoring time $t_c$ using the Weibull distribution
\[
F(t;\beta,\eta)=1-\exp\left[-\left(\frac{t}{\eta}\right)^\beta\right],\quad t>0.
\]
Then we constructed
prediction intervals for the number of failures in the future time interval $(t_c, t_w]$ using
several methods: plug-in, LR, direct-bootstrap, GPQ-bootstrap, and calibration-bootstrap methods.
As mentioned earlier, the plug-in method, which replaces the unknown parameter $\btheta=(\beta,\eta)$ with a consistent
estimate $\htheta$, fails to provide asymptotically correct prediction intervals (cf.~\citealt{tian2020pred}).
The last three methods are from \citet{tian2020pred} and have been established to be asymptotically correct.
\begin{figure}[t!]
	\centering
	\includegraphics[width=\textwidth]{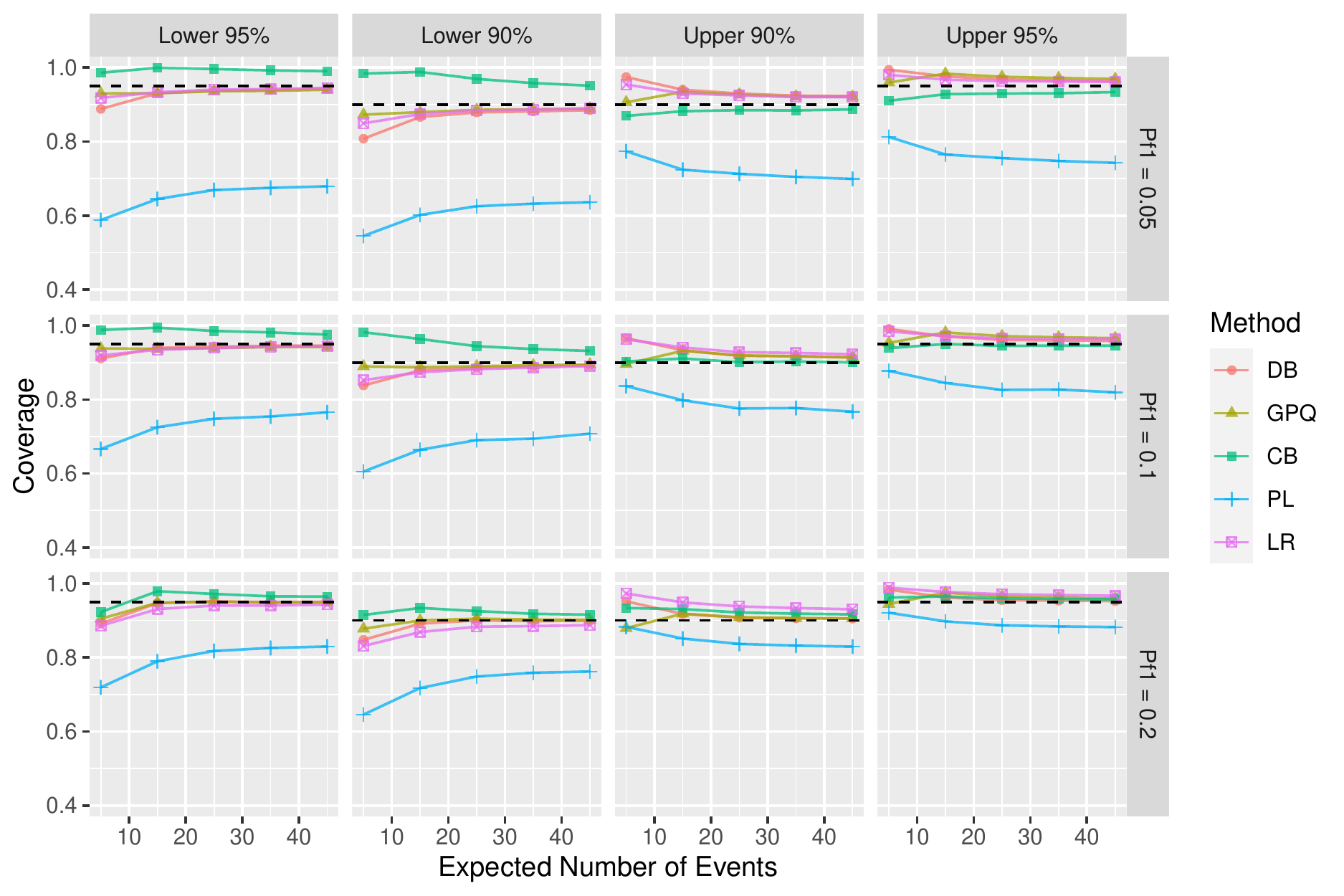}
	\caption{Coverage probabilities versus expected number of events (failures) for the direct-bootstrap (DB), GPQ-bootstrap (GPQ), calibration-bootstrap (CB), LR, and plug-in (PL) methods when $d=0.1$ and $\beta=2$.}
	\label{fig:within-sample-pred}
\end{figure}
The factors for this simulation study include
\begin{enumerate}
	\item The probability that a unit fails before the censoring time $t_c$:
	$p_{f1}=F(t_c;\beta,\eta)$.
	\item The expected number of failures at the censoring time $t_c$: $\text{E}(r)=np_{f1}$.
	\item The probability of a unit fails in the future time interval $(t_c, t_w]$: $d\equiv p_{f2}-p_{f1}$,
	where $p_{f2}=F(t_w;\beta,\eta)$.
	\item The Weibull shape parameter: $\beta$.
\end{enumerate}

We set the Weibull scale parameter as $\eta=1$ and, for other factors, we use the following
factor levels: (i) $p_{f1}=0.05, 0.1, 0.2$; (ii) $\text{E(r)}=5, 15, 25, 35, 45$; (iii) $d=0.1, 0.2$; (iv) $\beta=0.8, 1, 2, 4$. For the methods which involve bootstrap simulation, the bootstrap sample size is $B=5000$. The unconditional coverage probability is computed by averaging $N=5000$ conditional coverage probabilities
(i.e., the Monte Carlo sample size is $N=5000$).

Figure~\ref{fig:within-sample-pred} compares the coverage probabilities for the plug-in, direct-bootstrap, GPQ-bootstrap, calibration-bootstrap, and LR prediction methods when $d=0.1$ and $\beta=2$.
We can see that the LR, direct-bootstrap, and GPQ-bootstrap prediction method have similar coverage probabilities for within-sample prediction, where the latter two methods rely on bootstrap and the LR interval does not.
That is, the LR prediction method based on chi-square calibration has the advantage of being computationally easier than the direct-bootstrap or GPQ-bootstrap methods for this prediction problem, while providing comparable performance.
This pattern is consistent in the simulation results of other factor combinations (given in the online supplementary material).
While we have considered the LR prediction method for within-sample prediction for illustration and comparison, the LR prediction method is again general and not specific to within-sample prediction.

\subsection{Validating the Asymptotic Distribution}
\label{sec-theories-discrete}
In Sections~\ref{subsec:binomial-dist}--\ref{subsec:within-sample-prediction}, we construct the prediction intervals for certain discrete predictands $Y$ using the fact that the log-LR statistic has a chi-square limit with 1 degree of freedom in these prediction problems.
This section provides justification for these asymptotic results.

The prediction problems in Sections~\ref{subsec:binomial-dist} and \ref{subsec:poisson-dist} are similar in that the predictand $Y$ (as a $\text{Binom}(m,p)$ or $\text{Pois}(m\lambda)$ random variable) can be seen to have the same distribution as a sum of iid variables in both cases (i.e., $m$ iid $\text{Bern}(p)$ or $\text{Pois}(\lambda)$ random variables).
As a consequence, the log-LR statistic from Section~6.1, constructed on the basis of using $X\sim\mbox{Binom}(n,p)$ to predict
$Y\sim\mbox{Binom}(m,p)$, is the same as the log-LR statistic given in Theorem~3 based on the $X_1,\ldots,X_n$
and $Y_1,\ldots,Y_m$ being iid $\mbox{Binom}(1,p)$.
A similar statement holds for the Poisson prediction problem from Section~6.2.
Hence, the chi-square limit for the log-LR statistic in Sections~\ref{subsec:binomial-dist} and \ref{subsec:poisson-dist} follows from Theorem~\ref{theorem-discrete} below.
We provide Theorem~\ref{theorem-discrete} as a general result with standard regularity conditions given in the supplementary material.
For the prediction problem in Section~\ref{subsec:within-sample-prediction}, the proof is similar to that of Theorem~\ref{theorem-discrete}.
See Section A.3 of the online supplementary material for details.
\begin{theorem}\label{theorem-discrete}
	Suppose $X_1,\dots,X_n$ are iid random variables with common density $f(\cdot;\theta_1)$ and, independently, $Y_1,\dots,Y_m$ are iid random variables with a common density $f(\cdot;\theta_2)$, where $\theta_1,\theta_2\in \Theta$ denote real-valued parameters.
	Suppose further that mild regularity conditions hold (as described in Section~A.2 of the supplement).
	Then, if $\theta_1=\theta_2$, the log-LR statistic for testing $\theta_1=\theta_2$ has a limiting chi-square distribution with 1 degree of freedom as $n,m\to\infty$; that is,
	\[
	-2\log\left\{\frac{\sup_{\theta}\left[\prod_{i=1}^{n}f(x_i;\theta)\prod_{j=1}^{m}f(y_i;\theta)\right]}{\left[\sup_{\theta_1}\prod_{i=1}^{n}f(x_i;\theta_1)\right]\left[\sup_{\theta_2}\prod_{i=1}^{m}f(y_i;\theta_2)\right]}\right\}\stackrel{d}{\rightarrow} \chi_1^2.
	\]
\end{theorem}

\section{Comparison with the Predictive Likelihood Methods}\label{sec-relationship-pred-dist}

The predictive likelihood method, introduced in Section~\ref{subsec:literature}, is an important prediction method. While having similar-sounding names, the LR prediction method for prediction is different than the predictive likelihood method. The LR prediction method
may be classified as a type of test-based method (cf.~Section~\ref{subsec:literature}) for prediction intervals
which also share connections to approximate pivotal quantities (though technically, the LR statistic may not always be pivotal, even asymptotically, as shown in Section~\ref{asymptotic-results}, although its limiting distribution may then be estimated by bootstrap).
This section describes two specific types of predictive likelihood methods.
However, these predictive likelihood methods can fail to provide desirable prediction intervals in some prediction problems, where the LR prediction method emerges as having better properties.

\subsection{Profile Predictive Likelihood Method}
The profile predictive likelihood $\widetilde{\mathcal{L}}_p(\boldsymbol{x}_n,y)$ function for $y$ given data values $\boldsymbol{X}_n=\boldsymbol{x}_n$ is obtained by maximizing out the parameters in the joint likelihood function,
\[
\widetilde{\mathcal{L}}_p(\boldsymbol{x}_n,y)\equiv\sup_{\btheta}f(y;\btheta)\prod_{i=1}^{n}f(x_i;\btheta).
\]
Then, the predictive likelihood is normalized to give a predictive density function for $Y$,
\[
f_p(y;\boldsymbol{x}_n)=\frac{\widetilde{\mathcal{L}}_p(\boldsymbol{x}_n,y)}{\int_{-\infty}^{\infty}\widetilde{\mathcal{L}}_p(\boldsymbol{x}_n,y)dy},
\]
which is viewed as univariate distribution depending on $\boldsymbol{X}_n=\boldsymbol{x}_n$ for calibrating prediction intervals for $Y$.
Note that $\widetilde{\mathcal{L}}(\boldsymbol{x}_n,y)$ is the numerator of the LR statistic in (\ref{eq:likelihood-ratio}) so that the process of obtaining the profile
predictive likelihood may be viewed as a step in constructing LR-based prediction intervals.
However, in some prediction problems, discussed next, the profile predictive likelihood does not lead to an exact prediction interval for the predictand $Y$ when the LR prediction method does.

To illustrate this, consider a sample $\boldsymbol{X}_n$ from a normal distribution, and consider constructing prediction intervals for a future random variable $Y$ from the same distribution. From \citet{lejeune1982simple}, the profile predictive likelihood for $Y$ given data $\boldsymbol{X}_n=\boldsymbol{x}_n$ (i.e., the distribution to be used for predicting $Y$, as implied by the profile predictive likelihood density) is given by the distribution of
\begin{equation*}
	\bar{x}_n+s\sqrt{\frac{n^2-1}{n^2}}T,
\end{equation*}
where $\bar{x}_n$ is the sample mean, $s^2$ is the sample variance, and $T$ is an independent random variable having a $t$-distribution with $n$ degrees of freedom. However, in order
for the profile predictive likelihood method to produce an exact prediction interval for $Y$, the degrees of freedom for the $t$-distribution of $T$ above should be $n-1$ instead of $n$ (see (\ref{eq:normal-invert-t-test})).
Consequently, the profile predictive likelihood method is not exact in this example.
The LR prediction method, however, has exact coverage for this prediction problem, as shown in Section~\ref{sec:pivotal-quantity}.

\subsection{Approximate Predictive Likelihood Method}

\citet{davison1986} proposed an approximate predictive likelihood method that involves maximizing likelihood functions.
Let $\widehat{\boldsymbol{\theta}}$ be the maximizer of $\mathcal{L}(\btheta;\boldsymbol{x}_n)$, which is
the likelihood function for data $\boldsymbol{x}_n$ alone and $\widehat{\boldsymbol{\theta}}_y$ be the maximizer of the joint likelihood
function for $\boldsymbol{X}_n$ and $Y$, say $\mathcal{L}(\btheta;\boldsymbol{x}_n,y)$.
Then the approximate predictive likelihood is defined as
\begin{equation*}
	\widetilde{\mathcal{L}}(\boldsymbol{x}_n,y)=\frac{\mathcal{L}(\widehat{\boldsymbol{\theta}}_y;\boldsymbol{x}_n,y)|J_1(\widehat{\boldsymbol{\theta}})|^{1/2}}{\mathcal{L}(\widehat{\boldsymbol{\theta}};\boldsymbol{x}_n)|J_2(\widehat{\boldsymbol{\theta}}_y)|^{1/2}},
\end{equation*}
where $J_1(\btheta)$ is the minus Hessian of $\log\mathcal{L}(\btheta;\boldsymbol{x}_n)$, $J_2(\btheta)$ is the minus
Hessian of $\log\mathcal{L}(\btheta;\boldsymbol{x}_n,y)$,
and $|\cdot|$ is the determinant.

Suppose that $\boldsymbol{X}_n$ and $Y$ are mutually independent with a
common exponential distribution. From \citet{davison1986}, the approximate
predictive likelihood for $Y$ is
$$
\widetilde{\mathcal{L}}(\boldsymbol{x}_n,y)\propto\left(\sum_{i=1}^{n}x_i\right)^{n-1}\left(\sum_{i=1}^{n}x_i+y\right)^{-n}.
$$
Then, prediction intervals for $Y$ are computed from density on $y\in(0,\infty)$, which is obtained by normalizing $\widetilde{\mathcal{L}}(\boldsymbol{x}_n,y)$ with respect to $y$.
Moreover, as noted by \citet{hall1999}, the approximate predictive likelihood method is not exact here and has a coverage probability error of order $O(1/n)$.
For the LR prediction method, however, the LR statistic (\ref{eq:likelihood-ratio}) is
$$
\Lambda_n(\boldsymbol{x}_n,y)=\left(\frac{n\bar{x}_n+y}{\bar{x}_n}\right)^n\frac{n\bar{x}_n+y}{y}
$$
which, in this case, is a function of a pivotal quantity $Y/\bar{X}_n$.
This implies that the LR prediction method, based on bootstrap calibration, for example, has exact coverage probability, according to Theorem~\ref{theorem-location-scale-family-11} (see also Section~\ref{sec:pivotal-quantity}).

\section{Concluding Remarks}\label{sec:conclusion}

In this paper, we propose a general prediction procedure based on inverting an LR test.
The construction of the LR test requires enlarging the parameter space to create a quasi ``full model.''
To compute prediction intervals, we need to find the distribution of the LR statistic.
Apart from finding the distribution of the LR statistic analytically when possible, we may use chi-square distribution to calibrate its distribution when Wilks' theorem is applicable; we have demonstrated this for predictions involving discrete random variables.
Furthermore, we can use a parametric bootstrap as a general approach to approximate the distribution of the LR statistic, particularly in those cases where Wilks' theorem does not apply.
The proposed method will generally discover a pivotal quantity if one exists.
In such cases, the procedure will have exact coverage probability.
When a pivotal quantity is not available, we have shown that the LR method is asymptotically correct.
When the LR statistic is unimodal (as a function of $y$), then the proposed prediction region will correspond to an interval.
Relatedly, when the LR statistic is again unimodal, we provide an approach in Section~\ref{subsec:one-sided} to compute one-sided bounds in a computationally efficient manner (which is related to, but simpler than, working directly from the two-sided intervals in Section~\ref{subsec:lrt} in determining the endpoint for a one-sided bound).
While not encountered in any work for this paper, when the LRS is not unimodal, the prediction regions in Section~\ref{subsec:lrt} are still valid but these regions may be a union of several disconnected intervals and the algorithm of Section~\ref{subsec:one-sided} for finding one-sided bounds will not be applicable; one-sided bounds then need to be determined from the prediction regions of Section~\ref{subsec:lrt}.

We see several potential future research topics and list three below: (a) we only consider scalar random variables for prediction in this paper, but the proposed LR prediction framework could be extended to construct 2-d (or even higher dimensional) prediction regions using the same method as in (\ref{eq:bootstrap-find-lik-ratio}).
The main change is that $\boldsymbol{Y}$ in the joint likelihood function $\mathcal{L}(\boldsymbol{X}_n, \boldsymbol{Y})$ becomes a random vector.
(b) The proposed prediction framework could be applied to problems involving complicated data with regressors.
Examples include data with different types of censoring, mixed linear models, and generalized linear model structures.
(c) The LR prediction method could also be extended to dependent data.
We discuss an example involving dependence in Section~\ref{subsec:within-sample-prediction}.
But in future research, we might apply the LR prediction method to problems with non-trivial dependence structure such as time series or Markov Random Fields.

\section*{Acknowledgments}

We want to thank the anonymous reviewers and the editor, Galit Shmueli, who provided comments and suggestions that improved our paper.
Research was partially supported by NSF DMS-2015390.

\bibliographystyle{apalike}
\bibliography{reference}

\end{document}


\baselineskip24pt


\maketitle

\begin{sciabstract}

\ref{sec:proof-of-theorems} provides proof for the theorems in the main paper and \ref{sec:simulation-results} gives additional simulation results.

\end{sciabstract}

\section{Proof of Theoretical Results}
\label{sec:proof-of-theorems}

\subsection{Proof of Theorem~\ref{theorem-location-scale-family-11}}
For completeness, we first restate Theorem~\ref{theorem-location-scale-family-11} and then provide its proof.
\begin{theorem}\label{theorem-location-scale-family-11}
	(i) Suppose the LR-statistic (3) is a pivotal quantity.
	Then, the corresponding $1-\alpha$ prediction region (4) for $Y$ based on the parametric bootstrap will have exact coverage.
	That is,
	\[
	\Pr\left[Y\in\mathcal{P}_{1-\alpha}(\boldsymbol{X}_n)\right]=1-\alpha.
	\]
	(ii) Suppose also that both the data $X_1,\dots,X_n$ and $Y$ are from a location-scale distribution with density $f(\cdot;\mu,\sigma)=\phi\left[(x-\mu)/\sigma\right]$ with parameters $\btheta=(\mu,\sigma)\in\mathbb{R}\times(0,\infty)$.
	In the LR construction (3), suppose the full model involves parameters $\btheta=(\mu,\sigma)$ and $\btheta_{y}=(\mu_y,\sigma)$ (i.e., $X_1,\dots,X_n\sim f(\cdot;\mu,\sigma)$ and $Y\sim f(\cdot;\mu_y,\sigma)$).
	Then the LR statistic $\Lambda_{n}(\boldsymbol{X}_n,Y)$ (or $-2\log\Lambda_{n}(\boldsymbol{X}_n,Y)$) is a pivotal quantity and the result of Theorem~\ref{theorem-location-scale-family-11}(i) holds.
\end{theorem}

\begin{proof}
\noindent To establish~Theorem~\ref{theorem-location-scale-family-11}(i), the feature of the LR-statistic $\Lambda_n(\boldsymbol{X}_n,Y)$ from (3) being a pivotal quantity implies that its bootstrap counterpart $\Lambda_n(\boldsymbol{X}^\ast_n, Y^\ast)$ has the same pivotal distribution, or that
\[
\Lambda_n(\boldsymbol{X}_n,Y)\stackrel{d}{=}\Lambda_n(\boldsymbol{X}^\ast_n, Y^\ast)
\]
for any observed sample $\boldsymbol{X}_n=\boldsymbol{x}_n$.
Consequently, the $1-\alpha$ quantile $\lambda_{1-\alpha}^\ast$ of $-2\log\Lambda_n(\boldsymbol{X}_n^\ast,Y^\ast)$ is equal to that $\lambda_{1-\alpha}$ of $-2\log\Lambda_n(\boldsymbol{X}_n,Y)$,
that is, $\lambda_{1-\alpha}=\lambda_{1-\alpha}^\ast$.
The coverage probability of the bootstrap prediction region in (4) then follows as
\[
\begin{split}
\Pr\left[Y\in\mathcal{P}_{1-\alpha}(\boldsymbol{X}_n)\right]&=\Pr\left[-2\log\Lambda_n(\boldsymbol{X}_n, Y)\leq\lambda_{1-\alpha}^\ast\right]\\
&=\Pr\left[-2\log\Lambda_n(\boldsymbol{X}_n, Y)\leq\lambda_{1-\alpha}\right]\\&=1-\alpha.
\end{split}
\]

\noindent We next establish~Theorem~\ref{theorem-location-scale-family-11}(ii).  The full model is that $\boldsymbol{X}_n\sim\phi[(x-\mu_1)/\sigma_1]$ and $Y\sim\phi[(y-\mu_2)/\sigma_1]$ and the reduced model is
$\boldsymbol{X}_n\sim\phi[(x-\mu)/\sigma]$ and $Y\sim\phi[(y-\mu)/\sigma]$, where $\phi(\cdot)$ is a known pdf.
Under this formulation, the LR-statistic from (3) at values $(\boldsymbol{X}_n,Y)=(\boldsymbol{x}_n,y)$ is given by
\begin{equation*} \Lambda_n(\boldsymbol{x}_n,y)=\frac{\widehat{\sigma}^{-n-1}\prod_{i=1}^{n}\phi(\frac{x_i-\widehat{\mu}}{\widehat{\sigma}})\phi(\frac{y-\widehat{\mu}}{\widehat{\sigma}})}{\widehat{\sigma}_1^{-n-1}\prod_{i=1}^{n}\phi(\frac{x_i-\widehat{\mu}_1}{\widehat{\sigma}_1})\phi(\frac{y-\widehat{\mu}_2}{\widehat{\sigma}_1})},
\end{equation*}
where $(\widehat{\mu},\widehat{\sigma})$ is the ML estimator of $(\mu,\sigma)$ under reduced model.  and $(\widehat{\mu}_1,\widehat{\mu}_2,\widehat\sigma_1)$ is the ML estimator
of $(\mu_1,\mu_2,\sigma_1)$ under full model.
Denoting the mode of $\phi(\cdot)$ as $y_0$ and by letting $\widehat\mu_2=y-\widehat\sigma_1y_0$ for clarity, the value of $\phi[(y-\widehat{\mu}_2)/\widehat{\sigma}_1]$ may be seen as a constant
$C\equiv\phi(y_0)>0$. Then, the likelihood ratio can be re-written as
\begin{equation}\label{eq:exact-ls-dist}
	\Lambda_n(\boldsymbol{x}_n,y)=\frac{1}{C}\left(\frac{\widehat{\sigma_1}/\widehat{\sigma}}{\widehat{\sigma}/\sigma}\right)^{n+1}\frac{\prod_{i=1}^{n}\phi\left(\frac{x_i-\mu}{\sigma}\frac{\sigma}{\widehat{\sigma}}+\frac{\mu-\widehat{\mu}}{\sigma}\frac{\sigma}{\widehat{\sigma}}\right)\phi\left(\frac{y-\mu}{\sigma}\frac{\sigma}{\widehat{\sigma}}+\frac{\mu-\widehat{\mu}}{\sigma}\frac{\sigma}{\widehat{\sigma}}\right)}{\prod_{i=1}^{n}\phi\left(\frac{x_i-\mu}{\sigma}\frac{\sigma}{\widehat{\sigma}_1}+\frac{\mu-\widehat{\mu}_1}{\sigma}\frac{\sigma}{\widehat{\sigma}_1}\right)}.
\end{equation}
In (\ref{eq:exact-ls-dist}), the distribution of quantities $\widehat{\sigma}/\sigma$, $(\mu-\widehat{\mu})/\sigma$, $(Y-\mu)/\sigma$, and $(X_i-\mu)/\sigma$, $i=1,\ldots,n$, do not depend
on any parameters (cf.~\citet[Appendix E]{lawless_statistical_2002}). Hence,
if $\widehat{\sigma}_1/\sigma$  and $(\mu-\widehat{\mu}_1)/\sigma$ are likewise shown to have distributions not depending on any parameters (i.e., are
pivots), then the LR-statistic will be a pivot.  The parametric
bootstrap  will also then yield the exact distribution of $\Lambda_n(\boldsymbol{x}_n,y)$, so that the bootstrap prediction method is exact.

To prove that $\widehat{\sigma}_1/\sigma$  and $(\mu-\widehat{\mu}_1)/\sigma$ are indeed pivots,
we observe from (\ref{eq:exact-ls-dist}) that
$( \widehat{\mu}_1, \widehat{\sigma}_1)$ is  the maximizer of
\begin{equation}
\label{eqn:shift}
	\mathcal{L}(\mu_1,\sigma_1;\boldsymbol{x}_n)=\frac{1}{\sigma_1^{n+1}}\prod_{i=1}^{n}\phi\left(\frac{x_i-\mu_1}{\sigma_1}\right).
\end{equation}
If data $\boldsymbol{x}_n$ and parameter values $(\mu_1,\sigma_1)$ in (\ref{eqn:shift}) are scaled by a given positive $d>0$ and shifted by a given $c\in\mathbb{R}$, we may denote the resulting values as $\boldsymbol{x}_n^\prime=d\boldsymbol{x}_n+c$, $\mu^\prime_1=d\mu_1+c$, and $\sigma_1^\prime=d\sigma_1$, and we note that the corresponding objective function (\ref{eqn:shift}) would then become
\begin{equation*}
	\mathcal{L}(\mu_1^\prime,\sigma_1^\prime;\boldsymbol{x}_n^\prime)=(\sigma_1^\prime)^{-n-1}\prod_{i=1}^{n}\phi\left(\frac{x_i^\prime-\mu_1^\prime}{\sigma_1^\prime}\right)=(d\sigma_1)^{-n-1}\prod_{i=1}^{n}\phi\left(\frac{x_i-\mu_1}{\sigma_1}\right),
\end{equation*}
which would have a maximizer as $(\widehat{\mu}_1^\prime, \widehat{\sigma}_1^\prime)=(d\widehat{\mu}_1+c,d\widehat{\sigma}_1)$. This result implies that
$(\widehat{\mu}_1,\widehat{\sigma}_1)$ is an equivariant estimator.
Because $(\widehat{\mu}_1,\widehat{\sigma}_1)$ is equivariant and
$(X_i-\mu)/\sigma,i=1,\dots,n$ are pivots, the two quantities
\begin{equation*}
\begin{split}
	\widehat{\mu}_1\left(\frac{X_1-\mu}{\sigma},\dots,\frac{X_n-\mu}{\sigma}\right)=\frac{1}{\sigma}\left[\widehat{\mu}_1(X_1,\dots,X_n)-\mu\right],\\
	\widehat{\sigma}_1\left(\frac{X_1-\mu}{\sigma},\dots,\frac{X_n-\mu}{\sigma}\right)=\frac{\widehat{\sigma}_1(X_1,\dots,X_n)}{\sigma}
\end{split}
\end{equation*}
do not have any unknown parameters; thus, $\widehat{\sigma}_1/\sigma$ and $(\mu-\widehat{\mu}_1)/\sigma$ are pivotal quantities.
\end{proof}

\subsection{Proof of Theorem~\ref{theorem-1}}

For simplicity and clarity in presentation, we first state and prove a version of Theorem~\ref{theorem-1} in the case of single-parameter distributions.  The assumptions and regularity
conditions  described in the theorem statement are mild and will be discussed further after the theorem statement.  After establishing
this version of   Theorem~\ref{theorem-1},
we then discuss the extension to the case of multiple-parameters.

\begin{theorem} (Scalar parameter case.) \label{theorem-1}
Assume iid data $X_1,\dots,X_n$ and independent predictand $Y$ have common pdf $f(\cdot;\theta)$ depending on real-valued $\theta\in\Theta$. Let $\theta_0$ denote  the true parameter value.\\
(1)  Additionally, suppose the following conditions (a)-(e)  hold:
\begin{enumerate}[label=(\alph*)]
	\item  The ML estimator  $\tilde\theta_n$ based on $(X_1,\dots,X_n)$ satisfies $\tilde\theta_n\xrightarrow{p}\theta_0$.
	\item The log-density $\log f(x;\theta)$ is twice continuously differentiable in a neighborhood $O\subset\Theta$ of $\theta_0$.
	\item The first and second derivatives of $\log f(X_1;\theta)$ at $\theta_0$ have moments as
	\begin{equation*}
		\begin{split}
			&\text{E}_{\theta_0}\left[\frac{d\log f(X_1;\theta)}{d\theta}\bigg|_{\theta=\theta_0}\right]=0,\\
			&\text{I}(\theta_0)\equiv\text{E}_{\theta_0}\left[\frac{d\log f(X_1;\theta)}{d\theta}\bigg|_{\theta=\theta_0}\right]^2=-\text{E}_{\theta_0}\left[\frac{d^2\log f(X_1;\theta)}{d\theta^2}\bigg|_{\theta=\theta_0}\right]\in(0,\infty).
		\end{split}
	\end{equation*}
	\item Moments of second derivatives of $\log f(X_1;\theta)$ satisfy a continuity condition at $\theta_0$:
	\begin{equation*}
		 \text{E}_{\theta_0}\left[\sup_{|\theta-\theta_0|<\delta}\left|\frac{d^2\log f(X_1;\theta_0)}{d\theta^2} -\frac{d^2\log f(X_1;\theta)}{d\theta^2} \right|\right]\to0\text{ as }\delta\to0.
	\end{equation*}
	\item $\sup_{\theta\in\Theta}f(X_1;\theta)$ and $f(X_1;\theta_0)/\sup_{\theta\in\Theta}f(X_1;\theta)$ are  positive continuous random variables.
\end{enumerate}
Then, the asymptotic distribution of the LR statistic, as $n\to \infty$, is given by
\[
-2\log\Lambda_n(\boldsymbol{X}_n,Y) \xrightarrow{d}-2\log\left[\frac{f(Y;\theta_0)}{\sup_{\theta\in\Theta}f(Y;\theta)}\right],
\]
(2) In addition to conditions (a)-(e), further assume the following:
\begin{enumerate}[label=(\alph*)]
    \addtocounter{enumi}{5}
\item  Moments of second derivatives of $\log f(X_1;\theta)$ satisfy integrability conditions
		\begin{equation*}
			\begin{split}
				&\sup_{|\theta-\theta_0|<\delta}\text{E}_{\theta}\left[\sup_{|\theta^\dagger-\theta|<\delta}\left|\frac{d^2\log f(X_1;\theta^\dagger)}{d\theta^2}-\frac{d^2\log f(X_1;\theta )}{d\theta^2}\right|\right]\to0\text{ as }\delta\to0,\\
				&\sup_{\theta\in O}\text{E}_{\theta}\left[\left|\frac{d^2\log f(X_1;\theta_0)}{d\theta^2}\right|\text{I}\left(\left|\frac{d^2\log f(X_1;\theta_0)}{d\theta^2}\right|\geq M\right)\right]\to0\text{ as }M\to\infty,
			\end{split}
		\end{equation*}
where above $X_1\sim f(\cdot;\theta)$ under expectation $\text{E}_{\theta}$.
		\item   $x$-discontinuities of $f(x;\theta_0)$,  $d  \log f(x;\theta_0)/ d \theta $,  $d^2 \log f(x;\theta_0)/ d \theta^2$,  or of $\sup_{\theta\in\Theta}f(x;\theta)$ have probability 0 under $\Pr_{\theta_0}$.
		\item For any generic sequence $\{\theta_m\}$ of parameter values and associated random variables $Z_m\sim f(\cdot;\theta_m)$, if $\theta_m \to \theta_0$ holds, then $Z_m\xrightarrow{d} X_1$ where $X_1\sim f(\cdot;\theta_0)$.
\item
Letting $X_1^*,\ldots,X_n^*,Y^*$ denote iid  bootstrap  observations from $f(\cdot; \tilde\theta_n )$  (where $\tilde\theta_n$ is the ML estimator from    $(X_1,\dots,X_n)$) and letting $\widehat{\theta}^\ast_{n}$ and $\tilde{\theta}_n^\ast$ denote
    ML estimators from $(X_1^*,\ldots,X_n^*,Y^*)$ and $(X_1^*,\ldots,X_n^*)$, respectively, it holds that
    $\rho(|\widehat{\theta}^\ast_{n}-\tilde{\theta}_n|,0)+\rho(|\tilde{\theta}_n^\ast-\tilde{\theta}_n|,0) \stackrel{p}{\rightarrow}0$ as $n\to \infty$,
    where $\rho(\cdot,\cdot)$ denotes any metric for the distance between distributions that can be used to describe weak convergence.

		\end{enumerate}
		Then, the bootstrap method is asymptotically correct for the distribution of the LR-statistic: as $n\to \infty$,
	\[
	\sup_{\lambda\in\mathbb{R}^{+}}\left|\Pr{}_{\!\!\ast}\left[-2 \log \Lambda_n(\boldsymbol{X}_n^*,Y^*)\leq\lambda\right]-\Pr\left[-2 \log \Lambda_n(\boldsymbol{X}_n,Y) \leq\lambda\right]\right|\xrightarrow{p}0.
	\]

\end{theorem}
We comment on the conditions before presenting the proof of the scalar-parameter version of Theorem~\ref{theorem-1}.
Conditions (a)--(e) are used to determine the limit distribution of the log-LR statistic and correspond
to standard regularity conditions often applied in likelihood theory.  For example, conditions (b)--(c) require that the score-function
exists, with the usual mean zero and a variance interpretable as an information number.
Condition~(e) is natural for continuous data,
ensuring that the log of this quantity is a well-defined random variable.
Conditions~(f)--(i) are then further imposed
to establish the validity of the bootstrap approximation, though these conditions are also generally mild.
Condition~(f) provides a type of uniform integrability and convergence of moments arising from derivatives of the score-function (i.e., in a neighborhood
of the true parameter $\theta_0$).  Conditions~(g)--(h) are basic smoothness conditions on the   marginal data density $f(\cdot;\theta_0)$;
condition~(h) states that convergence of parameters implies convergence of underlying distributions, which would follow by Scheffe's theorem for example (i.e., verifying pointwise convergence of densities).  Condition~(i) corresponds to a basic condition on bootstrap parameter estimates,
which is technical due to the nature of bootstrap distributions being defined by observed data; this condition says that the distributions of bootstrap quantities $|\widehat{\theta}^\ast_{n}-\tilde{\theta}_n|$
and $|\tilde{\theta}^\ast_{n}-\tilde{\theta}_n|$, which are analogous to $|\tilde{\theta}_n-\theta_0|$ in the bootstrap world, converge to  zero in distribution (i.e., holding with high probability as sample size $n$ increases).
This condition is much weaker than the standard use of the bootstrap for parameter estimation, often requiring that $\sqrt{n}(\tilde{\theta}^\ast_{n}-\tilde{\theta}_n)$  and $\sqrt{n}( \tilde{\theta}_n-\theta_0)$ have the same non-degenerate limit distributions.  In condition~(i), the exact distance metric $\rho(\cdot,\cdot)$ is not important (e.g., Prokhorov or Levy metrics may be used);   for generic random variables, $\rho(Z_m,Z_0)\to 0$ must simply have the equivalent interpretation that $Z_m \stackrel{d}{\rightarrow}Z_0$.

\begin{proof}

\noindent To establish Theorem~\ref{theorem-1}(i) for a scalar parameter,  recall the ML estimator $\tilde\theta_n\equiv \tilde\theta_n(X_1,\dots,X_n)$ from the data $(X_1,\dots,X_n)$ (i.e., an iid sample of size $n$) maximizes the log-likelihood function
\[
l_n(\theta) \equiv \sum_{i=1}^n \log f(X_i;\theta),
\]
while the ML  estimator $\widehat\theta_n \equiv \widehat\theta_n(X_1,\dots,X_n,Y)$ from the data $(X_1,\dots,X_n,Y)$
(i.e., an iid sample $X_1,\dots,X_n,Y\equiv X_{n+1}$ of size $n+1$) maximizes the log-likelihood function
\[
l_{2,n}(\theta) \equiv \log f(Y;\theta) + \sum_{i=1}^n \log f(X_i;\theta) =\log f(Y;\theta) +  l_n(\theta).
\]
For simplicity in the following, we denote the first and second derivatives of $l_n(\theta),l_{2,n}(\theta)$ with respect to $\theta$ as
 $l_n^\prime(\theta),l_{2,n}^\prime(\theta)$ and  $l_n^{\prime\prime}(\theta),l_{2,n}^{\prime\prime}(\theta)$, when such derivatives appropriately exist.

Note that conditions~(a)-(b) imply that, with arbitrarily high probability
for large $n$, both $\tilde\theta_n, \widehat\theta_n $ lie in a neighborhood $O$ of $\theta_0$ and are solutions/roots of the log-likelihood equations
\begin{equation}
\label{eqn:score}
  0=l_n^\prime(\tilde\theta_n)\equiv \sum_{i=1}^n \frac{d \log f(X_i; \tilde\theta_n)}{d\theta}, \qquad 0= l_{2,n}^\prime(\widehat\theta_n) \equiv \frac{d \log f(Y; \widehat{\theta}_n)}{d\theta} +l_n^{\prime}(\widehat\theta_n).
\end{equation}
By Taylor expansion under condition~(b) and $l^\prime_n(\tilde{\theta}_n)=0$, we may then write the difference of the log-likelihood functions as
\begin{equation}
\label{eqn:a1}
l_n(\widehat\theta_n)-l_n(\tilde{\theta}_n)
=l_n^\prime(\tilde{\theta}_n)(\widehat{\theta}_n-\tilde{\theta}_n)
+\frac{1}{2}l_n^{\prime\prime}({\theta}_n^\ast)(\tilde{\theta}_n-\widehat{\theta}_n)^2 = \frac{1}{2}l^{\prime\prime}_n(\tau_n)(\tilde{\theta}_n-\widehat{\theta}_n)^2,
\end{equation}
in terms of some value $\tau_n $  between $\tilde{\theta}_n$ and $\widehat{\theta}_n$.   In addition, we may
re-write the score equation $0=l_{2,n}^\prime(\widehat\theta_n )$ in (\ref{eqn:score}) as
\begin{eqnarray*}
0 =  l_{2,n}^\prime(\widehat\theta_n )  =  \frac{d\log f(Y;\widehat{\theta}_n)}{d\theta}+l_n^\prime(\tilde{\theta}_n)+l_n^{\prime\prime}(\tau_{2,n})(\tilde{\theta}_n-\widehat{\theta}_n),
\end{eqnarray*}
using Taylor expansion $l_n^\prime(\widehat{\theta}_n)$ around $\tilde{\theta}_n$, where above $\tau_{2,n}$ denotes some value between $\tilde{\theta}_n$ and $\widehat{\theta}_n$;
because $l_n^\prime(\tilde{\theta}_n)=0$ in (\ref{eqn:score}), the above leads to
\begin{equation}
\label{eqn:a2}
	\frac{d\log f(Y;  \widehat{\theta}_n)}{d\theta} +l_n^{\prime\prime}(\tau_{2,n})(\tilde{\theta}_n-\widehat{\theta}_n)=0.
\end{equation}

Based on the developments above, define a scaled difference
$\Delta_n$ of second derivatives of log-likelihood functions as either $\Delta_n \equiv [ l_n^{\prime\prime}(\tau_{2,n}) - l_n^{\prime\prime}(\theta_0)]/n $
based on $\tau_{2,n}$ from (\ref{eqn:a2}) or as $\Delta_n \equiv [ l_n^{\prime\prime}(\tau_n) - l_n^{\prime\prime}(\theta_0)]/n $
based on $\tau_n$ from (\ref{eqn:a1}); both $\tau_{2,n}$ and $\tau_n$ similarly denote values between
  $\widehat{\theta}_n$ and $\tilde{\theta}_n$.   To complete the proof of Theorem~\ref{theorem-1}(i), we shall establish
the following three results, appearing in (\ref{eqn:a3})-(\ref{eqn:a5}) as $n\to \infty$:
\begin{equation}
\label{eqn:a3}
\Delta_n \stackrel{p}{\rightarrow} 0;
\end{equation}
\begin{eqnarray}
\label{eqn:a4} \frac{d\log f(Y; \widehat{\theta}_n)}{d\theta} =\frac{d\log f(Y;\theta_0)}{d\theta} + o_p(1), \\
\nonumber \log f(Y; \widehat{\theta}_n)  = \log f(Y;\theta_0) + o_p(1);\end{eqnarray}
\begin{equation} \label{eqn:a5} n|\widehat{\theta}_n - \tilde{\theta}_n| = O_p(1).
\end{equation}
When the above hold, we may re-write (\ref{eqn:a1}) as
\begin{equation}
\label{eqn:a6}
l_n(\widehat{\theta}_n)-l_n(\tilde{\theta}_n)=\frac{l_n^{\prime\prime}(\tau_{2,n})}{2n}\left[\sqrt{n}(\widehat{\theta}_n-\tilde{\theta}_n)\right]^2 = O_p(1)o_p(1)=o_p(1)
\end{equation}
using that $l_n^{\prime\prime}(\tau_{2,n})/n\xrightarrow{p} -I(\theta_0) \in (-\infty,0)$ by  (\ref{eqn:a3}) combined with
$l_n^{\prime\prime}(\theta_0)/n\xrightarrow{p} -I(\theta_0) $ by the SLLN (cf.~condition~(c)), while $\sqrt{n}|\widehat{\theta}_n-\tilde{\theta}_n|=o_p(1)$ by (\ref{eqn:a5}). From  (\ref{eqn:a4})  with (\ref{eqn:a6}), we then obtain
the distribution of the log-LR statistic as
\begin{eqnarray*}
-2\log\Lambda_n(\boldsymbol{X}_n,Y) &=& -2 l_n(\widehat{\theta}_n) -2   \log f(Y;\widehat{\theta}_n) + 2 l_n(\tilde{\theta}_n) + 2 \log \sup_{\theta\in\Theta} f(Y;\theta)\\&=& -2   \log f(Y;\theta_0) + 2 \log \sup_{\theta\in\Theta} f(Y;\theta) + o_p(1),
\end{eqnarray*}
where $\log \sup_{\theta\in\Theta} f(Y;\theta)$  and  $\log   f(Y;\theta_0) -\log \sup_{\theta\in\Theta} f(Y;\theta)$ are well-defined random variables (cf.~condition~(e) where $Y\stackrel{d}{=}X_1$);
that is,
\[
-2\log\Lambda_n(\boldsymbol{X_n},Y)\xrightarrow{d}-2\left[\log f(Y;\theta_0)-\log\sup_{\theta\in\Theta}f(Y;\theta)\right].
\]

We next establish (\ref{eqn:a3})-(\ref{eqn:a5}), beginning with (\ref{eqn:a3}).
 Consider $\Delta_n \equiv [ l_n^{\prime\prime}(\tau_{2,n}) - l_n^{\prime\prime}(\theta_0)]/n $
with $\tau_{2,n}$ from (\ref{eqn:a2}), lying between  $\widehat{\theta}_n$ and  $\tilde{\theta}_n$;
both latter estimators are consistent for $\theta_0$ by condition~(a).   For any $\epsilon>0$, we may pick $\delta$ to bound the probability
(under $\Pr\equiv \Pr_{\theta_0}$)
\begin{equation}\label{eq:conv}
	\begin{split}
		\Pr(|\Delta_n|>\epsilon)&\leq\Pr\left(\left|\Delta_n\right|>\epsilon,|\widehat{\theta}_n-\theta_0|<\delta/2, |\tilde{\theta}_n-\theta_0|<\delta/2\right)+\\
		&\Pr(|\widehat{\theta}_n-\theta_0|\geq\delta/2)+\Pr(|\tilde{\theta}_n-\theta_0|\geq\delta/2).
	\end{split}
\end{equation}
Using condition~(c) with Markov's inequality, we may bound the probability
\begin{equation*}
	\begin{split}
	&\Pr\left(\left|\Delta_n\right|>\epsilon,|\widehat{\theta}_n-\theta_0|<\delta/2, |\tilde{\theta}_n-\theta_0|<\delta/2\right)
	\leq\Pr(\left|\Delta_n\right|>\epsilon, \left|\tau_{2,n}-\theta_0\right|<\delta)\\
	\leq&\Pr\left[\frac{1}{n}\sum_{i=1}^{n}\sup_{|\theta-\theta_0|<\delta}\left|\frac{d^2\log f(X_i;\theta)}{d\theta^2} -\frac{d^2\log f(X_i;\theta_0)}{d\theta^2} \right|>\epsilon\right]\\
	\leq&\frac{h(\delta)}{\epsilon}, \quad h(\delta)\equiv \text{E}_{\theta_0} \sup_{|\theta-\theta_0|<\delta}\left|\frac{d^2\log f(X_1;\theta)}{d\theta^2} -\frac{d^2\log f(X_1;\theta_0)}{d\theta^2} \right|.
	\end{split}
\end{equation*}
Hence, for fixed $\epsilon,\delta>0$, taking limits as $n\to \infty$   in (\ref{eq:conv}) yields
\begin{eqnarray*}
 \limsup_{n\to \infty}\Pr\left(|\Delta_n|>\epsilon\right)
&\leq&\frac{h(\delta)}{\epsilon}+\limsup_{n\to \infty}\Pr(|\widehat{\theta}_n-\theta_0|\geq\delta/2)+\limsup_{n\to \infty}\Pr(|\tilde{\theta}_n-\theta_0|\geq\delta/2)\\
&
\leq& \frac{h(\delta)}{\epsilon},
\end{eqnarray*}
by consistency of $\widehat{\theta}_n$ and  $\tilde{\theta}_n$ (e.g., $\lim_{n\to \infty}\Pr(|\tilde{\theta}_n-\theta_0|\geq\delta/2)=0$).  Because $\delta>0$ was arbitrary, letting $\delta \to 0$ (i.e. using $\lim_{\delta \to 0}h(\delta) =0$ by condition~(c)) then yields $\limsup_{n\to \infty}\Pr\left(|\Delta_n|>\epsilon\right)=0$.
As $\epsilon>0$ was arbitrary,   (\ref{eqn:a3}) now follows for $\Delta_n \equiv [ l_n^{\prime\prime}(\tau_{2,n}) - l_n^{\prime\prime}(\theta_0)]/n $; the argument is the same for $\Delta_n \equiv [ l_n^{\prime\prime}(\tau_n) - l_n^{\prime\prime}(\theta_0)]/n $ with $\tau_n $ from (\ref{eqn:a1}).

To show (\ref{eqn:a4}),  we first consider expanding $d\log f(Y;\widehat{\theta}_n)/d\theta$ around $\theta_0$; in which case, we may write
\begin{equation}\label{eq:exp2}
\frac{d\log f(Y;\widehat{\theta}_n)}{d\theta} =\frac{d\log f(Y;\theta_0)}{d\theta} +\frac{d^2\log f(Y;\tau_{3,n})}{d\theta^2} (\widehat{\theta}_n-\theta_0),
\end{equation}
where $\tau_{3,n}$ is some value between $\theta_0$ and $\widehat{\theta}_n$.  Picking some small $\delta>0$,
when $|\widehat{\theta}_n-\theta_0|\leq\delta$ holds (with arbitrarily high probability when $n$ is large),
then $|\tau_{3,n}-\theta_0|\leq\delta$ also holds so that we may bound
\[
\left|\frac{d^2\log f(Y;\tau_{3,n})}{d\theta^2} \right|\leq\left|\frac{d^2\log f(Y;\theta_0)}{d\theta^2} \right|+\sup_{|\theta-\theta_0|<\delta}\left|\frac{d^2\log f(Y;\theta)}{d\theta^2}-\frac{d^2\log f(Y;\theta_0)}{d\theta^2} \right|,
\]
where the sum on the right-hand side above has finite/bounded expectation by conditions (c)--(d).  Consequently,
it follows that $ | d^2\log f(Y;\tau_{3,n})/d\theta^2 |=O_p(1)$ so that, from (\ref{eq:exp2}), we have
 \[
 \frac{d\log f(Y;\widehat{\theta}_n)}{d\theta} =\frac{d\log f(Y;\theta_0)}{d\theta} + O_p(1)o_p(1)
\]
in (\ref{eqn:a4}); the argument for  $\log f(Y;\widehat{\theta}_n) =  \log f(Y;\theta_0) + o_p(1)$ follows similarly.

To show (\ref{eqn:a5}), we may apply (\ref{eqn:a2}), (\ref{eqn:a3}), and (\ref{eqn:a4}) to write
\begin{equation}\label{find-diff-tilde-hat}
\frac{d\log f(Y|\theta)}{d\theta}\bigg|_{\theta=\theta_0}=o_p(1)+n(\tilde{\theta}_n-\widehat{\theta}_n)\cdot I(\theta_0)(1+o_p(1)),
\end{equation}
where $l_n^{\prime\prime}(\tau_{2,n})/n=-I(\theta_0)(1+ o_p(1))$ follows by (\ref{eqn:a3}) along with
$l_n^{\prime\prime}(\theta_0)/n\xrightarrow{p} -I(\theta_0)\in(-\infty,0) $ by the SLLN.
Because $o_p(1)$ terms  in (\ref{find-diff-tilde-hat}) will be less than $1/2$ (say) with high probability for large $n$, this aspect implies that
\[
n\left|\tilde{\theta}_n-\widehat{\theta}_n\right|\leq \frac{1}{I(\theta_0)}+\frac{2}{I(\theta_0)} \left|\frac{d\log f(Y;\theta_0)}{d\theta}\right|
\]
holds (with high probability for large $n$); note that the right-hand side has finite expectation (i.e., bounded in probability).
Consequently, it follows that $n|\tilde{\theta}_n-\widehat{\theta}_n|=O_p(1)$ in (\ref{eqn:a6}).

To establish  the bootstrap result of Theorem~\ref{theorem-1}(i) in the scalar parameter case, we establish
convergence in probability for the bootstrap approximation through a characterization of almost sure convergence
along subsequences.  Let $\{n_j\}\subset \{n\}$ denote a positive integer subsequence, and note that we may assume
the random variables $X_{n}$, $n \geq 1$, and $Y$ are defined on a common probability space $(\Omega,\mathcal{F}, \Pr)$ with $\sigma$-algebra $\mathcal{F}$ (and $\Pr \equiv \Pr_{\theta_0}$).
Then, under conditions~(a) and (i), there
exists a further subsequence $\{n_k\}\subset \{n_j\}$ and an event $A\in \mathcal{F}$ with $\Pr(A)=1$ such that, pointwise for $\omega \in A$,   it holds that $\tilde{\theta}_{n_k} \equiv \tilde{\theta}_{n_k}(\omega)\rightarrow \theta_0$ as $n_k \to \infty$ and
while $\rho( |\widehat{\theta}_{n_k}^\ast-\tilde{\theta}_{n_k}|,0) + \rho( |\tilde{\theta}_{n_k}^\ast-\tilde{\theta}_{n_k}|,0)\equiv
\rho( |\widehat{\theta}_{n_k}^\ast-\tilde{\theta}_{n_k}|,0)(\omega) + \rho( |\tilde{\theta}_{n_k}^\ast-\tilde{\theta}_{n_k}|,0)(\omega)
\rightarrow 0$
as $n_k \to \infty$.  Note that, pointwise   for each $\omega \in A$ (we will suppress the dependence of random variables $X_i, Y$ and estimators $\tilde{\theta}_{n}$ on $\omega$ hereafter), there is a sequence of bootstrap distributions, with each distribution indexed by $n_k$ and  defined by bootstrap observations $X_1^*,\ldots, X_{n_k}^*, Y^*\equiv Y^*_{n_k}$ as iid draws
from $f(\cdot; \tilde{\theta}_{n_k})$ based on the ML estimator $\tilde{\theta}_{n_k}$ from observed data $X_1,\ldots, X_{n_k}$.
That is, for fixed $\omega \in A$, we have  $\tilde{\theta}_{n_k} \rightarrow \theta_0$, $\widehat{\theta}_{n_k}^\ast-\tilde{\theta}_{n_k}
 \stackrel{p*}{\rightarrow}0$ and $\tilde{\theta}_{n_k}^\ast-\tilde{\theta}_{n_k}
 \stackrel{p*}{\rightarrow}0$ as $n_k \to \infty$, where  $\stackrel{p*}{\rightarrow}$ denotes convergence in bootstrap probability, and
we consider  establishing that the bootstrap log-LR statistic $-2\log\Lambda_{n_k}^\ast(\boldsymbol{X}_{n_k}^*,Y^*)$ converges in distribution as
$n_k \to \infty$, denoted as $\stackrel{d*}{\rightarrow}$, to the same limit $-2 \log [ f(Y;\theta_0)/\sup_{\theta \in \Theta} f(Y;\theta)]$ as the original log-LR statistic (i.e., $Y \sim f(\cdot;\theta_0)$).
The bootstrap proof is similar to that of  Theorem~\ref{theorem-1}(i), with some modifications for the bootstrap version $Y^*\equiv Y_{n_k}^*\sim f(\cdot;\tilde{\theta}_{n_k})$
of the predictand described next.

As $\widehat{\theta}^\ast_{n_k}\stackrel{p*}{\rightarrow} \theta_0$  and $\tilde{\theta}^\ast_{n_k}\stackrel{p*}{\rightarrow} \theta_0$ (so that $\widehat{\theta}^\ast_{n_k},\tilde{\theta}^\ast_{n_k}$ lie  in a neighborhood $O$ of $\theta_0$ in condition~(b)), we may
apply the same Taylor expansions used in the Theorem~\ref{theorem-1}(i) at the bootstrap level.
Define $l_{n_k}^*(\theta) $ and $l_{2,n_k}^*(\theta)$ to be the bootstrap counterparts of log-likelihood functions $l_{n_k}(\theta) $ and $l_{2,n_k}(\theta)$ (based on $X_1^*,\ldots,X_{n_k}^*,Y^*$ in place of $X_1,\ldots,X_{n_k},Y$), with corresponding derivatives $l_{n_k}^{*\prime}(\theta) ,l_{2,n_k}^{*\prime}(\theta)$  and second derivatives $l_{n_k}^{*\prime\prime}(\theta) ,l_{2,n_k}^{*\prime \prime}(\theta)$.  Then, the same expansions as in (\ref{eqn:a1})-(\ref{eqn:a2}) at the original data level apply for bootstrap data as
   \begin{eqnarray}
\label{eqn:aa1}
l_{n_k}^*(\widehat\theta_{n_k}^*)-l_{n_k}^*(\tilde{\theta}_{n_k}^*) = \frac{1}{2}l^{*\prime\prime}_{n_k}(\tau_{n_k}^*)(\tilde{\theta}_{n_k}^*-\widehat{\theta}_{n_k}^*)^2,\\
\label{eqn:aa2}\frac{d\log f(Y_{n_k}^*;  \widehat{\theta}_{n_k}^*)}{d\theta} +l_{n_k}^{*\prime\prime}(\tau_{2,n_k}^*)(\tilde{\theta}_{n_k}^*-\widehat{\theta}_{n_k}^*)=0,
\end{eqnarray}
in terms of some values $\tau_{n_k}^*, \tau_{2, n_k}^*$ between $\tilde{\theta}_{n_k}^*$ and $\widehat{\theta}_{n_k}^*$.
Define a bootstrap difference as $\Delta_{n_k}^* \equiv [ l_{n_k}^{*\prime\prime}(\tau_{2,n_k}^*) - l_{n_k}^{*\prime\prime}(\theta_0)]/n $
based on $\tau_{2,n_k}^*$ from (\ref{eqn:aa2}) or as $\Delta_{n_k}^* \equiv [ l_{n_k}^{*\prime\prime}(\tau_{n_k}^*) - l_{n_k}^{*\prime\prime}(\theta_0)]/n $
based on $\tau_{n_k}^*$ from (\ref{eqn:aa1}).  Similar to the proof of  Theorem~\ref{theorem-1}(i), to complete the proof of Theorem~\ref{theorem-1}(ii), we shall establish analog versions of (\ref{eqn:a3})-(\ref{eqn:a5}) given by
\begin{equation}
\label{eqn:aa23}
l_{n_k}^{*\prime\prime}(\theta_0) \stackrel{p*}{\rightarrow} - I(\theta_0) \in (0-\infty);
\end{equation}
\begin{equation}
\label{eqn:aa3}
\Delta_{n_k}^* \stackrel{p*}{\rightarrow} 0;
\end{equation}
\begin{eqnarray}
\label{eqn:aa4} \frac{d\log f(Y_{n_k}^*; \widehat{\theta}_{n_k}^*)}{d\theta} =\frac{d\log f(Y_{n_k}^*;\theta_0)}{d\theta} + o_{p*}(1), \\
\nonumber \frac{d\log f(Y_{n_k}^*;\theta_0)}{d\theta} =O_{p*}(1),\\
\nonumber \log f(Y_{n_k}^*; \widehat{\theta}_{n_k}^*)  = \log f(Y_{n_k}^*; \theta_0) + o_{p*}(1);\end{eqnarray}
\begin{equation} \label{eqn:aa5} n|\widehat{\theta}_{n_k}^* - \tilde{\theta}_{n_k}^*| = O_{p*}(1),
\end{equation}
where above $O_{p*}(1)$ (i.e., bounded in bootstrap probability) and $o_{p*}(1)$ (i.e., converging to zero in bootstrap probability) denote bootstrap probability orders   along the subsequence $n_k$.
When  (\ref{eqn:aa23})-(\ref{eqn:aa5}) hold, we can  re-write (\ref{eqn:aa1}) as
\begin{equation}
\label{eqn:aa6}
l_{n_k}^*(\widehat\theta_{n_k}^*)-l_{n_k}^*(\tilde\theta_{n_k}^*) = \frac{1}{2n}l^{*\prime\prime}_{n_k}(\tau_{n_k}^*)\left[\sqrt{n_k}(\widehat{\theta}_{n_k}^*-\tilde{\theta}_{n_k}^*)\right]^2=
 O_{p*}(1)o_{p*}(1)=o_{p*}(1)
\end{equation}
using  (\ref{eqn:aa23})-(\ref{eqn:aa3}) (i.e., establish  $l^{*\prime\prime}_{n_k}(\tau_{n_k}^*)/n \stackrel{p*}{\rightarrow} -I(\theta_0) \in (-\infty,0)$)  combined with (\ref{eqn:aa5}).
  From  (\ref{eqn:aa4})  with (\ref{eqn:aa6}), we then write the bootstrap log-LR statistic as
\begin{eqnarray}
\label{eqn:end}
&&-2\log\Lambda_{n_k}^*(\boldsymbol{X_n^*},Y_{n_k}^*) \\
\nonumber &=& -2 l_{n_k}^*(\tilde{\theta}_{n_k}^*) -2   \log f(Y_{n_k^*};\widehat{\theta}_{n_k}^*) + 2 l_{n_k}^*(\tilde{\theta}_{n_k}^*) + 2 \log \sup_{\theta\in\Theta} f(Y_{n_k}^*;\theta)\\&=& \nonumber -2   \log f(Y_{n_k}^*;\theta_0) + 2 \log \sup_{\theta\in\Theta} f(Y_{n_k}^*;\theta) + o_{p*}(1).
\end{eqnarray}
Because $Y^*\equiv Y_{n_k}^\ast\sim f(\cdot;{\tilde{\theta}_{n_k}})$ and $\tilde{\theta}_{n_k}\to\theta_0$, we have $Y^\ast_{n_k} \stackrel{d*}{\rightarrow} Y_0\sim f(\cdot|\theta_0)$ under condition~(h); from this, by the continuous mapping theorem under conditions~(e) and (g), it holds that the random pair
\[
	\left(\log f(Y^\ast_{n_k};\theta_0), \log \sup_{\theta \in \Theta}f(Y^\ast_{n_k};\theta) \right) \stackrel{d*}{\rightarrow}
\left(\log f(Y_0;\theta_0),
\log \sup_{\theta \in \Theta}f(Y_0;\theta)\right)
	\]
converges in distribution. Consequently,  the limit
\[-2\log\Lambda_{n_k}^*(\boldsymbol{X_n^*},Y_{n_k}^*) \stackrel{d*}{\rightarrow}  -2   \log f(Y_0;\theta_0) + 2 \log \sup_{\theta\in\Theta} f(Y_0;\theta)
\]
then follows in (\ref{eqn:end}) by Slutsky's theorem and the continuous mapping theorem; note that this limit corresponds to the same (continuous) distributional limit as the original log-LR statistic in Theorem~\ref{theorem-1}(i) (i.e., $Y_0\sim f(\cdot|\theta_0)$).
 By Polya's theorem, we then may write
	\[
	\lim_{n_k \to \infty}\sup_{\lambda\in\mathbb{R}^{+}}|\Pr{}_{\!\!\ast}[-2\log\Lambda_{n_k}^*(\boldsymbol{X}_{n_k}^*,Y_{n_k}^*)\leq\lambda]-
\Pr[-2\log\Lambda_{n_k}(\boldsymbol{X}_{n_k},Y)\leq\lambda]| =0.
	\]
As this convergence above holds (almost surely) along the subsequence $\{n_k\}\subset \{n_j\}$, where the latter subsequence was arbitrary, we now have
\[\sup_{\lambda\in\mathbb{R}^{+}}|\Pr{}_{\!\!\ast}[-2\log\Lambda_{n}^*(\boldsymbol{X}_n^*,Y_{n}^*)\leq\lambda]-
\Pr[-2\log\Lambda_{n}(\boldsymbol{X}_{n},Y)\leq\lambda]| \stackrel{p}{\rightarrow}0
\]
as $n\to \infty$ in Theorem~\ref{theorem-1}(ii).

 Finally, we establish (\ref{eqn:aa23})-(\ref{eqn:aa5}) to complete the proof, beginning with (\ref{eqn:aa23}).
 To show (\ref{eqn:aa23}), we first require some bootstrap moment results for the second derivatives of log-densities.
Because  again $\tilde{\theta}_{n_k} \to \theta_0$ and  $X_1^\ast\equiv X_{1,n_k}^\ast\sim f(\cdot;\tilde{\theta}_{n_k})$, then $X_1^\ast \stackrel{d*}{\rightarrow} Y_0\sim f(\cdot;\theta_0)$ as $n_k\to \infty$ by condition~(h) and we also subsequently have
\begin{equation}
\label{eqn:help}
	\frac{d^2\log f(X_1^\ast;\theta_0)}{d\theta^2}\stackrel{d*}{\rightarrow}\frac{d^2\log f(Y_0;\theta_0)}{d\theta^2}
	\end{equation}
by condition~(g) and the continuous mapping theorem.  Noting that the bootstrap sample $X_1^*,\ldots,X_{n_k}^*$ is iid under  $f(\cdot;\tilde{\theta}_{n_k})$, we shall truncate each $d^2\log f(X_i^\ast; \theta_0)/d\theta^2$, $i=1,\ldots,n_k$, as
	\[
	T_{i,{n_k}}^\ast(M)\equiv\frac{d^2\log f(X_i^\ast;\theta_0)}{d\theta^2}\text{I}\left(\left|\frac{d^2\log f(X_i^\ast;\theta_0)}{d\theta^2}\right|\leq M\right),
	\]
	where $M$ is a continuity point in the distribution (cdf) of $d^2\log f(Y_0;\theta_0)/d\theta^2$ and
	$Y_0\sim f(\cdot;\theta_0)$ (i.e., there can be only countably many discontinuity points).  Then, at any given continuity point $M$,
we have from the continuous mapping theorem and (\ref{eqn:help}) that
\[
T_{1,{n_k}}^\ast(M)\stackrel{d*}{\rightarrow} \frac{d^2\log f(Y_0;\theta_0)}{d\theta^2}\text{I}\left(\left|\frac{d^2\log f(Y_0;\theta_0)}{d\theta^2}\right|\leq M\right)
\]
for $Y_0\sim f(\cdot;\theta_0)$.  Because $T_{1,{n_k}}^\ast(M)$ is bounded (i.e., uniformly integrable), the latter convergence in distribution implies the following convergence of bootstrap moments
\begin{eqnarray}\label{key1}
	\text{E}_{\tilde{\theta}_{n_k}}^\ast\left(\frac{1}{n_k}\sum_{i=1}^{n_k}T_{i,{n_k}}^\ast(M)  \right) &=& \text{E}_{\tilde{\theta}_{n_k}}^\ast [T_{1,{n_k}}^\ast(M)] \\\nonumber &\rightarrow &\text{E}_{\theta_0}\left[\frac{d^2\log f(Y_0;\theta_0)}{d\theta^2}\text{I}\left(\left|\frac{d^2\log f(Y_0;\theta_0)}{d\theta^2}\right|\leq M\right)\right]\\
\nonumber &\equiv& Q(M),
	\end{eqnarray}
 using that $\{T_{i,{n_k}}^\ast(M)\}_{i=1}^{n_k}$ are iid bootstrap variables.
By condition~(b) and the dominated convergence theorem, note that
	  $Q(M)\to-I(\theta_0)$ as $M\to\infty$.    We now bound the  bootstrap expectation (absolutely) between
$l_{n_k}^{\ast\prime\prime}(\theta_0)$ and $-I(\theta_0)$ as
\begin{eqnarray}\label{eq:2333}
	&&\text{E}_{\tilde{\theta}_{n_k}}^\ast\left|\frac{1}{n_k}l_{n_k}^{\ast\prime\prime}(\theta_0)+I(\theta_0)\right|\\
\nonumber &\leq&
\text{E}_{\tilde{\theta}_{n_k}}^\ast\left|\frac{1}{n_k}\sum_{i=1}^{n_k} [T_{i,{n_k}}^\ast(M)-\text{E}^\ast_{\tilde{\theta}_{n_k}}T_{i,{n_k}}^\ast(M)]\right|  +\left|\text{E}^\ast_{\tilde{\theta}_{n_k}}T_{1,{n_k}}^\ast(M)+\text{I}(\theta_0)\right|\\
	\nonumber \qquad && + \frac{1}{n_k}\sum_{i=1}^{n_k}\text{E}_{\tilde{\theta}_{n_k}}^\ast\left|\frac{d^2\log f(X_i^\ast;\theta_0)}{d\theta^2}\right|\text{I}\left(\left|\frac{d^2\log f(X_i^\ast;\theta_0)}{d\theta^2}\right|>M\right)\\
 \nonumber & \equiv&  a_{1,n_k}(M) + a_{2,n_k}(M)+a_{3,n_k}(M).
	\end{eqnarray}
Note that, for fixed $M>0$, $[a_{1,n_k}(M)]^2$ is bounded by the bootstrap variance
\begin{eqnarray}
\label{eqn:help2}
  [a_{1,n_k}(M)]^2 &\leq &\text{E}^\ast_{\tilde{\theta}_{n_k}}\left\{\frac{1}{n_k}\sum_{i=1}^{n_k}\left[T_{i,{n_k}}^\ast(M) -\text{E}^\ast_{\tilde{\theta}_{n_k}}(T_{i,{n_k}}) \right]\right\}^2\\
  \nonumber &=& \frac{1}{n_k}\text{E}^\ast_{\tilde{\theta}_{n_k}}\left[ T_{1,{n_k}}^\ast(M) -\text{E}^\ast_{\tilde{\theta}_{n_k}}(T_{1,{n_k}} )\right]^2 \\
   \nonumber &\leq& \frac{M^2}{n_k},
\end{eqnarray}
using that $T_{i,{n_k}}^\ast(M)$, $i=1,\ldots,n_k$ are iid and bounded by $M$.  Additionally, we may bound
\begin{eqnarray*}
a_{3,n_k}(M)& \equiv &\frac{1}{n_k}\sum_{i=1}^{n_k}\text{E}_{\tilde{\theta}_{n_k}}^\ast\left|\frac{d^2\log f(X_i^\ast;\theta_0)}{d\theta^2}\right|\text{I}\left(\left|\frac{d^2\log f(X_i^\ast;\theta_0)}{d\theta^2}\right|>M\right)
\\&= &\text{E}_{\tilde{\theta}_{n_k}}^\ast\left|\frac{d^2\log f(X_1^\ast;\theta_0)}{d\theta^2}\right|\text{I}\left(\left|\frac{d^2\log f(X_1^\ast;\theta_0)}{d\theta^2}\right|>M\right)\\& \leq &\sup_{\theta \in O}\text{E}_{\theta}\left|\frac{d^2\log f(X_1;\theta_0)}{d\theta^2}\right|\text{I}\left(\left|\frac{d^2\log f(X_1;\theta_0)}{d\theta^2}\right|>M\right) \equiv a_3(M),
\end{eqnarray*}
under condition~(f) (e.g., note $\tilde{\theta}_{n_k}\to\theta_0\in O$) where $\lim_{M\to \infty}a_3(M)=0$.
Now fixing $M>0$ and taking limits as $n_k\to \infty$ in (\ref{eq:2333}), we have
\[
\limsup_{n_k \to \infty}\text{E}_{\tilde{\theta}_{n_k}}^\ast\left|\frac{1}{n_k}l_{n_k}^{\ast\prime\prime}(\theta_0)+I(\theta_0)\right| \leq  |Q(M)+I(\theta_0)| + a_3(M)
\]
using that  $\limsup_{n_k \to \infty} a_{1,n_k}(M)=0$ by (\ref{eqn:help2}), $\limsup_{n_k \to \infty} a_{2,n_k}(M)=|Q(M)+I(\theta_0)|$
by (\ref{eqn:help}), and $\limsup_{n_k \to \infty} a_{3,n_k}(M) \leq a_3(M)$.  Now letting $M\to \infty$ and using $\lim_{M\to \infty}Q(M)=-I(\theta_0)$ and $\lim_{M\to \infty}a_3(M)=0$, we find
\[
\limsup_{n_k \to \infty}\text{E}_{\tilde{\theta}_{n_k}}^\ast\left|\frac{1}{n_k}l_{n_k}^{\ast\prime\prime}(\theta_0)+I(\theta_0)\right| =0;
\]
this establishes (\ref{eqn:aa23}).

To   show   (\ref{eqn:aa3}), the fact that $\Delta_{n_k}^* \stackrel{p*}{\rightarrow} 0$  follows from  $\tilde{\theta}_{n_k}^* \stackrel{p*}{\rightarrow} \theta_0$ and $\widehat{\theta}_{n_k}^* \stackrel{p*}{\rightarrow} \theta_0$ along with  condition~(f) (e.g., note $\tilde{\theta}_{n_k}\to\theta_0\in O$).  To   establish   (\ref{eqn:aa4}),  as $\widehat{\theta}^\ast_{n_k}\stackrel{p*}{\rightarrow} \theta_0$ (so that $\widehat{\theta}^\ast_{n_k}$ lies in a neighborhood $O$ of $\theta_0$ in condition~(b)), we may
 expand $d\log f(Y_{n_k}^\ast;\widehat{\theta}^\ast_{n_k})/d \theta$ around $\theta_0$ as
 \[
	\frac{d\log f(Y_{n_k}^\ast; \widehat{\theta}_{n_k}^\ast)}{d \theta}=\frac{d\log f(Y_{n_k}^\ast; \theta_0)}{d \theta} +\frac{d^2\log f(Y^\ast_{n_k};\zeta_{n_k}^\ast)}{d\theta^2}(\widehat{\theta}^\ast_{n_k}-\theta_0),
\]
where $\zeta_{n_k}^\ast$ is between $\theta_0$ and $\widehat{\theta}_{n_k}^\ast$. Because
$|\zeta_{n_k}^\ast-\theta_0| \leq |\widehat{\theta}_{n_k}^\ast-\theta_0| \stackrel{p*}{\rightarrow}0$
and because $Y^\ast_{n_k} \sim f(;\tilde{\theta}_{n_k})$ where $ \tilde{\theta}_{n_k} \rightarrow \theta_0 $,
it follows from condition~(f) that the bootstrap expectation $\text{E}_{\tilde{\theta}_{n_k}}^*  |d^2\log f(Y_{n_k}^\ast;\zeta_{n_k}^\ast)/d\theta^2|I(| \zeta_{n_k}^\ast-\theta_0| <\delta)$ is a bounded sequence in $n_k$ (for a given
small $\delta$), which implies  $d^2\log f(Y_{n_k}^\ast;\zeta_{n_k}^\ast)/d\theta^2=O_{p*}(1)$ is tight in bootstrap probability.
Consequently, we have
\[
	\frac{d\log f(Y_{n_k}^\ast; \widehat{\theta}_{n_k}^\ast)}{d \theta}=\frac{d\log f(Y_{n_k}^\ast; \theta_0)}{d \theta} + o_{p*}(1)
\]
in (\ref{eqn:aa4}).
Additionally,  because $Y^\ast_{n_k} \sim f(;\tilde{\theta}_{n_k})$ where $ \tilde{\theta}_{n_k} \rightarrow \theta_0 $, it follows from conditions (g)-(h) with the continuous mapping theorem that   $d\log f(Y_{n_k}^\ast; \theta_0)/d \theta \stackrel{d*}{\rightarrow} d\log f(Y_0; \theta_0)/d \theta $
and $\log \sup_{\theta \in \Theta}f(Y^\ast_{n_k};\theta)\stackrel{d*}{\rightarrow}\log \sup_{\theta \in \Theta}f(Y_0;\theta)$
for $Y_0\sim f(\cdot;\theta_0)$ in  (\ref{eqn:aa4}) (e.g.,  consequently $d\log f(Y_{n_k}^\ast; \theta_0)/d \theta =O_{p*}(1)$ is  tight in bootstrap probability).  Finally,  to establish (\ref{eqn:aa5}),
we may apply the same arguments that used to derive (\ref{find-diff-tilde-hat}) to find
	\[
	\frac{d\log f(Y^\ast_{n_k};\widehat{\theta}^*_{n_k})}{d\theta} =o_{p*}(1)+I(\theta_0)(\tilde{\theta}_{n_k}^\ast-\widehat{\theta}_{n_k}^\ast)n_k[1+o_{p^\ast}(1)],
	\]
	implying $n|\tilde{\theta}_{n_k}^\ast-\widehat{\theta}_{n_k}^\ast|=O_{p*}(1)$.
\end{proof}
\noindent\textbf{Remark~1}.
Here we discuss the extension to the case  of multiple parameters $\theta$.  To connect the multiple parameter
setting to scalar parameter version of Theorem~\ref{theorem-1} and the notation/proof developed there,
we write the multiple parameter as $\btheta = (\phi, \vartheta)\in \boldsymbol\Theta$, where $\phi$ is real-valued and  $\vartheta$
may be vector valued.  For constructing an LR statistic (3), in this notation the full model
is given by $X_1\ldots,X_n \sim f(\cdot; \phi, \vartheta)$ and $Y\sim f(\cdot; \varphi, \vartheta)$
while the reduced model is $X_1\ldots,X_n,Y \sim f(\cdot; \phi, \vartheta)$; that is, under the full model,
the density of the predictand $Y$ involves a real parameter $\varphi$ that may differ from the counterpart parameter
$\phi$ in the density of the data $X_1,\ldots,X_n$, though the remaining parameters $\vartheta$ are common to these densities in the full model.

In this notation, define $\tilde{\btheta}_n$ as the maximizer of the log-likelihood $l_n(\btheta) \equiv \sum_{i=1}^n \log f(X_i;\btheta)$
and define $\widehat{\btheta}_n$ as the maximizer of the log-likelihood $l_{2,n}(\btheta) \equiv \log f(Y;\btheta) +\sum_{i=1}^n \log f(X_i;\btheta)$.
The definitions of $\tilde{\btheta}_n,\widehat{\btheta}_n$ here also match those used in the proof of the scalar case
of  Theorem~\ref{theorem-1}.  Next define $\bar{\btheta}_n$ as the maximizer
of $l_{3,n}(\btheta) \equiv l_{3,n}(\phi,\vartheta) \equiv  \log \sup_{\varphi} f(Y;\varphi,\vartheta) +\sum_{i=1}^n \log f(X_i;\btheta)$.
In the scalar parameter case, the distinction between $l_n(\btheta)$ and $l_{3,n}(\btheta)$ as log-likelihood functions is small
in the sense that the difference $l_{3,n}(\btheta)-l_n(\btheta)$ does not depend on $\btheta$ (i.e., when there
is one parameter, we can write the difference $l_{3,n}(\btheta)-l_n(\btheta) = \log \sup_{\phi} f(Y;\phi)$ in a way that does not depend on the parameter $\btheta$).  That is, when there is one parameter,  $\tilde{\btheta}_n$ and  $\bar{\btheta}_n$ are the same; however, $\tilde{\btheta}_n$ and  $\bar{\btheta}_n$ may not be the same estimators of $\btheta$ in the multi-parameter setting.  Hence, to connect the construction
of LR statistics between the multiple and scalar parameter case, we can write the LR statistic (3) as
\begin{eqnarray*}
-2 \log \Lambda_n(\boldsymbol{X}_n,Y) &=& -2 [l_n(\widehat{\btheta}_n) - l_{3,n}(\bar{\btheta}_n)  ]\\
& =&   -2 [l_n(\widehat{\btheta}_n) - l_{2,n}(\tilde{\btheta}_n)  ] +  2 [l_{3,n}(\bar{\btheta}_n)-l_{2,n}(\tilde{\btheta}_n)].
\end{eqnarray*}
Under the same  conditions and arguments used in the scalar version of  Theorem~\ref{theorem-1}, the component $-2 [l_n(\widehat{\btheta}_n) - l_{3,n}(\bar{\btheta}_n)  ]$ above behaves the same way in the multiple parameter case as in the scalar case; that is,
\[
-2 [l_n(\widehat{\btheta}_n) - l_{3,n}(\bar{\btheta}_n)  ] = -2 \log f(Y;\btheta_0) + o_p(1)
\]
holds where $\btheta_0$ denotes the true parameter value (e.g., $\btheta_0=(\phi_0, \vartheta_0)$ for multiple parameters).
However, the component  $[l_{3,n}(\bar{\btheta}_n)-l_{2,n}(\tilde{\btheta}_n)]$ behaves slightly differently between scalar and multiple parameter cases.  For a single parameter $\btheta$, we have $2[l_{3,n}(\bar{\btheta}_n)-l_{2,n}(\tilde{\btheta}_n)]=2 \log\sup_{\varphi} f(Y;\varphi)$ exactly, where $Y\sim f(\cdot;\btheta_0)$ again, and there are then no further steps needed to obtain the limit of
$-2 \log \Lambda_n(\boldsymbol{X}_n,Y)$.  For multiple parameters though, we  repeat an expansion of $2[l_{3,n}(\bar{\btheta}_n)-l_{2,n}(\tilde{\btheta}_n)]$, that is similar to that for  $-2 [l_n(\widehat{\btheta}_n) - l_{2,n}(\tilde{\btheta}_n)  ]$, to obtain
\[
2[l_{3,n}(\bar{\btheta}_n)-l_{2,n}(\tilde{\btheta}_n)] = 2 \log \sup_{\phi} f(Y;\phi,\vartheta_0) + o_p(1),
\]
where    $\btheta_0=(\phi_0, \vartheta_0)$ again denotes the true parameter value.  Upon this step, the limit of the log-LR statistic $-2 \log \Lambda_n(\boldsymbol{X}_n,Y)$ follows in a unified way for both scalar and multiple parameter cases.

The conditions required in the multiple parameter case   largely match those given in the scalar version of Theorem~\ref{theorem-1},
with the understanding that first/second derivatives in the condition statements are replaced by first/second partial derivatives
and that parameter estimators $\tilde{\btheta}_n, \widehat{\btheta}_n, \bar{\btheta}_n$ are defined as above (in agreement between scalar and multiple parameter cases).
The other modifications to the conditions given in the scalar version of Theorem~\ref{theorem-1} are as follows.
Let $\btheta_0=(\phi_0,\vartheta_0)$ denote the true parameter value.
Condition (a) is augmented to include  $\bar{\btheta}_n \stackrel{p}{\rightarrow}\btheta_0$; condition~(h)
is augmented to include $\rho(|\bar{\btheta}_n^*- \tilde{\btheta}_n|,0)\stackrel{p}{\rightarrow}$; and ``$\sup_{\btheta \in \Theta} f(x;\btheta)$"
is replaced by ``$\sup_{\phi} f(x;\phi,\vartheta_0)$" in conditions (e) and (g). Finally, one additional assumption is required
that $g(x;\vartheta)\equiv \sup_{\phi} f(x;\phi,\vartheta)$ is continuously differentiable in a neighborhood $\mathcal{O}$ of $\vartheta_0$ and
that, over a neighborhood $O$ of $\btheta_0$, expectations $\sup_{\btheta \in O}E_{\btheta} \sup_{ \vartheta \in \mathcal{O}} |\partial g(X_1;\vartheta)/\partial \vartheta  |<\infty$ are bounded,
 where $\text{E}_{\btheta}$ denotes expected valued with respect to $X_1\sim f(\cdot;\btheta)$.\\

 \noindent\textbf{Remark~2}.  Here we describe the derivation of the limit distribution of the signed log-LR statistic, as given in Remark~2 from
the main manuscript.  Let $\Lambda_n(\boldsymbol{X}_n,y)$ denote an LR statistic (3), as a function of $y$ for a given data $\boldsymbol{X}_n$,
which is assumed to be unimodal (either with probability 1 or with probability approaching 1 as $n\to \infty$) with a mode at $y_0\equiv y_0(\boldsymbol{X}_n)$.  The limit of the signed log-LR statistic $-(-1)^{I(Y \leq y_0)}2 \Lambda_n(\boldsymbol{X}_n,Y)$ follows
from the established limit of the log-LR statistic $-2 \Lambda_n(\boldsymbol{X}_n,Y)$ in Theorem~\ref{theorem-1} combined with
showing $(-1)^{I(Y \leq y_0)} \stackrel{p}{\rightarrow}  (-1)^{I(Y \leq m_0)}$ as $n\to \infty$, where $m_0$ is the mode/maximizer
of $h(y) = f(y;  \boldsymbol{\theta}_0)/ \sup_{\vartheta} f(y;\vartheta,\boldsymbol{\theta}_0^\prime)$ and
$\boldsymbol{\theta}_0 = (\vartheta_0, \boldsymbol{\theta}_0^\prime)$ denotes the true parameter value (in the notation of Remark~2).
Because $Y\sim f(\cdot; \boldsymbol{\theta}_0)$ is a continuous random variable, the convergence   $(-1)^{I(Y \leq y_0)} \stackrel{p}{\rightarrow}  (-1)^{I(Y \leq m_0)}$ follows by showing  $  y_0  \stackrel{p}{\rightarrow}  m_0$; from this, $(Y,y_0) \stackrel{d}{\rightarrow} (Y,m_0)$ then holds and  the continuous mapping theorem  then yields  $(-1)^{I(Y \leq y_0)} \stackrel{p}{\rightarrow}  (-1)^{I(Y \leq m_0)}$.
To establish $  y_0(\boldsymbol{X_n})\equiv y_0  \stackrel{p}{\rightarrow}  m_0$, we pick/fix some small $\epsilon>0$.  Then, the LR statistic $\Lambda_n(\boldsymbol{X}_n,m)\stackrel{p}{\rightarrow} h(m)$ for $m\in\{m_0,\ m_0 \pm \epsilon\}$; this implies, because $h(m_0-\epsilon)<h(m_0)$
and $h(m_0+\epsilon)<h(m_0)$ and because $\Lambda_n(\boldsymbol{X}_n,y)$ is unimodal in $y$,
it must be that the maximizer $y_0$  of $\Lambda_n(\boldsymbol{X}_n,y)$
lies in the interval $(m_0-\epsilon,m_0+\epsilon)$ (with arbitrarily high probability for large $n$) due to the fact
that $\Lambda_n(\boldsymbol{X}_n,y)$ at the endpoints $y=m_0\pm \epsilon$ is smaller than  $\Lambda_n(\boldsymbol{X}_n,y)$ at $y=m_0$.
This shows $\lim_{n\to \infty} \Pr(|y_0(\boldsymbol{X_n}) - m_0|<\epsilon)=1$ and, since $\epsilon>0$ was arbitrary,  $y_0(\boldsymbol{X_n})\equiv y_0  \stackrel{p}{\rightarrow}  m_0$ follows.

\subsection{Proof of Theorem~\ref{theorem-chi-square-1}}
After presenting a proof of Theorem~\ref{theorem-chi-square-1}, we provide some further explanation (cf.~Remark~3 below) of the theoretical details mentioned in Section 6.4 of the main manuscript.  These details concern how Theorem~\ref{theorem-chi-square-1} applies for justifying likelihood ratio statistics used in prediction problems for discrete random variables (e.g., binomial and Poisson predictions from Sections~6.1-6.2).
\begin{theorem}\label{theorem-chi-square-1}
Suppose iid data $X_1,\dots,X_n$ have a marginal density
$f(\cdot;\theta)$, depending on scalar parameter $\theta\in \Theta$, and satisfy Assumptions~(a)-(d) of Theorem~\ref{theorem-1}
with true parameter value $\theta_0$.  Suppose further that, independently,  $Y_1,\dots,Y_m$ are iid random variables with the same marginal density.

Consider a hypothesis test where
the null hypothesis (or reduced model)  is that $X_1,\dots,X_n$ and $Y_1,\dots,Y_m$ have the same density $f(\cdot;\theta)$, while
the alternative hypothesis (or full model)  is that  
\[
X_1,\dots,X_n\sim f(\cdot;\theta),\quad Y_1,\dots Y_m\sim f(\cdot;\theta+\delta);
\]
above $\delta=0$ corresponds to the null hypothesis and the assumed true data distribution.

Denoting the parameter vector as $\boldsymbol{\xi}=(\delta,\theta)$ and the log-likelihood function based on $(X_1,\dots,X_n,Y_1,\dots,Y_m)$ as $l_{n,m}(\boldsymbol{\xi})$, let $\tilde{\boldsymbol\xi} = (0,\tilde{\theta})$ and $\widehat{\boldsymbol\xi} = (\widehat{\delta},\widehat{\theta})$ denote the ML estimators under the reduced and full models, respectively.
Then, as $n,m\to \infty$,
the likelihood ratio statistic has a chi-square limit with 1 degree of freedom, namely
\[
-2\log\Lambda_{n,m}\equiv 2\left[l_{n,m}(\widehat{\boldsymbol{\xi}})-l_{n,m}(\tilde{\boldsymbol{\xi}})\right]\stackrel{d}{\rightarrow} \chi_1^2.
\]
%

\end{theorem}

\begin{proof} Let $\boldsymbol{\xi}_0=(0,\theta_0)$ denote the true parameter. We first provide some notation to describe partial derivatives of the log-likelihood function.  For the density $f(y;\theta+\delta)$ of $Y_i$ under the full model and letting $\eta=\theta+\delta$, note that
\[
\frac{\partial\log f(y;\theta+\delta)}{\partial\theta}=\frac{d\log f(y;\eta)}{d\eta}=\frac{\partial\log f(y;\theta+\delta)}{\partial\delta}, 
\]
and that
\begin{equation}\label{second-for-y}
\frac{\partial^2\log f(y;\theta+\delta)}{\partial\theta^2}=\frac{\partial^2\log f(y;\theta+\delta)}{\partial\theta\partial\delta}=\frac{\partial^2\log f(y;\theta+\delta)}{\partial\delta^2};
\end{equation}
for the density $f(x;\theta)$ of $X_i$ (under full or reduced models), we have
\begin{equation}\label{second-for-x}
\frac{\partial\log f(x;\theta)}{\partial\delta}= \frac{\partial^2\log f(x;\theta)}{d\delta^2}=\frac{\partial^2\log f(x;\theta)}{d\theta d\delta}=\frac{\partial^2\log f(x;\theta)}{ d\delta d\theta}=0.
\end{equation}
Then, the bivariate vector of first partial derivatives  of  the log-likelihood, at $\boldsymbol{\xi}=(\delta,\theta)$, is given by
\[
l_{n,m}^{\prime }(\boldsymbol{\xi}) \equiv \left[\begin{array}{l }
	\displaystyle{\sum_{j=1}^{m} \frac{\partial \log f(Y_j;\theta+\delta)}{\partial\delta }}  \\
	 \displaystyle{\sum_{i=1}^{n}\frac{\partial \log f(X_i;\theta)}{\partial\theta }+\sum_{j=1}^{m}\frac{\partial \log f(Y_i;\theta+\delta)}{\partial\theta }}
\end{array}\right],
\]
while,
by (\ref{second-for-y})-(\ref{second-for-x}), the second partial derivative matrix of the log-likelihood evaluated at $\boldsymbol{\xi}=(\delta,\theta)$, is given by
\[
l_{n,m}^{\prime\prime}(\boldsymbol{\xi})\equiv \left[\begin{array}{lcl}
	\displaystyle{\sum_{j=1}^{m} \frac{\partial^2\log f(Y_j;\theta+\delta)}{\partial\delta^2}} &&\displaystyle{ \sum_{j=1}^{m}\frac{\partial^2\log f(Y_j;\theta+\delta)}{\partial\theta\partial\delta} }\\
	\displaystyle{\sum_{j=1}^{m}\frac{\partial^2\log f(Y_j;\theta+\delta)}{\partial\theta\partial\delta} }& & \displaystyle{\sum_{i=1}^{n}\frac{\partial^2\log f(X_i;\theta)}{\partial\theta^2}+\sum_{j=1}^{m}\frac{\partial^2\log f(Y_i;\theta+\delta)}{\partial\theta^2}}
\end{array}\right].
\]
Furthermore, at the true parameter $\boldsymbol{\xi}_0=(0,\theta_0)$, we have
\[
\text{E}_{\theta_0} \frac{\partial\log f(X_1;\theta_0)}{\partial\theta} = \text{E}_{\theta_0} \frac{\partial\log f(Y_1;\theta_0)}{\partial\theta}=
\text{E}_{\theta_0} \frac{\partial\log f(Y_1;\theta_0)}{\partial\delta} = 0,
\]
while
\begin{eqnarray}
\label{second-equal}
   \text{E}_{\theta_0} \frac{\partial^2\log f(X_1;\theta_0)}{\partial\theta^2} =    \text{E}_{\theta_0} \frac{\partial^2\log f(Y_1;\theta_0)}{\partial\theta^2} &=&  \text{E}_{\theta_0} \frac{\partial^2\log f(Y_1;\theta_0)}{\partial\theta \partial \delta}\\
 \nonumber & =&  \text{E}_{\theta_0} \frac{\partial^2\log f(Y_1;\theta_0)}{\partial \delta \partial\theta} = -I_{\theta_0}
\end{eqnarray}
holds along with
\begin{equation}
\label{second-equal2}
\text{Var}_{\theta_0}\left( \frac{\partial\log f(X_1;\theta_0)}{\partial\theta} \right) = \text{Var}_{\theta_0}\left( \frac{\partial\log f(Y_1;\theta_0)}{\partial\theta} \right)= \text{Var}_{\theta_0}\left(  \frac{\partial\log f(Y_1;\theta_0)}{\partial\delta} \right)=I_{\theta_0}
\end{equation}
  for an information number $I_{\theta_0}\in (0,\infty)$.

 Under the conditions, both $\tilde{\boldsymbol\xi} \stackrel{p}{\rightarrow} \boldsymbol{\xi}_0$ and $\widehat{\boldsymbol\xi} \stackrel{p}{\rightarrow} \boldsymbol{\xi}_0$ hold as $n,m\to \infty$
  so that estimators lie in a neighborhood of $\boldsymbol{\xi}_0=(0,\theta_0)$ (for large $m,n$) and the log-likelihood $l_{n,m}(\boldsymbol{\xi})$
is twice continuously differentiable in this parameter neighborhood.
To determine the distributional limit of the log-likelihood ratio statistic
\[
-2\log\Lambda_{n,m}=2\left[l_{n,m}(\widehat{\boldsymbol\xi})-l_{n,m}(\tilde{\boldsymbol\xi})\right],
\]
we may expand $l_{n,m}(\tilde{\boldsymbol\xi})$ at $\widehat{\boldsymbol\xi}$ to find
\[
l_{n,m}(\tilde{\boldsymbol\xi})=l_{n,m}(\widehat{\boldsymbol\xi})+l_{n,m}^\prime(\widehat{\boldsymbol\xi})(\tilde{\boldsymbol\xi}-\widehat{\boldsymbol\xi})+\frac{1}{2}(\tilde{\boldsymbol\xi}-\widehat{\boldsymbol\xi})^Tl_{n,m}^{\prime\prime}(\boldsymbol\xi^\ast)(\tilde{\boldsymbol\xi}-\widehat{\boldsymbol\xi}),
\]
where $\boldsymbol\xi^\ast$ is between $\tilde{\boldsymbol\xi}$ and $\widehat{\boldsymbol\xi}$.
Because $l_{n,m}^\prime(\widehat{\boldsymbol\xi})=\boldsymbol{0}\equiv (0,0)^{T}$ holds for the maximizer $\widehat{\boldsymbol\xi}$, the likelihood ratio statistic may be written as
\begin{equation}\label{eq-lik-ratio}
-2\log\Lambda_{n,m}=(n+m)(\tilde{\boldsymbol\xi}-\widehat{\boldsymbol\xi})^T\left[-\frac{l_{n,m}^{\prime\prime}(\boldsymbol\xi^\ast)}{n+m}\right](\tilde{\boldsymbol\xi}-\widehat{\boldsymbol\xi}).
\end{equation}

Let
$\{(n_j,m_j)\}$ denote an arbitrary subsequence of $\{(n,m)\}$.   Then, there exists a further subsequence $\{(n_k,m_k)\}\subset \{(n_j,m_j)\}$ and a value $c\in [0,1]$ such that $m_k/(m_k+n_k)\rightarrow c$ as $k \to \infty$, due to the boundedness of $\{m/(n+m)\}$.
 To show the chi-square limit of the log-likelihood ratio statistic (\ref{eq-lik-ratio}), it suffices to establish   that $-2\log\Lambda_{n_k,m_k} \stackrel{d}{\rightarrow} \chi_1^2$ holds.  We establish this by considering three cases $c\in (0,1)$, $c=0$,  or $c=1$.   For simplicity in the following, we will suppress the subsequence notation and consider sample sizes ``$(n,m)$" in place of ``$(n_k,m_k)$."
By using the smoothness conditions and the law of large numbers along with (\ref{second-equal}) and $m/(m+n)\to c\in [0,1]$ in (\ref{eq-lik-ratio}), we have
\begin{equation}\label{eq-fisher-inf}
-\frac{l_{n,m}^{\prime\prime}(\boldsymbol\xi^\ast)}{n+m}\xrightarrow{p} \begin{bmatrix}
c & c\\
c & 1
\end{bmatrix} I_{\theta_0}   \equiv \boldsymbol{I}(\boldsymbol \xi_0).
\end{equation}
We next discuss the cases where (i) $0<c<1$; (ii) $c=0$; or (iii) $c=1$.

\noindent\textbf{(i)} When $0<c<1$, note that the matrix $\boldsymbol{I}(\boldsymbol \xi_0)$ is invertible.
From (\ref{eq-lik-ratio}) and (\ref{eq-fisher-inf}), the limit distribution of the log-likelihood statistic will follow from the distribution of
$\sqrt{n+m}(\Tilde{\boldsymbol\xi}-\widehat{\boldsymbol\xi})$.  To determine the latter,
we start by expanding $l^\prime_{n,m}(\Tilde{\boldsymbol\xi})$ at $\widehat{\boldsymbol\xi}$ to find
\begin{equation}\label{eq-expand-2}
    \frac{1}{\sqrt{n+m}}l^\prime_{n,m}(\Tilde{\boldsymbol\xi})=\frac{1}{\sqrt{n+m}}l^\prime_{n,m}(\widehat{\boldsymbol\xi})+\frac{l_{n,m}^{\prime\prime}(\boldsymbol\xi^{\ast}_2)}{n+m}\sqrt{n+m}(\Tilde{\boldsymbol\xi}-\widehat{\boldsymbol\xi}),
\end{equation}
where $\boldsymbol{\xi}^{\ast}_2$ denotes some value between $\tilde{\boldsymbol{\xi}}$ and $\widehat{\boldsymbol{\xi}}$.
Because $l^\prime_{n,m}(\widehat{\boldsymbol{\xi}})=\textbf0$ holds and because $[-{l_{n,m}^{\prime\prime}(\boldsymbol\xi^{\ast}_2)}/{(n+m)}]^{-1} \xrightarrow{p} [\boldsymbol{I}(\boldsymbol \xi_0)]^{-1}$ as in (\ref{eq-fisher-inf}), we may find the limit distribution of $l^\prime_{n,m}(\Tilde{\boldsymbol\xi})/\sqrt{n+m}$ to determine the distribution of $\sqrt{n+m}(\Tilde{\boldsymbol\xi}-\widehat{\boldsymbol\xi})$ in (\ref{eq-expand-2}).

We define a matrix $H$ as
\[
H=\begin{bmatrix}
0 & 0\\
0 & 1/I_{\theta_0}
\end{bmatrix}.
\]
so that
\[
Hl_{n,m}^{\prime}(\tilde{\boldsymbol{\xi}})=\textbf{0}
\]
holds because the second element of $l_{n,m}^\ast(\Tilde{\boldsymbol\xi})$ is 0 (maximizing the likelihood with respect to $\theta$ when $\delta=0$).
We then expand $l^\prime_{n,m}(\Tilde{\boldsymbol\xi})/\sqrt{n+m}$  at $\boldsymbol\xi_0$ to find
\begin{equation}\label{eq:expand-1}
\frac{l_{n,m}^\prime(\Tilde{\boldsymbol\xi})}{\sqrt{n+m}}=\frac{l_{n,m}^\prime(\boldsymbol\xi_0)}{\sqrt{n+m}}+\frac{l_{n,m}^{\prime\prime}(\boldsymbol\xi^{\ast}_3)}{n+m}\sqrt{n+m}(\tilde{\boldsymbol\xi}-\boldsymbol\xi_0),
\end{equation}
where $\boldsymbol{\xi}^{\ast}_3$ is some value between $\tilde{\boldsymbol{\xi}}$ and $\boldsymbol{\xi}$.
By multiplying (\ref{eq:expand-1}) with $H$, we have
\[
H\frac{l_{n,m}^\prime(\Tilde{\boldsymbol\xi})}{\sqrt{n+m}}=H\frac{l_{n,m}^\prime(\boldsymbol\xi_0)}{\sqrt{n+m}}+H\frac{l_{n,m}^{\prime\prime}
(\boldsymbol\xi^{\ast}_3)}{n+m}\sqrt{n+m}(\tilde{\boldsymbol{\xi}}-\boldsymbol{\xi}_0)=\boldsymbol0;
\]
 using $\|-l_{n,m}^\prime(\boldsymbol\xi_3^*)/(n+m) -\boldsymbol{I}(\boldsymbol \xi_0)\| =o_p(1)
$  as in
(\ref{eq-fisher-inf}) along with $\|\tilde{\boldsymbol{\xi}} -\boldsymbol{\xi}_0\|=\|(0,\tilde{\theta}) - (0,\theta_0)\|=|\tilde{\theta}  - \theta_0| = O_p((n+m)^{-1/2})$ (this order following from the CLT applied to the second entry in (\ref{eq:expand-1})), we may further
re-write as
\begin{eqnarray}
\label{eq-whatever-key}	
H\frac{l_{n,m}^\prime(\boldsymbol\xi_0)}{\sqrt{n+m}} + \boldsymbol{R}_{n,m}  &
=&H \boldsymbol{I}(\boldsymbol \xi_0) \sqrt{n+m}(\tilde{\boldsymbol{\xi}}-\boldsymbol{\xi}_0)
\\
\nonumber &=& \sqrt{n+m}H \boldsymbol{I}(\boldsymbol \xi_0)  (\tilde{\boldsymbol{\xi}}-\boldsymbol{\xi}_0) \\
\nonumber &=&
\sqrt{n+m}\begin{bmatrix}
	0 & 0\\
	0 & 1/I_{\theta_0}
\end{bmatrix}
I_{\theta_0}
\begin{bmatrix}
	c & c\\
	c & 1
\end{bmatrix}
\begin{bmatrix}
	0\\
	\tilde{\theta}-\theta_0
\end{bmatrix}\\
\nonumber&=&\sqrt{n+m}(\tilde{\boldsymbol{\xi}}-\boldsymbol{\xi}_0),
\end{eqnarray}
where bivariate $\boldsymbol{R}_{n,m}$ denotes a remainder term $\|\boldsymbol{R}_{n,m}\|=o_p(1)$.
By replacing this form  (\ref{eq-whatever-key})  of $\sqrt{n+m}(\tilde{\boldsymbol{\xi}}-\boldsymbol{\xi}_0)$ in  (\ref{eq:expand-1})
and using $\|-l_{n,m}^\prime(\boldsymbol\xi_3^*)/(n+m) - \boldsymbol{I}(\boldsymbol \xi_0) \|=o_p(1)$ again,
we have
\begin{equation}\label{eqn:end2}
\frac{l_{n,m}^\prime(\Tilde{\boldsymbol\xi})}{\sqrt{n+m}}=\frac{\textbf{I}-\boldsymbol{I}(\boldsymbol\xi_0)H}{\sqrt{n+m}}
l_{n,m}^\prime(\boldsymbol\xi_0)+
  \tilde{\boldsymbol{R}}_{n,m}
\end{equation}
where $\textbf{I}$ is a $2\times2$ identity matrix and   $\tilde{\boldsymbol{R}}_{n,m}$ is a bivariate remainder term $\|\tilde{\boldsymbol{R}}_{n,m}\|=o_p(1)$.
Using the standard CLT with (\ref{second-equal2}), we have that
\[
\frac{l^\prime_{n,m}(\boldsymbol\xi_0)}{\sqrt{n+m}}=\sqrt{n+m}
\left(\frac{l^\prime_{n,m}(\boldsymbol\xi_0)}{n+m}-\boldsymbol0\right)\xrightarrow{d}W \sim \text{MVN}\left(\boldsymbol0,{\boldsymbol{I}(\boldsymbol\xi_0)}\right)
\]
as $n,m\to\infty$ due to independence and $m/(n+m)\to c$, where $W$ has a bivariate normal distribution; the continuous mapping theorem with (\ref{eqn:end2}) then gives
\begin{equation}\label{eq-almost-there}
\frac{l_{n,m}^\prime(\tilde{\boldsymbol{\xi}})}{\sqrt{n+m}}\xrightarrow{d}\left[\textbf{I}-\boldsymbol{I}(\boldsymbol\xi_0)H\right]W.
\end{equation} 
Then, by applying  (\ref{eq-almost-there}) with $\|-l_{n,m}^\prime(\boldsymbol\xi_2^*)/(n+m) -\boldsymbol{I}(\boldsymbol \xi_0)\| =o_p(1)
$    in (\ref{eq-expand-2}),    the limit distribution of $\sqrt{n+m}(\Tilde{\boldsymbol\xi}-\widehat{\boldsymbol\xi})$ follows as
\[
\sqrt{n+m}(\Tilde{\boldsymbol\xi}-\widehat{\boldsymbol\xi})\xrightarrow{d}-[\boldsymbol{I}(\boldsymbol\xi_0)]^{-1}[\textbf{I}-\boldsymbol{I}(\boldsymbol\xi_0)H]W,
\] 
so that  the log-likelihood ratio statistic in (\ref{eq-lik-ratio}) has an asymptotic distribution as
\[
\begin{split}
-2\log\Lambda_{n,m}&\xrightarrow{d} W^T[\textbf{I}-\boldsymbol{I}(\boldsymbol\xi_0)H]^T [\boldsymbol{I}(\boldsymbol\xi_0)]^{-1}[\textbf{I}-\boldsymbol{I}(\boldsymbol\xi_0)H]W\\
&=Z^T (\textbf{I} -[\boldsymbol{I}(\boldsymbol\xi_0)]^{1/2} H [\boldsymbol{I}(\boldsymbol\xi_0)]^{1/2}) Z \sim \chi_1^2
\end{split}
\]
using that $Z\sim\text{MVN}(0,\textbf{I})$ and  that $\textbf{I} -[\boldsymbol{I}(\boldsymbol\xi_0)]^{1/2} H [\boldsymbol{I}(\boldsymbol\xi_0)]^{1/2}$ is an idempotent matrix of rank/trace $1$.

\noindent\textbf{(ii)} When $c=0$, we use a similar argument to find
 $\|\tilde{\boldsymbol{\xi}} -\boldsymbol{\xi}_0\|=\|(0,\tilde{\theta}) - (0,\theta_0)\|=|\tilde{\theta}  - \theta_0| = O_p((n+m)^{-1/2})$ (following from the CLT and law of large numbers applied to the second entry in (\ref{eq:expand-1})); note this implies
   $\sqrt{m}\|\tilde{\boldsymbol{\xi}} -\boldsymbol{\xi}_0\| = o_p(1)$ since $m/(n+m)\rightarrow c=0$.  Now we multiply 
 (\ref{eq:expand-1}) by $\sqrt{n+m}/\sqrt{m}$ and consider  the first components, say   $l_{n,m}^{\prime(1)}(\Tilde{\boldsymbol\xi})$ and $l_{n,m}^{\prime(1)}(  \boldsymbol{\xi}_0)$, of
  $l_{n,m}^{\prime }(\Tilde{\boldsymbol\xi})$ and $l_{n,m}^{\prime}(  \boldsymbol{\xi}_0)$ in (\ref{eq:expand-1}), respectively,
  along with the first row, say  $l_{n,m}^{\prime\prime (1) }(\boldsymbol{\xi}^*_3)$, of   $l_{n,m}^{\prime\prime   }(\boldsymbol{\xi}^*_3)$
  in (\ref{eq:expand-1});  we then have from (\ref{eq:expand-1}) that
 \begin{eqnarray}
\nonumber   \frac{l_{n,m}^{\prime(1)}(\Tilde{\boldsymbol\xi})}{\sqrt{m}} &=&    \frac{l_{n,m}^{\prime(1)}(  \boldsymbol{\xi}_0)}{\sqrt{m}} +
   \frac{l_{n,m}^{\prime\prime (1) }(\boldsymbol{\xi}^*_3)}{m} \sqrt{m}(\tilde{\boldsymbol{\xi}} -\boldsymbol{\xi}_0)\\
  \label{eqn:gr2} & \stackrel{d}{\rightarrow}& N(0,I_{\theta_0})
  \end{eqnarray}
  as $n,m\to \infty$, using above that $l_{n,m}^{\prime(1)}(  \boldsymbol{\xi}_0)/\sqrt{m} \stackrel{d}{\rightarrow} N(0,I_{\theta_0})$
  by the standard CLT along with  $ l_{n,m}^{\prime\prime (1) }(\boldsymbol{\xi}^*_3)/m \stackrel{p}{\rightarrow} (I_{\theta_0},I_{\theta_0})$
  by the law of large numbers and with $\sqrt{m}\|\tilde{\boldsymbol{\xi}} -\boldsymbol{\xi}_0\| = o_p(1)$.
  We now write the $2\times 2$ matrix $l_{n,m}^{\prime\prime   }(\boldsymbol{\xi}^*_2)$ in (\ref{eq-expand-2}) as 
  \[
  l_{n,m}^{\prime\prime   }(\boldsymbol{\xi}^*_2) = \left[\begin{array}{lcl}
    a_{n,m,1} && a_{n,m,2}\\
    a_{n,m,2} && a_{n,m,3}
  \end{array}\right],
  \]
  noting that the component sample averages satisfy (for the information number $I_{\theta_0}\in (0,\infty)$)
    \begin{equation}
  \label{eqn:gr}
   \frac{a_{n,m,1}}{m}\stackrel{p}{\rightarrow}-I_{\theta_0}, \qquad \frac{a_{n,m,2}}{m} \stackrel{p}{\rightarrow}-I_{\theta_0}, \qquad
  \frac{a_{n,m,3}}{n+m} \stackrel{p}{\rightarrow}-I_{\theta_0}
   \end{equation} (i.e., as $a_{n,m,3}$ involves a sum of $n+m$ terms while $ a_{n,m,1}, a_{n,m,2}$ are sums of $m$ terms).  Recalling that $(\Tilde{\boldsymbol\xi}-\widehat{\boldsymbol\xi}) = ( -\widehat{\delta}, \tilde{\theta}-\widehat{\theta})$, the second component of (\ref{eq-expand-2}) entails that $\tilde{\theta}-\widehat{\theta} =  \widehat{\delta} a_{n,m,2}/a_{n,m,3}$; substitution of this quantity into
  the first component of (\ref{eq-expand-2}) then gives that 
  \begin{eqnarray*}
   \frac{l_{n,m}^{\prime(1)}(\Tilde{\boldsymbol\xi})}{\sqrt{m}} &=& \frac{a_{n,m,2}}{m}  \sqrt{m}(\tilde{\theta}-\widehat{\theta}) -\frac{a_{n,m,1}}{m} \sqrt{m}\widehat{\delta}\\
   &=& \left(\frac{a_{n,m,2}}{m}  \frac{a_{n,m,2}}{a_{n,m,3}} -\frac{a_{n,m,1}}{m} \right) \sqrt{m}\widehat{\delta}.
  \end{eqnarray*}
  From this with (\ref{eqn:gr}) and $m/(n+m)\to 0$, we have 
  \begin{eqnarray*}
 \frac{a_{n,m,2}}{m}  \frac{a_{n,m,2} }{a_{n,m,3} } -\frac{a_{n,m,1}}{m} & = &\frac{m}{n+m} \frac{a_{n,m,2}}{m}  \frac{a_{n,m,2}/m}{a_{n,m,3}/(n+m)} -\frac{a_{n,m,1}}{m} \\&=& o_p(1)-\frac{a_{n,m,1}}{m} = I_{\theta_0}(1+o_p(1))
  \end{eqnarray*}
   and then that
    \begin{equation}
  \label{eqn:gr3}
 \sqrt{m}\widehat{\delta} \stackrel{d}{\rightarrow} N(0,1/I_\theta) 
   \end{equation}
  by (\ref{eqn:gr2}) and Slutsky's theorem.
  Now similarly write the $2\times 2$ matrix $l_{n,m}^{\prime\prime   }(\boldsymbol{\xi}^*)$ in (\ref{eq-lik-ratio}) as
   \[
  l_{n,m}^{\prime\prime   }(\boldsymbol{\xi}^*_2) =  \left[\begin{array}{lcl}
    \tilde{a}_{n,m,1} && \tilde{a}_{n,m,2}\\
    \tilde{a}_{n,m,2} && \tilde{a}_{n,m,3}
  \end{array}\right],
  \]
 where  (\ref{eqn:gr}) likewise  holds for the counterpart sample averages $ \tilde{a}_{n,m,1}/m$, $\tilde{a}_{n,m,2}/m$, $\tilde{a}_{n,m,3}/(n+m)$;
  then re-writing (\ref{eq-lik-ratio}) in terms of  $(\Tilde{\boldsymbol\xi}-\widehat{\boldsymbol\xi}) = ( -\widehat{\delta}, \tilde{\theta}-\widehat{\theta})$, where $\tilde{\theta}-\widehat{\theta} =  \widehat{\delta} a_{n,m,2}/a_{n,m,3}$ again, we have
  \begin{eqnarray*}
 -2\log\Lambda_{n,m}&=& - m(\tilde{\boldsymbol\xi}-\widehat{\boldsymbol\xi})^T\left[ \frac{l_{n,m}^{\prime\prime}(\boldsymbol\xi^\ast)}{m}\right](\tilde{\boldsymbol\xi}-\widehat{\boldsymbol\xi}).
  \\&=&
  -m (\widehat{\delta})^2 \left[\begin{array}{cc}
  -1 & \frac{a_{n,m,2}}{a_{n,m,3}}
\end{array}
   \right]\left[\begin{array}{lcl}
    \tilde{a}_{n,m,1}/m && \tilde{a}_{n,m,2}/m\\
    \tilde{a}_{n,m,2}/m && \tilde{a}_{n,m,3}/m
  \end{array}\right] \left[\begin{array}{c}
  -1 \\ \frac{a_{n,m,2}}{a_{n,m,3}}
\end{array}
  \right]
  \end{eqnarray*}
in matrix form.
Note that $a_{n,m,2}/a_{n,m,3} = m/(m+n)\cdot (a_{n,m,2}/m)/(a_{n,m,3}/(n+m))= o_p(1)$ while 
\[
\frac{\tilde{a}_{n,m,3}}{m} \frac{a_{n,m,2}}{a_{n,m,3}}= \frac{a_{n,m,2}}{m} \frac{\tilde{a}_{n,m,3}/(n+m)}{a_{n,m,3}/(n+m)}\stackrel{p}{\rightarrow}  -I_{\theta_0},
\] 
and $ \tilde{a}_{n,m,1}/m  \stackrel{p}{\rightarrow}  -I_{\theta_0} $, $ \tilde{a}_{n,m,2}/m  \stackrel{p}{\rightarrow}  -I_{\theta_0} $
as in (\ref{eqn:gr}); consequently, we then have
\[-2\log\Lambda_{n,m} = m(\widehat{\delta})^2  I_{\theta_0}(1+o_p(1)) \stackrel{d}{\rightarrow} \chi_1^2, 
  \]
 by (\ref{eqn:gr3}) and the continuous mapping theorem (i.e.,  $\sqrt{I_{\theta_0}}\sqrt{m} \widehat{\delta}$ has a standard normal limit).

 \noindent\textbf{(iii)} The argument when $c=1$ is the same of that for $c=0$ upon reversing the roles of $X_1,\ldots,X_n$
 and $Y_1,\ldots,Y_{m}$; that is, we construct the log-likelihood ratio statistic assuming a full model with $X_1,\ldots,X_n \sim f(\cdot;\delta+\theta)$ and $Y_1,\ldots,Y_m \sim f(\cdot; \theta)$ and with a reduced model where $X_1,\ldots,X_n, Y_1,\ldots,Y_m \sim f(\cdot; \theta)$.  This does not change the likelihood ratio statistic, though the proof when $c=0$ (or $m/(n+m)\rightarrow 0$) now applies
 for the case $c=1$    with the roles of sample sizes $n,m$   reversed (i.e.,   $n/(n+m)\rightarrow 0$ holds in the $m/(n+m)\rightarrow c=1$ case).
\end{proof}
\noindent\textbf{Remark~3}.  In support of Section~6.4 of the main manuscript, here we provide   details regarding how Theorem~\ref{theorem-chi-square-1}
applies for establishing the chi-square limit (i.e., namely $\chi_1^2$) of the log-likelihood ratio statistic
for prediction in some important discrete data cases.  

The binomial prediction problem from Section~6.1 involves the 
prediction of $Y \sim \mbox{Binom}(m,p)$ from available data $X   \sim \mbox{Binom}(n,p)$, where $n,m$ are integers and $p \in (0,1)$.  The log-likelihood statistic $-2 \log \Lambda_{m,n}$ based on $(X,Y)$ from Section~6.1 is the same as the log-likelihood statistic $-2 \log \Lambda_{m,n}$ from Theorem~\ref{theorem-chi-square-1},
when the latter is based on iid $X_1,\ldots,X_{n},Y_1,\ldots,Y_m\sim\mbox{Binom}(1,p)$ (i.e., involving a comparison of full and reduced models
as $X_1,\ldots,X_{n}\sim\mbox{Binom}(m,p)$ with $Y_1,\ldots,Y_{m}\sim\mbox{Binom}(m,p_1 = p+\delta)$ for the full model vs.~$X_1,\ldots,X_{n},Y_1,\ldots,Y_m\sim\mbox{Binom}(1,p)$ for the reduced model).
Consequently, as $m,n\to \infty$, the limit of the log-likelihood statistic $-2 \log \Lambda_{m,n} \stackrel{d}{\rightarrow}\chi_1^2$ follows from 
Theorem~\ref{theorem-chi-square-1}.
The equivalence of likelihood-statistics owes to the fact that $(X,Y)$ has the same distribution as $(\sum_{i=1}^n X_i, \sum_{i=1}^m Y_i)$ here.

Similarly, the Poisson prediction problem from Section~6.2 involves the
prediction of $Y \sim \mbox{Poi}(m\lambda)$ from available data $X \sim \mbox{Poi}(n\lambda)$, where $n,m$ are integers and $\lambda>0$. The log-likelihood statistic $-2 \log \Lambda_{m,n}$ given from $(X,Y)$ in Section~6.2 is likewise the same as the log-likelihood statistic $-2 \log \Lambda_{m,n}$ from Theorem~\ref{theorem-chi-square-1},
when the latter is based on iid $X_1,\ldots,X_{n},Y_1,\ldots,Y_m\sim\mbox{Poi}(\lambda)$.    Hence, in the Poisson case, the limit of   $-2 \log \Lambda_{m,n} \stackrel{d}{\rightarrow}\chi_1^2$ again follows from
Theorem~\ref{theorem-chi-square-1}  as $m,n\to \infty$.   Note here  $(X,Y)$ has again the same distribution as $(\sum_{i=1}^n X_i, \sum_{i=1}^m Y_i)$.

The within-sample prediction problem of Section~6.3 involves predicting a binomial count $Y \equiv \sum_{i=1}^n I(T_i \in (t_c, t_w])$ based on
event time random variables $T_1,\ldots,T_n$.  Here there are two counts (i.e., discrete random variables) involving 
the observed number, say $X \equiv \sum_{i=1}^n I(T_i \leq t_c)$, of times occurring before a censoring point $t_c$ along with the number of times $Y \equiv \sum_{i=1}^n I(T_i \in (t_c, t_w])$ occurring in a future interval $(t_c,t_w]$.  (Technically, the value of $T_i$ is also assumed to be available whenever  $T_i \leq t_c$ occurs.)  The structure of this prediction problem is similar to the binomial prediction case, 
though the counts $(X,Y)$ here are multinomial (instead of two independent binomial variables).  The proof of the chi-square limit of the log-likelihood statistic $-2 \log \Lambda_{m,n} \stackrel{d}{\rightarrow}\chi_1^2$ follows with an argument similar to that of  Theorem~\ref{theorem-chi-square-1}.

\newpage
\section{Simulation Results}
\label{sec:simulation-results}

This section provides simulation results for other factor combinations in Section~6.3.2.
\begin{figure}[ht]
	\centering
	\includegraphics[width=\textwidth]{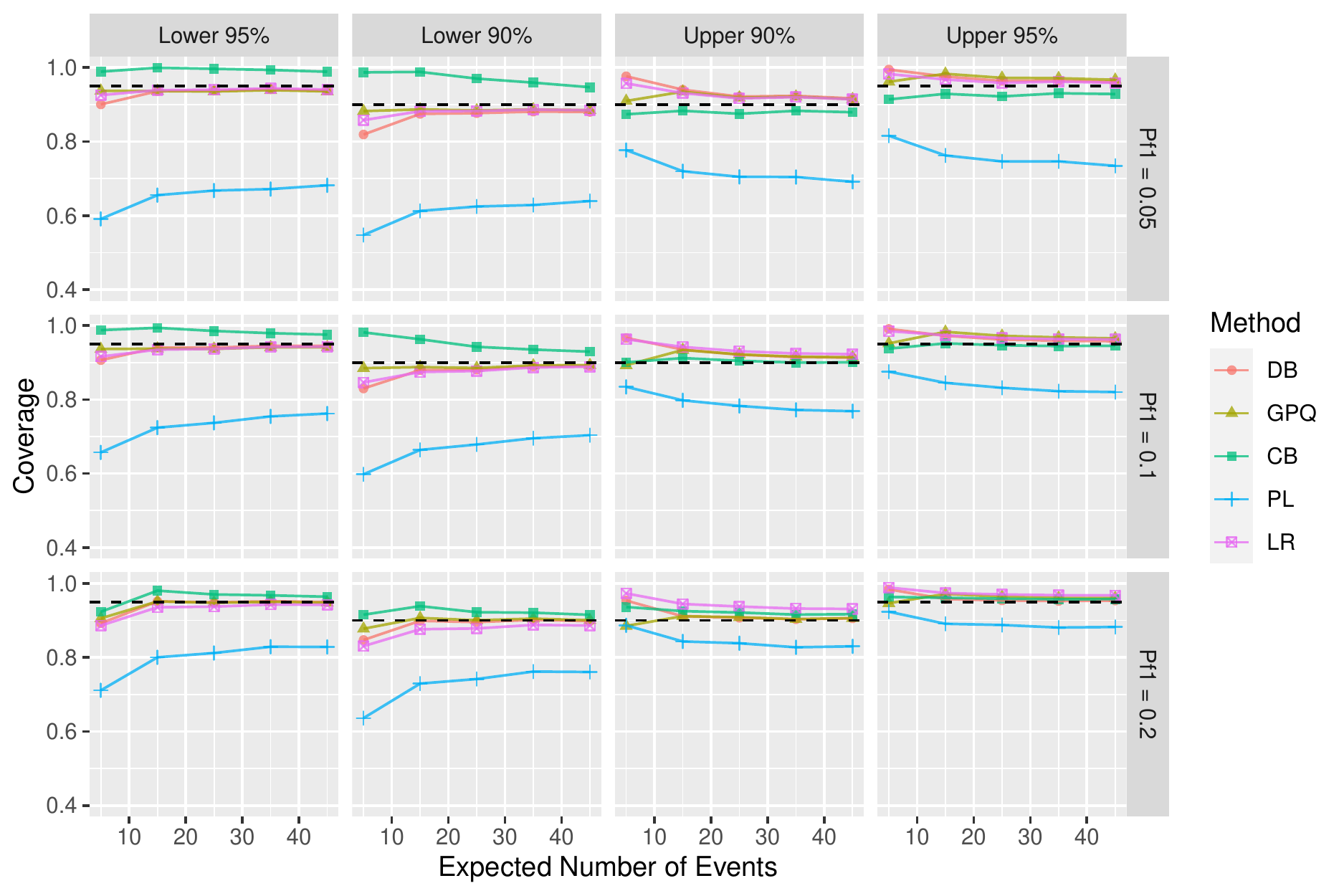}
	\caption{Coverage probabilities versus expected number of events (failures) for the direct-bootstrap (DB), GPQ-bootstrap (GPQ), calibration-bootstrap (CB), LR (LR), and plug-in (PL) methods when $d=0.1$ and $\beta=0.8$.}
	\label{fig:within-sample-pred-1123123}
\end{figure}
\begin{figure}[ht]
	\centering
	\includegraphics[width=\textwidth]{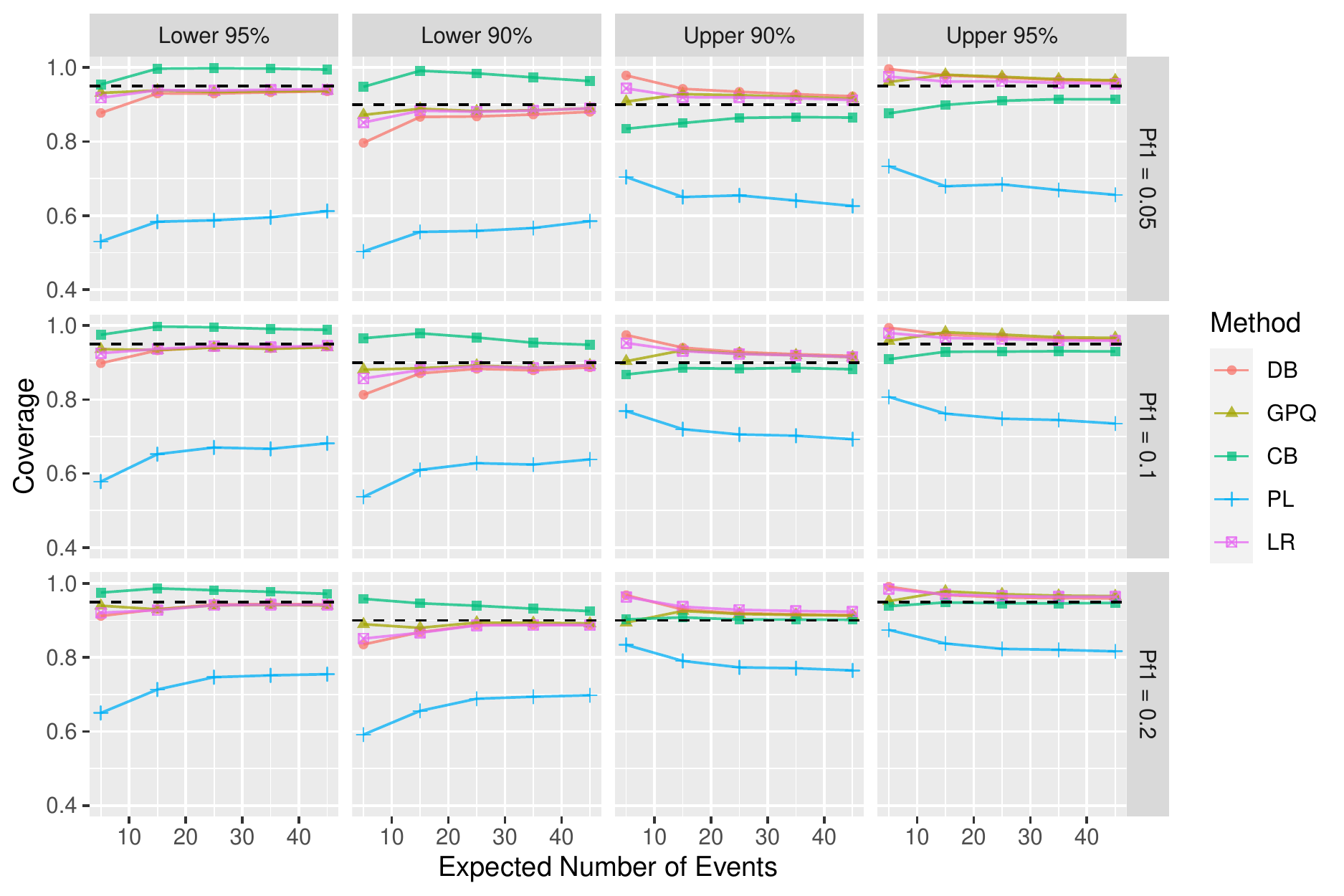}
	\caption{Coverage probabilities versus expected number of events (failures) for the direct-bootstrap (DB), GPQ-bootstrap (GPQ), calibration-bootstrap (CB), LR (LR), and plug-in (PL) methods when $d=0.2$ and $\beta=0.8$.}
	\label{fig:within-sample-pred-14341}
\end{figure}
\begin{figure}[t!]
	\centering
	\includegraphics[width=\textwidth]{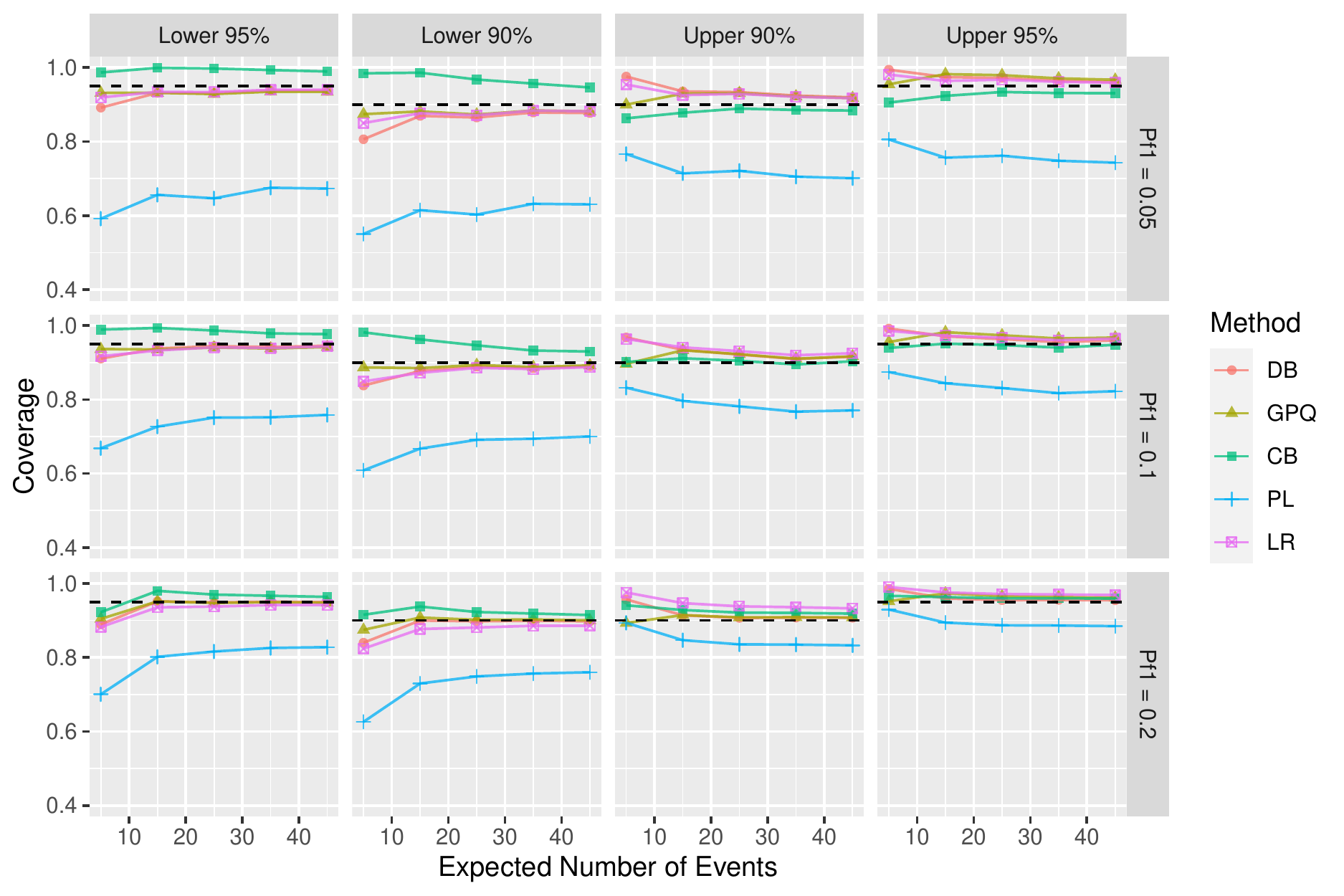}
	\caption{Coverage probabilities versus expected number of events (failures) for the direct-bootstrap (DB), GPQ-bootstrap (GPQ), calibration-bootstrap (CB), LR (LR), and plug-in (PL) methods when $d=0.1$ and $\beta=1$.}
	\label{fig:within-sample-pred-1434}
\end{figure}
\begin{figure}[t!]
	\centering
	\includegraphics[width=\textwidth]{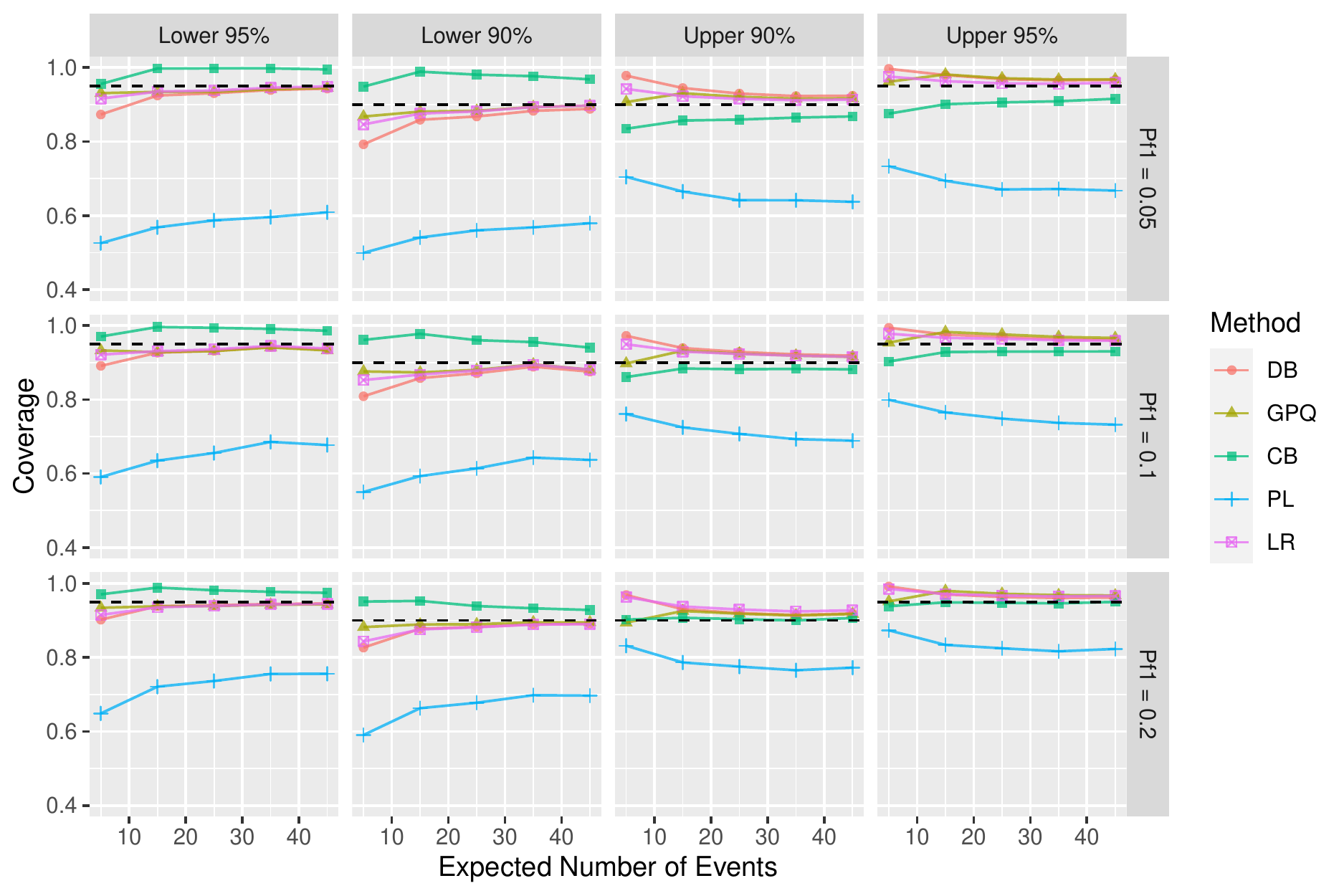}
	\caption{Coverage probabilities versus expected number of events (failures) for the direct-bootstrap (DB), GPQ-bootstrap (GPQ), calibration-bootstrap (CB), LR (LR), and plug-in (PL) methods when $d=0.2$ and $\beta=1$.}
	\label{fig:within-sample-pred-121}
\end{figure}
\begin{figure}[t!]
	\centering
	\includegraphics[width=\textwidth]{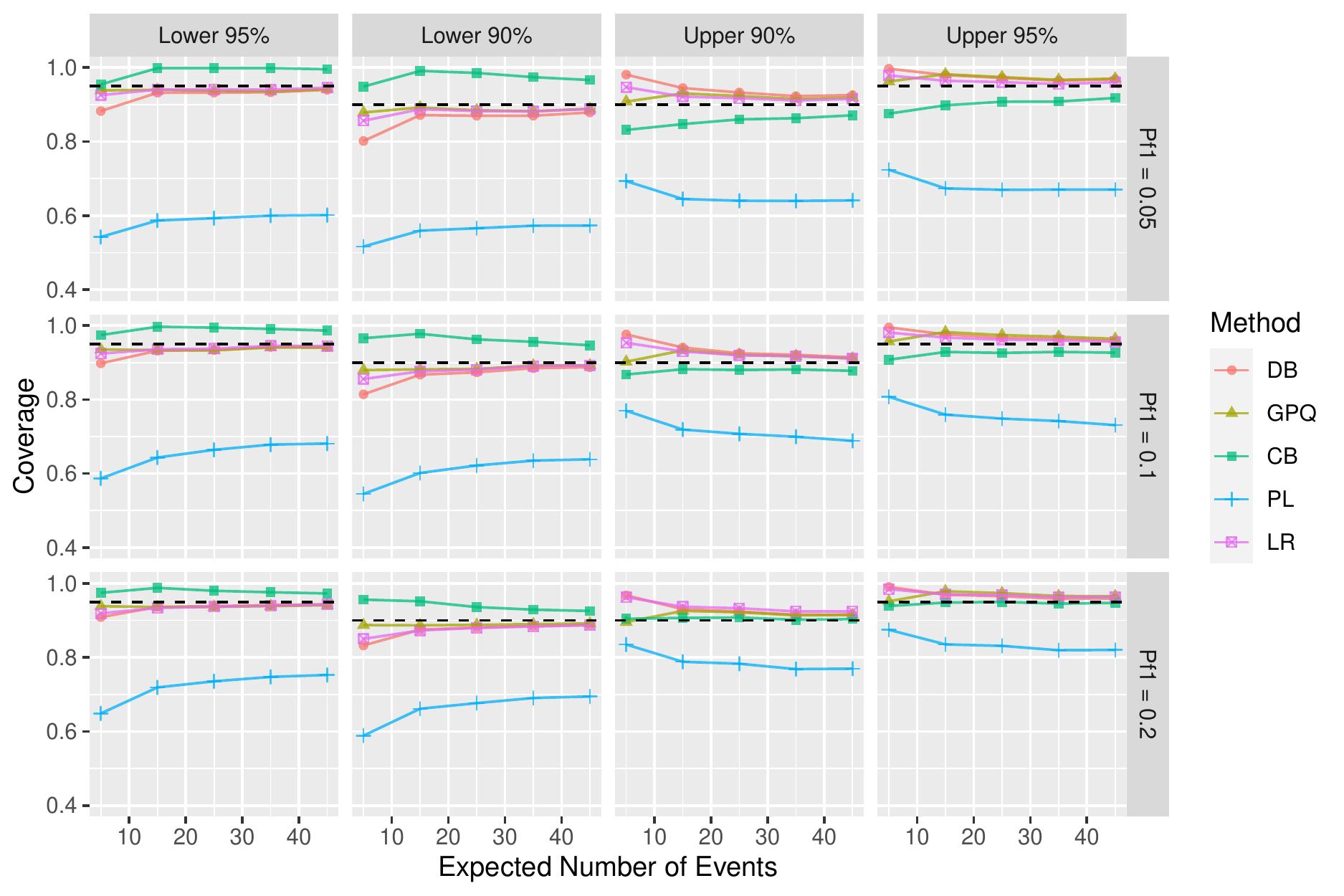}
	\caption{Coverage probabilities versus expected number of events (failures) for the direct-bootstrap (DB), GPQ-bootstrap (GPQ), calibration-bootstrap (CB), LR (LR), and plug-in (PL) methods when $d=0.2$ and $\beta=2$.}
	\label{fig:within-sample-pred-341}
\end{figure}
\begin{figure}[t!]
	\centering
	\includegraphics[width=\textwidth]{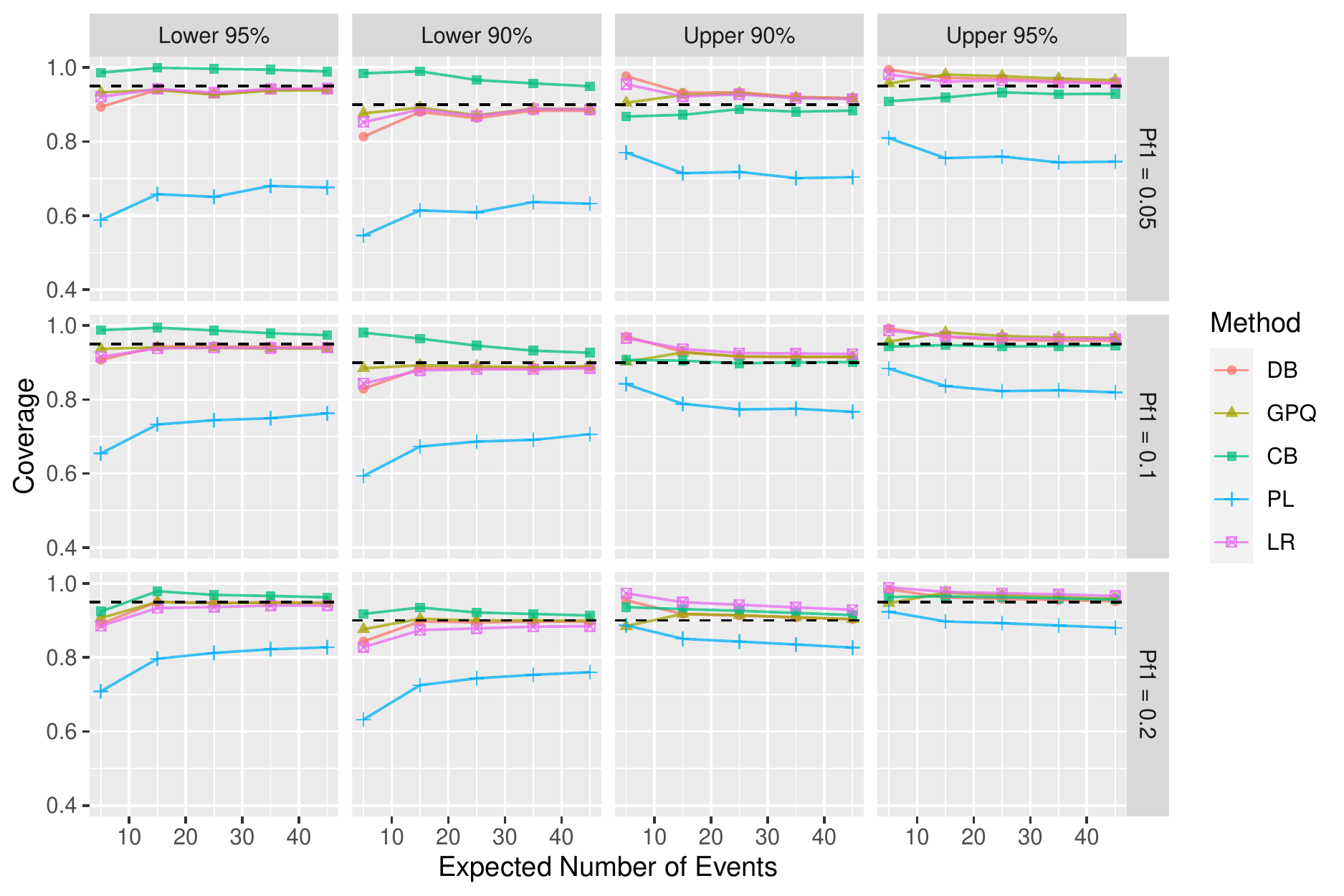}
	\caption{Coverage probabilities versus expected number of events (failures) for the direct-bootstrap (DB), GPQ-bootstrap (GPQ), calibration-bootstrap (CB), LR (LR), and plug-in (PL) methods when $d=0.1$ and $\beta=4$.}
	\label{fig:within-sample-pred-11}
\end{figure}
\begin{figure}[t!]
	\centering
	\includegraphics[width=\textwidth]{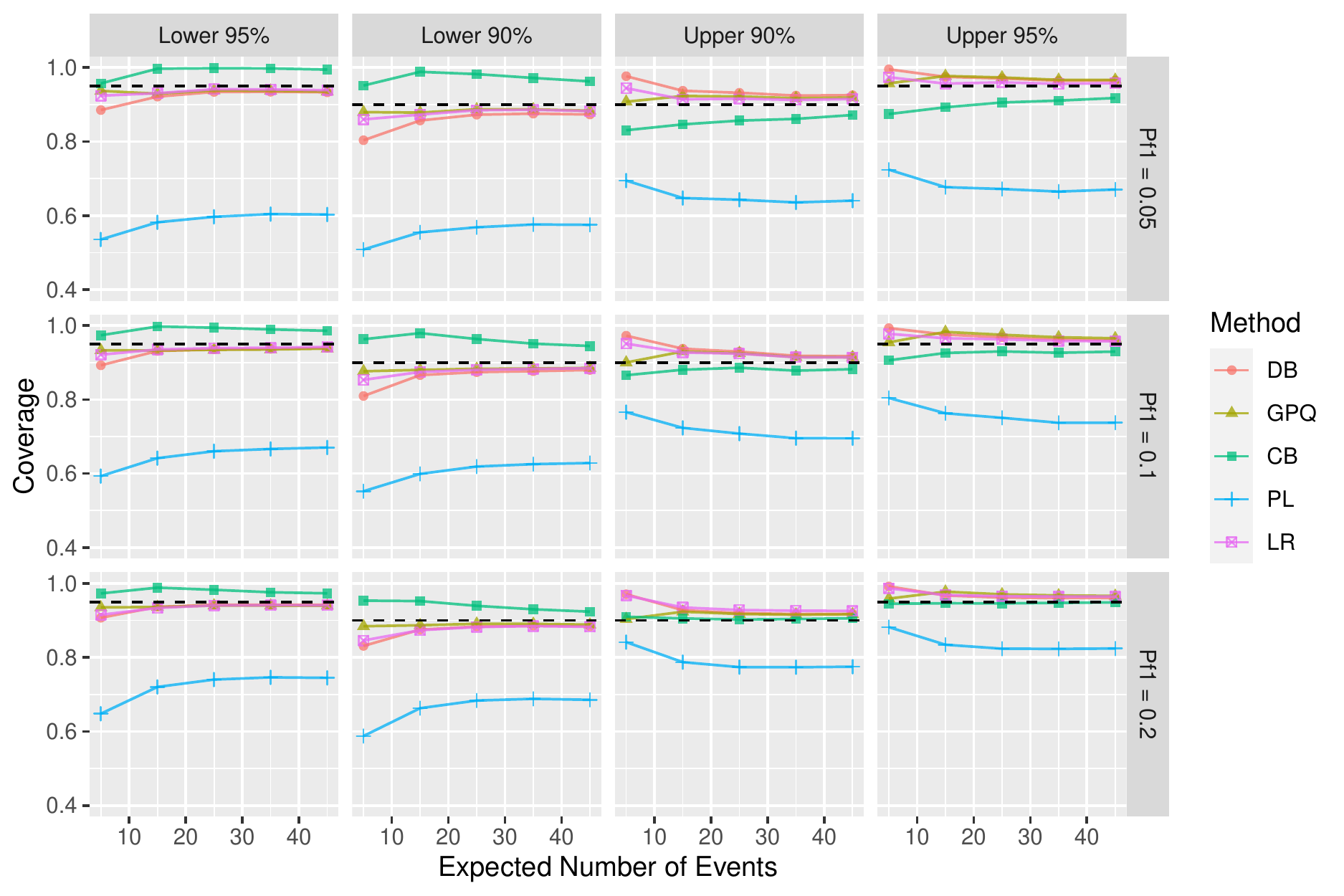}
	\caption{Coverage probabilities versus expected number of events (failures) for the direct-bootstrap (DB), GPQ-bootstrap (GPQ), calibration-bootstrap (CB), LR (LR), and plug-in (PL) methods when $d=0.2$ and $\beta=4$.}
	\label{fig:within-sample-pred-12}
\end{figure}
\clearpage
\bibliographystyle{apalike}
\bibliography{reference}